\pdfoutput=1
\documentclass[lettersize,journal,onecolumn,11pt]{IEEEtran}
\usepackage[letterpaper,margin=1in]{geometry}
\usepackage{amsmath,amssymb,amsfonts}
\usepackage{algorithmic}
\usepackage{algorithm}
\usepackage{array}
\usepackage[caption=false,font=footnotesize,labelfont=sf,textfont=sf]{subfig}
\usepackage{textcomp}
\usepackage{stfloats}          %
\usepackage{url}
\usepackage{graphicx}
\usepackage{cite}
\usepackage{booktabs}          %
\usepackage{multirow}
\usepackage{xcolor}
\usepackage{tikz}              %
\usepackage{pifont}            %
\newcommand{\cmark}{\ding{51}}
\newcommand{\xmark}{\ding{55}}
\newcommand{\pmark}{$\circ$}   %

\newtheorem{definition}{Definition}
\newtheorem{proposition}{Proposition}

\usepackage[hidelinks]{hyperref}
\usepackage[capitalise,nameinlink,noabbrev]{cleveref}
\crefname{section}{Section}{Sections}
\Crefname{section}{Section}{Sections}
\crefname{subsection}{Section}{Sections}
\Crefname{subsection}{Section}{Sections}
\crefname{figure}{Fig.}{Figs.}
\Crefname{figure}{Fig.}{Figs.}
\crefname{table}{Table}{Tables}
\Crefname{table}{Table}{Tables}
\crefname{equation}{}{}
\Crefname{equation}{Equation}{Equations}
\creflabelformat{equation}{(#2#1#3)}
\crefname{algorithm}{Algorithm}{Algorithms}
\crefname{definition}{Definition}{Definitions}
\crefname{proposition}{Proposition}{Propositions}
\crefname{lemma}{Lemma}{Lemmas}
\crefname{enumi}{C}{C}         %

\newcommand{\safedse}{SafeDSE}                 %
\newcommand{\desyde}{DeSyDe}
\newcommand{\mzn}{MiniZinc}
\newcommand{\cpsat}{CP-SAT}
\newcommand{\code}[1]{\texttt{#1}}              %
\newcommand{\pat}{\mathit{pat}}
\newcommand{\act}{\mathit{active}}
\newcommand{\pot}{\mathit{pot}}
\newcommand{\tok}{\mathit{tok}}
\newcommand{\Gplus}{G^{+}}
\newcommand{\proc}{\mathit{proc}}              %
\newcommand{\sil}{\mathit{sil}}
\newcommand{\csil}{\mathit{csil}}
\newcommand{\partn}{\mathit{part}}             %
\newcommand{\fcr}{\mathit{fcr}}
\newcommand{\ctype}{\mathit{ctype}}
\newcommand{\silreq}{\mathit{SIL}^{\mathrm{req}}}
\newcommand{\silmax}{\mathit{SIL}^{\max}}

\newcommand{\resBaselineInstances}{9}

\newcommand{\resBaselineNPInfeasible}{0}
\newcommand{\resBaselineNPMaxLoss}{0.0}
\newcommand{\resBaselineNPWorse}{0}

\newcommand{\resBaselinePFCheapInfeasible}{0}
\newcommand{\resBaselinePFCheapMaxLoss}{0.0}
\newcommand{\resBaselinePFCheapWorse}{0}

\newcommand{\resBaselinePFLightInfeasible}{0}
\newcommand{\resBaselinePFLightMaxLoss}{30.8}

\newcommand{\resBaselinePFLightWorse}{2}

\newcommand{\resBaselinePFMaxInfeasible}{0}
\newcommand{\resBaselinePFMaxMaxLoss}{8.8}

\newcommand{\resBaselinePFMaxWorse}{5}

\newcommand{\resBaselinePFMinInfeasible}{0}
\newcommand{\resBaselinePFMinMaxLoss}{0.0}
\newcommand{\resBaselinePFMinWorse}{0}

\newcommand{\resBaselineSLInfeasible}{0}
\newcommand{\resBaselineSLMaxLoss}{0.0}
\newcommand{\resBaselineSLWorse}{0}
\newcommand{\resCertCertified}{136}

\newcommand{\resCertThrAbove}{45}
\newcommand{\resCertThrAboveExact}{45}
\newcommand{\resCertTwoOfTwoDiverse}{365}
\newcommand{\resCertTwoOfTwoRasta}{340}
\newcommand{\resCertTwoOfTwoSobel}{512}
\newcommand{\resCertUnsatRows}{149}
\newcommand{\resCommCPatIdealCost}{3349}

\newcommand{\resCommCPatIdealMu}{235}

\newcommand{\resCommCPatIdealNodes}{26}

\newcommand{\resCommCPatScopeAllCost}{3349}

\newcommand{\resCommCPatScopeAllMu}{279}
\newcommand{\resCommCPatScopeAllMuPct}{19}
\newcommand{\resCommCPatScopeAllNodes}{392}

\newcommand{\resCommCPatSilCoreCost}{3349}

\newcommand{\resCommCPatSilCoreMu}{235}

\newcommand{\resCommCPatSilCoreNodes}{71}

\newcommand{\resCommCPatSilInheritCost}{5246}
\newcommand{\resCommCPatSilInheritCostPct}{57}
\newcommand{\resCommCPatSilInheritMu}{277}

\newcommand{\resCommCPatSilInheritNodes}{71}

\newcommand{\resCommCPatTdmaCost}{3349}

\newcommand{\resCommCPatTdmaMu}{235}

\newcommand{\resCommCPatTdmaNodes}{71}

\newcommand{\resCommPcSobelIdealCost}{1401}

\newcommand{\resCommPcSobelIdealMu}{333}

\newcommand{\resCommPcSobelIdealNodes}{15}

\newcommand{\resCommPcSobelScopeAllCost}{1401}

\newcommand{\resCommPcSobelScopeAllMu}{369}
\newcommand{\resCommPcSobelScopeAllMuPct}{11}
\newcommand{\resCommPcSobelScopeAllNodes}{270}

\newcommand{\resCommPcSobelSilCoreCost}{1401}

\newcommand{\resCommPcSobelSilCoreMu}{333}

\newcommand{\resCommPcSobelSilCoreNodes}{57}

\newcommand{\resCommPcSobelSilInheritCost}{2608}
\newcommand{\resCommPcSobelSilInheritCostPct}{86}
\newcommand{\resCommPcSobelSilInheritMu}{333}

\newcommand{\resCommPcSobelSilInheritNodes}{57}

\newcommand{\resCommPcSobelTdmaCost}{1401}

\newcommand{\resCommPcSobelTdmaMu}{333}

\newcommand{\resCommPcSobelTdmaNodes}{57}

\newcommand{\resCommSobelBound}{400}
\newcommand{\resCommSobelMu}{333}

\newcommand{\resConsolidationOnePromotion}{1913}

\newcommand{\resConsolidationThreePromotion}{0}

\newcommand{\resConsolidationTwoPromotion}{293}

\newcommand{\resDsdComparable}{7}
\newcommand{\resDsdDesydeProven}{2}

\newcommand{\resDsdFourRaJpDesydePeriod}{1514}
\newcommand{\resDsdFourRaJpDesydeProven}{no}
\newcommand{\resDsdFourRaJpDesydeTime}{3600}
\newcommand{\resDsdFourRaJpPubPeriod}{1388}
\newcommand{\resDsdFourRaJpPubTime}{279}
\newcommand{\resDsdFourRaJpSafePeriod}{1388}
\newcommand{\resDsdFourRaJpSafeTime}{18}
\newcommand{\resDsdFourSoRaJpDesydePeriod}{1766}
\newcommand{\resDsdFourSoRaJpDesydeProven}{no}
\newcommand{\resDsdFourSoRaJpDesydeTime}{3600}
\newcommand{\resDsdFourSoRaJpPubPeriod}{1766}
\newcommand{\resDsdFourSoRaJpPubTime}{23705}
\newcommand{\resDsdFourSoRaJpSafePeriod}{1766}
\newcommand{\resDsdFourSoRaJpSafeTime}{67}
\newcommand{\resDsdMaxReps}{5}

\newcommand{\resDsdSafeProven}{6}
\newcommand{\resDsdSafeTimeLimit}{3600}
\newcommand{\resDsdScenarios}{11}
\newcommand{\resDsdSkipped}{4}
\newcommand{\resDsdThreads}{4}
\newcommand{\resDsdThreeRaJpDesydePeriod}{2524}
\newcommand{\resDsdThreeRaJpDesydeProven}{yes}
\newcommand{\resDsdThreeRaJpDesydeTime}{4}
\newcommand{\resDsdThreeRaJpPubPeriod}{2524}
\newcommand{\resDsdThreeRaJpPubTime}{10}
\newcommand{\resDsdThreeRaJpSafePeriod}{2524}
\newcommand{\resDsdThreeRaJpSafeTime}{18}
\newcommand{\resDsdThreeSoRaJpDesydePeriod}{2544}
\newcommand{\resDsdThreeSoRaJpDesydeProven}{no}
\newcommand{\resDsdThreeSoRaJpDesydeTime}{3600}
\newcommand{\resDsdThreeSoRaJpPubPeriod}{2544}
\newcommand{\resDsdThreeSoRaJpPubTime}{3124}
\newcommand{\resDsdThreeSoRaJpSafePeriod}{2544}
\newcommand{\resDsdThreeSoRaJpSafeTime}{3602}
\newcommand{\resDsdThreeSoSuJpDesydePeriod}{2544}
\newcommand{\resDsdThreeSoSuJpDesydeProven}{no}
\newcommand{\resDsdThreeSoSuJpDesydeTime}{3600}
\newcommand{\resDsdThreeSoSuJpPubPeriod}{2544}
\newcommand{\resDsdThreeSoSuJpPubTime}{2662}
\newcommand{\resDsdThreeSoSuJpSafePeriod}{2544}
\newcommand{\resDsdThreeSoSuJpSafeTime}{2829}
\newcommand{\resDsdThreeSoSuRaJpDesydePeriod}{3882}
\newcommand{\resDsdThreeSoSuRaJpDesydeProven}{yes}
\newcommand{\resDsdThreeSoSuRaJpDesydeTime}{2885}
\newcommand{\resDsdThreeSoSuRaJpPubPeriod}{3882}
\newcommand{\resDsdThreeSoSuRaJpPubTime}{1906}
\newcommand{\resDsdThreeSoSuRaJpSafePeriod}{3882}
\newcommand{\resDsdThreeSoSuRaJpSafeTime}{31}
\newcommand{\resDsdThreeSuRaJpDesydePeriod}{2544}
\newcommand{\resDsdThreeSuRaJpDesydeProven}{no}
\newcommand{\resDsdThreeSuRaJpDesydeTime}{3600}
\newcommand{\resDsdThreeSuRaJpPubPeriod}{2544}
\newcommand{\resDsdThreeSuRaJpPubTime}{2562}
\newcommand{\resDsdThreeSuRaJpSafePeriod}{2544}
\newcommand{\resDsdThreeSuRaJpSafeTime}{3312}
\newcommand{\resExactAgree}{1577}
\newcommand{\resExactConclusive}{1578}

\newcommand{\resExactInconclusive}{79}

\newcommand{\resExactRows}{1657}

\newcommand{\resExactSLUnsatRows}{6}
\newcommand{\resExactTimeRatioMedian}{1.09}

\newcommand{\resExactUnsatRows}{155}

\newcommand{\resFlatEightMB}{2}
\newcommand{\resFlatEightRssMB}{44}
\newcommand{\resFlatEightTime}{0.8}

\newcommand{\resFlatFortyEightMB}{116}
\newcommand{\resFlatFortyEightRssMB}{736}
\newcommand{\resFlatFortyEightTime}{61.9}

\newcommand{\resFlatSixteenMB}{8}
\newcommand{\resFlatSixteenRssMB}{86}
\newcommand{\resFlatSixteenTime}{2.9}

\newcommand{\resFlatSixtyFourMB}{252}
\newcommand{\resFlatSixtyFourRssMB}{1588}
\newcommand{\resFlatSixtyFourTime}{139.0}

\newcommand{\resFlatSixtyFourVars}{259}
\newcommand{\resFlatTableTimeMax}{2.4}
\newcommand{\resFlatTableVarsMax}{22}
\newcommand{\resFlatTableVarsMin}{5}
\newcommand{\resFlatThirtyTwoMB}{36}
\newcommand{\resFlatThirtyTwoRssMB}{305}
\newcommand{\resFlatThirtyTwoTime}{15.8}

\newcommand{\resFlatTwentyFourMB}{20}
\newcommand{\resFlatTwentyFourRssMB}{173}
\newcommand{\resFlatTwentyFourTime}{8.3}

\newcommand{\resGsnBothChecks}{51}
\newcommand{\resGsnBothDiverse}{3}
\newcommand{\resGsnBothElements}{151}
\newcommand{\resGsnBothGoals}{54}
\newcommand{\resGsnBothLeaves}{39}
\newcommand{\resGsnBothOutOfScope}{0}
\newcommand{\resGsnBothSolutionShare}{49}
\newcommand{\resGsnBothSolutions}{19}
\newcommand{\resGsnBothUndeveloped}{20}
\newcommand{\resGsnNvpChecks}{71}
\newcommand{\resGsnNvpDiverse}{7}
\newcommand{\resGsnNvpElements}{165}
\newcommand{\resGsnNvpGoals}{54}
\newcommand{\resGsnNvpLeaves}{53}
\newcommand{\resGsnNvpOutOfScope}{0}
\newcommand{\resGsnNvpSolutionShare}{62}
\newcommand{\resGsnNvpSolutions}{33}
\newcommand{\resGsnNvpUndeveloped}{20}
\newcommand{\resGsnRandomHwChecks}{42}
\newcommand{\resGsnRandomHwDiverse}{0}
\newcommand{\resGsnRandomHwElements}{117}
\newcommand{\resGsnRandomHwGoals}{41}
\newcommand{\resGsnRandomHwLeaves}{31}
\newcommand{\resGsnRandomHwOutOfScope}{0}
\newcommand{\resGsnRandomHwSolutionShare}{52}
\newcommand{\resGsnRandomHwSolutions}{16}
\newcommand{\resGsnRandomHwUndeveloped}{15}
\newcommand{\resModExtBothCheaper}{38}
\newcommand{\resModExtBothDearer}{4}

\newcommand{\resModExtBothJointRows}{54}
\newcommand{\resModExtBothMaxExtraPct}{9.8}
\newcommand{\resModExtBothMaxSavingPct}{18.5}
\newcommand{\resModExtBothMinExtraPct}{5.8}
\newcommand{\resModExtBothMinSavingPct}{3.1}

\newcommand{\resModExtBothPFCheapHundredFiftyInfeasible}{0}

\newcommand{\resModExtBothPFCheapHundredInfeasible}{0}

\newcommand{\resModExtBothPFCheapHundredTenInfeasible}{0}

\newcommand{\resModExtBothPFCheapHundredTwentyFiveInfeasible}{0}

\newcommand{\resModExtBothPFCheapMaxLossPct}{0.0}

\newcommand{\resModExtBothPFCheapTwoHundredInfeasible}{0}

\newcommand{\resModExtBothPFLightHundredFiftyInfeasible}{0}

\newcommand{\resModExtBothPFLightHundredInfeasible}{0}

\newcommand{\resModExtBothPFLightHundredTenInfeasible}{0}

\newcommand{\resModExtBothPFLightHundredTwentyFiveInfeasible}{0}

\newcommand{\resModExtBothPFLightMaxLossPct}{0.0}

\newcommand{\resModExtBothPFLightTwoHundredInfeasible}{0}

\newcommand{\resModExtBothPFMaxHundredFiftyInfeasible}{3}

\newcommand{\resModExtBothPFMaxHundredInfeasible}{6}

\newcommand{\resModExtBothPFMaxHundredTenInfeasible}{6}

\newcommand{\resModExtBothPFMaxHundredTwentyFiveInfeasible}{6}

\newcommand{\resModExtBothPFMaxMaxLossPct}{46.3}

\newcommand{\resModExtBothPFMaxTwoHundredInfeasible}{0}

\newcommand{\resModExtBothPFMinHundredFiftyInfeasible}{3}

\newcommand{\resModExtBothPFMinHundredInfeasible}{6}

\newcommand{\resModExtBothPFMinHundredTenInfeasible}{6}

\newcommand{\resModExtBothPFMinHundredTwentyFiveInfeasible}{6}

\newcommand{\resModExtBothPFMinMaxLossPct}{23.2}

\newcommand{\resModExtBothPFMinTwoHundredInfeasible}{0}

\newcommand{\resModExtBothPatternsChanged}{11}

\newcommand{\resModExtBothSLHundredFiftyInfeasible}{0}

\newcommand{\resModExtBothSLHundredInfeasible}{2}

\newcommand{\resModExtBothSLHundredTenInfeasible}{0}

\newcommand{\resModExtBothSLHundredTwentyFiveInfeasible}{0}

\newcommand{\resModExtBothSLMaxLossPct}{5.9}

\newcommand{\resModExtBothSLTwoHundredInfeasible}{0}

\newcommand{\resModExtBothTimeRatioMed}{1.26}

\newcommand{\resModExtDiverseCheaper}{6}
\newcommand{\resModExtDiverseDearer}{9}

\newcommand{\resModExtDiverseJointRows}{54}
\newcommand{\resModExtDiverseMaxExtraPct}{37.0}
\newcommand{\resModExtDiverseMaxSavingPct}{7.0}
\newcommand{\resModExtDiverseMinExtraPct}{4.8}
\newcommand{\resModExtDiverseMinSavingPct}{6.6}
\newcommand{\resModExtDiversePatternsChanged}{10}

\newcommand{\resModExtDiverseTimeRatioMed}{1.03}

\newcommand{\resModExtDoerCheaper}{48}
\newcommand{\resModExtDoerDearer}{0}

\newcommand{\resModExtDoerInstances}{9}
\newcommand{\resModExtDoerJointRows}{54}
\newcommand{\resModExtDoerMaxSavingPct}{19.6}
\newcommand{\resModExtDoerMinSavingPct}{1.9}
\newcommand{\resModExtDoerPatternsChanged}{15}

\newcommand{\resModExtDoerTimeRatioMed}{1.29}

\newcommand{\resModExtTokOneCheaper}{0}
\newcommand{\resModExtTokOneDearer}{0}

\newcommand{\resModExtTokOneHalfCheaper}{2}
\newcommand{\resModExtTokOneHalfDearer}{0}

\newcommand{\resModExtTokOneHalfJointRows}{45}
\newcommand{\resModExtTokOneHalfMaxSavingPct}{11.5}
\newcommand{\resModExtTokOneHalfMinSavingPct}{8.9}

\newcommand{\resModExtTokOneHalfPFCheapHundredFiftyInfeasible}{0}

\newcommand{\resModExtTokOneHalfPFCheapHundredInfeasible}{2}

\newcommand{\resModExtTokOneHalfPFCheapHundredTenInfeasible}{2}

\newcommand{\resModExtTokOneHalfPFCheapHundredTwentyFiveInfeasible}{2}

\newcommand{\resModExtTokOneHalfPFCheapMaxLossPct}{0.0}
\newcommand{\resModExtTokOneHalfPFCheapTwoHundredInfeasible}{0}

\newcommand{\resModExtTokOneHalfPFLightHundredFiftyInfeasible}{0}

\newcommand{\resModExtTokOneHalfPFLightHundredInfeasible}{0}

\newcommand{\resModExtTokOneHalfPFLightHundredTenInfeasible}{0}

\newcommand{\resModExtTokOneHalfPFLightHundredTwentyFiveInfeasible}{0}

\newcommand{\resModExtTokOneHalfPFLightMaxLossPct}{30.2}
\newcommand{\resModExtTokOneHalfPFLightTwoHundredInfeasible}{0}

\newcommand{\resModExtTokOneHalfPFMaxHundredFiftyInfeasible}{0}

\newcommand{\resModExtTokOneHalfPFMaxHundredInfeasible}{6}

\newcommand{\resModExtTokOneHalfPFMaxHundredTenInfeasible}{5}

\newcommand{\resModExtTokOneHalfPFMaxHundredTwentyFiveInfeasible}{5}

\newcommand{\resModExtTokOneHalfPFMaxMaxLossPct}{18.1}
\newcommand{\resModExtTokOneHalfPFMaxTwoHundredInfeasible}{0}

\newcommand{\resModExtTokOneHalfPFMinHundredFiftyInfeasible}{0}

\newcommand{\resModExtTokOneHalfPFMinHundredInfeasible}{6}

\newcommand{\resModExtTokOneHalfPFMinHundredTenInfeasible}{5}

\newcommand{\resModExtTokOneHalfPFMinHundredTwentyFiveInfeasible}{5}

\newcommand{\resModExtTokOneHalfPFMinMaxLossPct}{0.9}
\newcommand{\resModExtTokOneHalfPFMinTwoHundredInfeasible}{0}

\newcommand{\resModExtTokOneHalfPatternsChanged}{5}

\newcommand{\resModExtTokOneHalfSLHundredFiftyInfeasible}{0}

\newcommand{\resModExtTokOneHalfSLHundredInfeasible}{2}

\newcommand{\resModExtTokOneHalfSLHundredTenInfeasible}{0}

\newcommand{\resModExtTokOneHalfSLHundredTwentyFiveInfeasible}{0}

\newcommand{\resModExtTokOneHalfSLMaxLossPct}{13.1}
\newcommand{\resModExtTokOneHalfSLTwoHundredInfeasible}{0}

\newcommand{\resModExtTokOneHalfTimeRatioMed}{1.13}

\newcommand{\resModExtTokOneJointRows}{45}

\newcommand{\resModExtTokOnePFCheapHundredFiftyInfeasible}{2}

\newcommand{\resModExtTokOnePFCheapHundredInfeasible}{2}

\newcommand{\resModExtTokOnePFCheapHundredTenInfeasible}{2}

\newcommand{\resModExtTokOnePFCheapHundredTwentyFiveInfeasible}{2}

\newcommand{\resModExtTokOnePFCheapMaxLossPct}{0.0}
\newcommand{\resModExtTokOnePFCheapTwoHundredInfeasible}{0}

\newcommand{\resModExtTokOnePFLightHundredFiftyInfeasible}{0}

\newcommand{\resModExtTokOnePFLightHundredInfeasible}{0}

\newcommand{\resModExtTokOnePFLightHundredTenInfeasible}{0}

\newcommand{\resModExtTokOnePFLightHundredTwentyFiveInfeasible}{0}

\newcommand{\resModExtTokOnePFLightMaxLossPct}{30.2}
\newcommand{\resModExtTokOnePFLightTwoHundredInfeasible}{0}

\newcommand{\resModExtTokOnePFMaxHundredFiftyInfeasible}{3}

\newcommand{\resModExtTokOnePFMaxHundredInfeasible}{6}

\newcommand{\resModExtTokOnePFMaxHundredTenInfeasible}{6}

\newcommand{\resModExtTokOnePFMaxHundredTwentyFiveInfeasible}{6}

\newcommand{\resModExtTokOnePFMaxMaxLossPct}{18.1}
\newcommand{\resModExtTokOnePFMaxTwoHundredInfeasible}{0}

\newcommand{\resModExtTokOnePFMinHundredFiftyInfeasible}{3}

\newcommand{\resModExtTokOnePFMinHundredInfeasible}{6}

\newcommand{\resModExtTokOnePFMinHundredTenInfeasible}{6}

\newcommand{\resModExtTokOnePFMinHundredTwentyFiveInfeasible}{6}

\newcommand{\resModExtTokOnePFMinMaxLossPct}{0.9}
\newcommand{\resModExtTokOnePFMinTwoHundredInfeasible}{0}

\newcommand{\resModExtTokOnePatternsChanged}{3}

\newcommand{\resModExtTokOneSLHundredFiftyInfeasible}{0}

\newcommand{\resModExtTokOneSLHundredInfeasible}{2}

\newcommand{\resModExtTokOneSLHundredTenInfeasible}{0}

\newcommand{\resModExtTokOneSLHundredTwentyFiveInfeasible}{0}

\newcommand{\resModExtTokOneSLMaxLossPct}{4.6}
\newcommand{\resModExtTokOneSLTwoHundredInfeasible}{0}

\newcommand{\resModExtTokOneTimeRatioMed}{1.02}

\newcommand{\resModExtTokTwoCheaper}{28}
\newcommand{\resModExtTokTwoDearer}{0}

\newcommand{\resModExtTokTwoJointRows}{45}
\newcommand{\resModExtTokTwoMaxSavingPct}{22.2}
\newcommand{\resModExtTokTwoMinSavingPct}{0.1}
\newcommand{\resModExtTokTwoPatternsChanged}{15}

\newcommand{\resModExtTokTwoTimeRatioMed}{1.33}

\newcommand{\resModExtTokTwoVariantCompared}{40}

\newcommand{\resModExtTokTwoVariantTimeRatioMed}{1.01}

\newcommand{\resMultiCost}{5665}
\newcommand{\resMultiExactAgree}{21}

\newcommand{\resMultiExactOpen}{14}

\newcommand{\resMultiHundredFiftyJointCost}{5852}

\newcommand{\resMultiHundredFiftyNPLoss}{0.0}

\newcommand{\resMultiHundredFiftyPFCheapLoss}{0.0}

\newcommand{\resMultiHundredFiftyPFLightLoss}{0.0}

\newcommand{\resMultiHundredFiftyPFMaxLoss}{0.0}

\newcommand{\resMultiHundredFiftyPFMinLoss}{0.0}

\newcommand{\resMultiHundredTenJointExactLong}{UNSAT}
\newcommand{\resMultiHundredTenJointExactLongTime}{994.7}

\newcommand{\resMultiHundredTwentyFiveJointCost}{5866}

\newcommand{\resMultiHundredTwentyFiveNPLoss}{0.0}

\newcommand{\resMultiHundredTwentyFivePFCheapLoss}{0.0}

\newcommand{\resMultiHundredTwentyFivePFLightLoss}{0.0}

\newcommand{\resMultiMuAloneRasta}{235}
\newcommand{\resMultiMuAloneSobel}{416}
\newcommand{\resMultiMuAloneSusan}{1177}
\newcommand{\resMultiNoneJointCost}{5665}

\newcommand{\resMultiNoneNPLoss}{0.0}

\newcommand{\resMultiNonePFCheapLoss}{0.0}

\newcommand{\resMultiNonePFLightLoss}{0.0}

\newcommand{\resMultiNonePFMaxLoss}{0.4}

\newcommand{\resMultiNonePFMinLoss}{0.0}

\newcommand{\resMultiNoneSLLoss}{0.4}
\newcommand{\resMultiNprocs}{6}

\newcommand{\resMultiTime}{44.0}
\newcommand{\resMultiTimeMax}{63.3}
\newcommand{\resMultiTimeMin}{37.5}
\newcommand{\resMultiTwoHundredJointCost}{5768}

\newcommand{\resMultiTwoHundredNPLoss}{0.0}

\newcommand{\resMultiTwoHundredPFCheapLoss}{0.0}

\newcommand{\resMultiTwoHundredPFLightLoss}{0.0}

\newcommand{\resMultiTwoHundredPFMaxLoss}{0.0}

\newcommand{\resMultiTwoHundredPFMinLoss}{0.0}

\newcommand{\resMutArgBadIndex}{50}

\newcommand{\resMutArgCaughtAudit}{250}
\newcommand{\resMutArgCaughtCur}{300}

\newcommand{\resMutArgMutants}{300}
\newcommand{\resMutArgOperators}{6}

\newcommand{\resMutGenCatalogCases}{17}
\newcommand{\resMutGenCatalogRefused}{17}

\newcommand{\resMutGenControls}{10}

\newcommand{\resMutGenFaultModelAmbiguous}{3}
\newcommand{\resMutGenFaultModelCases}{18}
\newcommand{\resMutGenFaultModelRefused}{18}

\newcommand{\resMutGsnEmittedFaultyArt}{230}

\newcommand{\resMutModelFamiliesManifested}{14}

\newcommand{\resMutModelManifested}{108}

\newcommand{\resMutModelModelOnly}{3}
\newcommand{\resMutModelMutants}{18}

\newcommand{\resMutModelRejectedArt}{79}
\newcommand{\resMutModelRejectedCur}{108}

\newcommand{\resMutModelRuns}{163}
\newcommand{\resMutOpCeilingManifestedArt}{21}

\newcommand{\resMutOpCeilingRejectedArt}{21}
\newcommand{\resMutOpCeilingRejectedCur}{21}

\newcommand{\resMutOpCostManifestedArt}{10}

\newcommand{\resMutOpCostRejectedArt}{0}
\newcommand{\resMutOpCostRejectedCur}{10}

\newcommand{\resMutOpDeactivateManifestedArt}{32}

\newcommand{\resMutOpDeactivateRejectedArt}{0}
\newcommand{\resMutOpDeactivateRejectedCur}{32}

\newcommand{\resMutOpDropFieldManifestedArt}{20}

\newcommand{\resMutOpDropFieldRejectedArt}{10}
\newcommand{\resMutOpDropFieldRejectedCur}{20}

\newcommand{\resMutOpMixManifestedArt}{23}

\newcommand{\resMutOpMixRejectedArt}{23}
\newcommand{\resMutOpMixRejectedCur}{23}

\newcommand{\resMutOpOrderCrossManifestedArt}{19}

\newcommand{\resMutOpOrderCrossRejectedArt}{0}
\newcommand{\resMutOpOrderCrossRejectedCur}{19}

\newcommand{\resMutOpOrderCycleManifestedArt}{30}

\newcommand{\resMutOpOrderCycleRejectedArt}{2}
\newcommand{\resMutOpOrderCycleRejectedCur}{30}

\newcommand{\resMutOpOrderSplitManifestedArt}{30}

\newcommand{\resMutOpOrderSplitRejectedArt}{2}
\newcommand{\resMutOpOrderSplitRejectedCur}{30}

\newcommand{\resMutOpPartForgeManifestedArt}{23}

\newcommand{\resMutOpPartForgeRejectedArt}{0}
\newcommand{\resMutOpPartForgeRejectedCur}{23}

\newcommand{\resMutOpPatternNoneManifestedArt}{30}

\newcommand{\resMutOpPatternNoneRejectedArt}{0}
\newcommand{\resMutOpPatternNoneRejectedCur}{30}

\newcommand{\resMutOpPeriodManifestedArt}{12}

\newcommand{\resMutOpPeriodRejectedArt}{12}
\newcommand{\resMutOpPeriodRejectedCur}{12}

\newcommand{\resMutOpPlaceManifestedArt}{32}

\newcommand{\resMutOpPlaceRejectedArt}{32}
\newcommand{\resMutOpPlaceRejectedCur}{32}

\newcommand{\resMutOpSilCoreManifestedArt}{39}

\newcommand{\resMutOpSilCoreRejectedArt}{39}
\newcommand{\resMutOpSilCoreRejectedCur}{39}

\newcommand{\resMutOpSilNodeManifestedArt}{50}

\newcommand{\resMutOpSilNodeRejectedArt}{50}
\newcommand{\resMutOpSilNodeRejectedCur}{50}

\newcommand{\resMutOpWcetManifestedArt}{50}

\newcommand{\resMutOpWcetRejectedArt}{0}
\newcommand{\resMutOpWcetRejectedCur}{50}

\newcommand{\resMutRecheckSolutions}{4113}

\newcommand{\resMutSolInstances}{10}

\newcommand{\resMutSolManifested}{421}

\newcommand{\resMutSolMutants}{421}
\newcommand{\resMutSolOperators}{15}
\newcommand{\resMutSolRejectedArt}{191}
\newcommand{\resMutSolRejectedCur}{421}

\newcommand{\resNvpNprocs}{6}

\newcommand{\resNvpTime}{5.5}
\newcommand{\resOneStepRastaBound}{550}
\newcommand{\resOneStepRastaHwCost}{18}
\newcommand{\resOneStepRastaMuMax}{543}
\newcommand{\resOneStepRastaTime}{0.6}
\newcommand{\resOneStepSobelBound}{400}
\newcommand{\resOneStepSobelHwCost}{16}

\newcommand{\resOneStepSobelMuMax}{384}
\newcommand{\resOneStepSobelTime}{0.5}
\newcommand{\resOneStepSusanBound}{2050}
\newcommand{\resOneStepSusanHwCost}{12}

\newcommand{\resOneStepSusanMuMax}{1659}
\newcommand{\resOneStepSusanTime}{0.5}

\newcommand{\resPatConsNPInfeasiblePairs}{19}

\newcommand{\resPatConsPromoMax}{4550}
\newcommand{\resPatConsPromoMin}{293}
\newcommand{\resPatternTieInstances}{3}

\newcommand{\resPipePinMax}{1.2}

\newcommand{\resPlacementRandomHwNprocs}{3}

\newcommand{\resPlacementSystematicSwNprocs}{2}

\newcommand{\resPriceDoOneSeventyEightBOneOneExtra}{59}
\newcommand{\resPriceDoOneSeventyEightBOneOnePromotion}{1400}

\newcommand{\resPriceDoOneSeventyEightBSixteenOneExtra}{14}

\newcommand{\resPriceKlostermanConsolidateFactor}{8}

\newcommand{\resPriceKlostermanFourNprocs}{3}

\newcommand{\resPriceKlostermanLowConsolidateFactor}{4}

\newcommand{\resPriceKlostermanLowFourNprocs}{1}

\newcommand{\resPriceKlostermanLowOneOnePromotion}{125}

\newcommand{\resPriceKlostermanOneOnePromotion}{275}

\newcommand{\resPriceMaxFactor}{16}
\newcommand{\resPriceMinFactor}{1}

\newcommand{\resPriceMyklebustOneOneExtra}{98}
\newcommand{\resPriceMyklebustOneOnePromotion}{1913}

\newcommand{\resPriceMyklebustSixteenNprocs}{3}
\newcommand{\resPriceMyklebustSixteenOneExtra}{26}

\newcommand{\resRepsMaxSpread}{86}
\newcommand{\resRepsMedianSpread}{9}
\newcommand{\resRepsPerRow}{5}
\newcommand{\resRepsRuns}{270}
\newcommand{\resRosvallAllOnestepStatus}{OPTIMAL}
\newcommand{\resRosvallAllOnestepTime}{7.3}
\newcommand{\resRosvallAllTwoStepHwCost}{54}
\newcommand{\resRosvallAllTwoStepNprocs}{6}
\newcommand{\resRosvallAllTwoStepStepOneTime}{3.6}
\newcommand{\resRosvallAllTwoStepStepTwoTime}{5.5}
\newcommand{\resRosvallAllTwoStepTime}{9.1}

\newcommand{\resScalCatFiveProven}{5}

\newcommand{\resScalCatFiveTime}{18.2}

\newcommand{\resScalCatNineFlatVars}{38572}
\newcommand{\resScalCatNineProven}{4}

\newcommand{\resScalCatNineTime}{76.5}

\newcommand{\resScalCatOneBuildFail}{5}
\newcommand{\resScalCatThreeFlatVars}{14786}
\newcommand{\resScalCatThreeProven}{5}

\newcommand{\resScalCatThreeTime}{13.3}

\newcommand{\resScalEightMuFound}{5}
\newcommand{\resScalEightMuProven}{5}
\newcommand{\resScalEightNoMu}{0}
\newcommand{\resScalEightProven}{5}
\newcommand{\resScalEightReps}{5}
\newcommand{\resScalEightTime}{1.8}
\newcommand{\resScalEightTimeout}{0}
\newcommand{\resScalFortyEightMuFound}{0}
\newcommand{\resScalFortyEightMuProven}{0}
\newcommand{\resScalFortyEightNoMu}{5}
\newcommand{\resScalGapCount}{1}
\newcommand{\resScalGapMax}{12.2}

\newcommand{\resScalMaxProvenActors}{16}

\newcommand{\resScalSixteenMuFound}{5}
\newcommand{\resScalSixteenMuProven}{4}
\newcommand{\resScalSixteenNoMu}{0}
\newcommand{\resScalSixteenProven}{5}
\newcommand{\resScalSixteenReps}{5}
\newcommand{\resScalSixteenTime}{56.1}
\newcommand{\resScalSixteenTimeout}{0}
\newcommand{\resScalSixtyFourMuFound}{0}
\newcommand{\resScalSixtyFourMuProven}{0}
\newcommand{\resScalSixtyFourNoMu}{5}

\newcommand{\resScalThirtyTwoMuFound}{4}
\newcommand{\resScalThirtyTwoMuProven}{0}
\newcommand{\resScalThirtyTwoNoMu}{1}
\newcommand{\resScalThirtyTwoProven}{1}
\newcommand{\resScalThirtyTwoReps}{4}
\newcommand{\resScalThirtyTwoTime}{273.7}
\newcommand{\resScalThirtyTwoTimeout}{3}

\newcommand{\resScalTwentyFourMuFound}{5}
\newcommand{\resScalTwentyFourMuProven}{2}
\newcommand{\resScalTwentyFourNoMu}{0}
\newcommand{\resScalTwentyFourProven}{2}
\newcommand{\resScalTwentyFourReps}{5}
\newcommand{\resScalTwentyFourTime}{84.9}
\newcommand{\resScalTwentyFourTimeout}{3}
\newcommand{\resSetupCheckerScale}{30}
\newcommand{\resSetupCpSat}{9.15}
\newcommand{\resSetupCpu}{Intel Core i7-7700HQ, 2.80\,GHz}
\newcommand{\resSetupLogicalCpus}{8}
\newcommand{\resSetupMiniZinc}{2.10.1}
\newcommand{\resSetupRamGiB}{15.5}

\newcommand{\resSetupScalCatActors}{24}
\newcommand{\resSetupScalMaxActors}{64}
\newcommand{\resSetupScalMinActors}{8}
\newcommand{\resSetupScalMuTimeLimit}{300}
\newcommand{\resSetupScalReps}{5}
\newcommand{\resSetupScalSlots}{8}
\newcommand{\resSetupScalTimeLimit}{600}
\newcommand{\resSetupScalTypes}{3}
\newcommand{\resSetupThreads}{4}
\newcommand{\resSetupTimeLimit}{300}
\newcommand{\resThrBaseInstances}{9}
\newcommand{\resThrBaseMuStarMax}{333}
\newcommand{\resThrBaseMuStarMin}{235}
\newcommand{\resThrBaseNPAbove}{0}

\newcommand{\resThrBaseNPMaxRatio}{1.00}
\newcommand{\resThrBasePFCheapAbove}{2}

\newcommand{\resThrBasePFCheapMaxRatio}{1.55}

\newcommand{\resThrBasePFLightAbove}{0}

\newcommand{\resThrBasePFLightMaxRatio}{1.00}
\newcommand{\resThrBasePFMaxAbove}{6}

\newcommand{\resThrBasePFMaxMaxRatio}{1.55}

\newcommand{\resThrBasePFMinAbove}{6}

\newcommand{\resThrBasePFMinMaxRatio}{1.55}
\newcommand{\resThrBasePFMinMinAboveRatio}{1.45}
\newcommand{\resThrBaseSLAbove}{2}

\newcommand{\resThrBaseSLMaxRatio}{1.09}

\newcommand{\resThrChkFiftyMuStarMax}{384}
\newcommand{\resThrChkFiftyMuStarMin}{235}
\newcommand{\resThrChkFiftyNPAbove}{0}

\newcommand{\resThrChkFiftyNPMaxRatio}{1.00}
\newcommand{\resThrChkFiftyPFCheapAbove}{2}

\newcommand{\resThrChkFiftyPFCheapMaxRatio}{1.43}

\newcommand{\resThrChkFiftyPFLightAbove}{2}

\newcommand{\resThrChkFiftyPFLightMaxRatio}{1.43}

\newcommand{\resThrChkFiftyPFMaxAbove}{7}

\newcommand{\resThrChkFiftyPFMaxMaxRatio}{1.43}

\newcommand{\resThrChkFiftyPFMinAbove}{6}

\newcommand{\resThrChkFiftyPFMinMaxRatio}{1.43}

\newcommand{\resThrChkFiftySLAbove}{2}

\newcommand{\resThrChkFiftySLMaxRatio}{1.13}

\newcommand{\resThrChkTenMuStarMax}{282}
\newcommand{\resThrChkTenMuStarMin}{235}
\newcommand{\resThrChkTenNPAbove}{0}

\newcommand{\resThrChkTenNPMaxRatio}{1.00}
\newcommand{\resThrChkTenPFCheapAbove}{2}

\newcommand{\resThrChkTenPFCheapMaxRatio}{1.55}

\newcommand{\resThrChkTenPFLightAbove}{0}

\newcommand{\resThrChkTenPFLightMaxRatio}{1.00}
\newcommand{\resThrChkTenPFMaxAbove}{6}

\newcommand{\resThrChkTenPFMaxMaxRatio}{1.82}

\newcommand{\resThrChkTenPFMinAbove}{6}

\newcommand{\resThrChkTenPFMinMaxRatio}{1.82}

\newcommand{\resThrChkTenSLAbove}{0}

\newcommand{\resThrChkTenSLMaxRatio}{1.00}

\newcommand{\resThrTokTwoMuStarMax}{256}
\newcommand{\resThrTokTwoMuStarMin}{235}
\newcommand{\resThrTokTwoNPAbove}{0}

\newcommand{\resThrTokTwoNPMaxRatio}{1.00}
\newcommand{\resThrTokTwoPFCheapAbove}{0}

\newcommand{\resThrTokTwoPFCheapMaxRatio}{1.00}
\newcommand{\resThrTokTwoPFLightAbove}{0}

\newcommand{\resThrTokTwoPFLightMaxRatio}{1.00}
\newcommand{\resThrTokTwoPFMaxAbove}{0}

\newcommand{\resThrTokTwoPFMaxMaxRatio}{1.00}
\newcommand{\resThrTokTwoPFMinAbove}{0}

\newcommand{\resThrTokTwoPFMinMaxRatio}{1.00}
\newcommand{\resThrTokTwoSLAbove}{0}

\newcommand{\resThrTokTwoSLMaxRatio}{1.00}
\newcommand{\resTightHundredFiftyInstances}{9}

\newcommand{\resTightHundredFiftyNPInfeasible}{0}
\newcommand{\resTightHundredFiftyNPMaxLoss}{0.0}
\newcommand{\resTightHundredFiftyNPWorse}{0}

\newcommand{\resTightHundredFiftyPFCheapInfeasible}{2}
\newcommand{\resTightHundredFiftyPFCheapMaxLoss}{0.0}
\newcommand{\resTightHundredFiftyPFCheapWorse}{0}

\newcommand{\resTightHundredFiftyPFLightInfeasible}{0}
\newcommand{\resTightHundredFiftyPFLightMaxLoss}{12.4}

\newcommand{\resTightHundredFiftyPFLightWorse}{2}

\newcommand{\resTightHundredFiftyPFMaxInfeasible}{3}
\newcommand{\resTightHundredFiftyPFMaxMaxLoss}{10.6}

\newcommand{\resTightHundredFiftyPFMaxWorse}{2}

\newcommand{\resTightHundredFiftyPFMinInfeasible}{3}
\newcommand{\resTightHundredFiftyPFMinMaxLoss}{0.1}

\newcommand{\resTightHundredFiftyPFMinWorse}{1}

\newcommand{\resTightHundredFiftySLInfeasible}{0}
\newcommand{\resTightHundredFiftySLMaxLoss}{0.3}

\newcommand{\resTightHundredFiftySLWorse}{1}
\newcommand{\resTightHundredInstances}{9}

\newcommand{\resTightHundredNPInfeasible}{0}
\newcommand{\resTightHundredNPMaxLoss}{0.0}
\newcommand{\resTightHundredNPWorse}{0}

\newcommand{\resTightHundredPFCheapInfeasible}{2}
\newcommand{\resTightHundredPFCheapMaxLoss}{0.0}
\newcommand{\resTightHundredPFCheapWorse}{0}

\newcommand{\resTightHundredPFLightInfeasible}{0}
\newcommand{\resTightHundredPFLightMaxLoss}{0.0}
\newcommand{\resTightHundredPFLightWorse}{0}

\newcommand{\resTightHundredPFMaxInfeasible}{6}
\newcommand{\resTightHundredPFMaxMaxLoss}{9.9}

\newcommand{\resTightHundredPFMaxWorse}{1}

\newcommand{\resTightHundredPFMinInfeasible}{6}
\newcommand{\resTightHundredPFMinMaxLoss}{0.0}
\newcommand{\resTightHundredPFMinWorse}{0}

\newcommand{\resTightHundredSLInfeasible}{2}
\newcommand{\resTightHundredSLMaxLoss}{2.7}

\newcommand{\resTightHundredSLWorse}{3}
\newcommand{\resTightHundredTenInstances}{9}

\newcommand{\resTightHundredTenNPInfeasible}{0}
\newcommand{\resTightHundredTenNPMaxLoss}{0.0}
\newcommand{\resTightHundredTenNPWorse}{0}

\newcommand{\resTightHundredTenPFCheapInfeasible}{2}
\newcommand{\resTightHundredTenPFCheapMaxLoss}{0.0}
\newcommand{\resTightHundredTenPFCheapWorse}{0}

\newcommand{\resTightHundredTenPFLightInfeasible}{0}
\newcommand{\resTightHundredTenPFLightMaxLoss}{0.0}
\newcommand{\resTightHundredTenPFLightWorse}{0}

\newcommand{\resTightHundredTenPFMaxInfeasible}{6}
\newcommand{\resTightHundredTenPFMaxMaxLoss}{9.7}

\newcommand{\resTightHundredTenPFMaxWorse}{1}

\newcommand{\resTightHundredTenPFMinInfeasible}{6}
\newcommand{\resTightHundredTenPFMinMaxLoss}{0.0}
\newcommand{\resTightHundredTenPFMinWorse}{0}

\newcommand{\resTightHundredTenSLInfeasible}{0}
\newcommand{\resTightHundredTenSLMaxLoss}{0.6}

\newcommand{\resTightHundredTenSLWorse}{3}
\newcommand{\resTightHundredTwentyFiveInstances}{9}

\newcommand{\resTightHundredTwentyFiveNPInfeasible}{0}
\newcommand{\resTightHundredTwentyFiveNPMaxLoss}{0.0}
\newcommand{\resTightHundredTwentyFiveNPWorse}{0}

\newcommand{\resTightHundredTwentyFivePFCheapInfeasible}{2}
\newcommand{\resTightHundredTwentyFivePFCheapMaxLoss}{0.0}
\newcommand{\resTightHundredTwentyFivePFCheapWorse}{0}

\newcommand{\resTightHundredTwentyFivePFLightInfeasible}{0}
\newcommand{\resTightHundredTwentyFivePFLightMaxLoss}{0.0}
\newcommand{\resTightHundredTwentyFivePFLightWorse}{0}

\newcommand{\resTightHundredTwentyFivePFMaxInfeasible}{6}
\newcommand{\resTightHundredTwentyFivePFMaxMaxLoss}{18.1}

\newcommand{\resTightHundredTwentyFivePFMaxWorse}{1}

\newcommand{\resTightHundredTwentyFivePFMinInfeasible}{6}
\newcommand{\resTightHundredTwentyFivePFMinMaxLoss}{0.0}
\newcommand{\resTightHundredTwentyFivePFMinWorse}{0}

\newcommand{\resTightHundredTwentyFiveSLInfeasible}{0}
\newcommand{\resTightHundredTwentyFiveSLMaxLoss}{0.9}

\newcommand{\resTightHundredTwentyFiveSLWorse}{4}
\newcommand{\resTightTwoHundredInstances}{9}

\newcommand{\resTightTwoHundredNPInfeasible}{0}
\newcommand{\resTightTwoHundredNPMaxLoss}{0.0}
\newcommand{\resTightTwoHundredNPWorse}{0}

\newcommand{\resTightTwoHundredPFCheapInfeasible}{0}
\newcommand{\resTightTwoHundredPFCheapMaxLoss}{0.0}
\newcommand{\resTightTwoHundredPFCheapWorse}{0}

\newcommand{\resTightTwoHundredPFLightInfeasible}{0}
\newcommand{\resTightTwoHundredPFLightMaxLoss}{30.2}

\newcommand{\resTightTwoHundredPFLightWorse}{2}

\newcommand{\resTightTwoHundredPFMaxInfeasible}{0}
\newcommand{\resTightTwoHundredPFMaxMaxLoss}{9.0}

\newcommand{\resTightTwoHundredPFMaxWorse}{6}

\newcommand{\resTightTwoHundredPFMinInfeasible}{0}
\newcommand{\resTightTwoHundredPFMinMaxLoss}{0.9}

\newcommand{\resTightTwoHundredPFMinWorse}{3}

\newcommand{\resTightTwoHundredSLInfeasible}{0}
\newcommand{\resTightTwoHundredSLMaxLoss}{4.6}

\newcommand{\resTightTwoHundredSLWorse}{3}
\newcommand{\resVarChkFiftyHundredFiftyInstances}{9}

\newcommand{\resVarChkFiftyHundredFiftyPFCheapInfeasible}{0}
\newcommand{\resVarChkFiftyHundredFiftyPFCheapMaxLoss}{0.0}
\newcommand{\resVarChkFiftyHundredFiftyPFCheapWorse}{0}

\newcommand{\resVarChkFiftyHundredFiftyPFLightInfeasible}{0}
\newcommand{\resVarChkFiftyHundredFiftyPFLightMaxLoss}{13.7}

\newcommand{\resVarChkFiftyHundredFiftyPFLightWorse}{2}

\newcommand{\resVarChkFiftyHundredFiftyPFMaxInfeasible}{0}
\newcommand{\resVarChkFiftyHundredFiftyPFMaxMaxLoss}{10.6}

\newcommand{\resVarChkFiftyHundredFiftyPFMaxWorse}{5}

\newcommand{\resVarChkFiftyHundredFiftyPFMinInfeasible}{0}
\newcommand{\resVarChkFiftyHundredFiftyPFMinMaxLoss}{1.3}

\newcommand{\resVarChkFiftyHundredFiftyPFMinWorse}{3}

\newcommand{\resVarChkFiftyHundredFiftySLInfeasible}{0}
\newcommand{\resVarChkFiftyHundredFiftySLMaxLoss}{1.3}

\newcommand{\resVarChkFiftyHundredFiftySLWorse}{2}
\newcommand{\resVarChkFiftyHundredInstances}{9}

\newcommand{\resVarChkFiftyHundredPFCheapInfeasible}{2}
\newcommand{\resVarChkFiftyHundredPFCheapMaxLoss}{0.0}
\newcommand{\resVarChkFiftyHundredPFCheapWorse}{0}

\newcommand{\resVarChkFiftyHundredPFLightInfeasible}{2}
\newcommand{\resVarChkFiftyHundredPFLightMaxLoss}{0.0}
\newcommand{\resVarChkFiftyHundredPFLightWorse}{0}

\newcommand{\resVarChkFiftyHundredPFMaxInfeasible}{7}
\newcommand{\resVarChkFiftyHundredPFMaxMaxLoss}{0.0}
\newcommand{\resVarChkFiftyHundredPFMaxWorse}{0}

\newcommand{\resVarChkFiftyHundredPFMinInfeasible}{6}
\newcommand{\resVarChkFiftyHundredPFMinMaxLoss}{0.0}
\newcommand{\resVarChkFiftyHundredPFMinWorse}{0}

\newcommand{\resVarChkFiftyHundredSLInfeasible}{2}
\newcommand{\resVarChkFiftyHundredSLMaxLoss}{2.7}

\newcommand{\resVarChkFiftyHundredSLWorse}{3}
\newcommand{\resVarChkFiftyHundredTenInstances}{9}

\newcommand{\resVarChkFiftyHundredTenPFCheapInfeasible}{2}
\newcommand{\resVarChkFiftyHundredTenPFCheapMaxLoss}{0.0}
\newcommand{\resVarChkFiftyHundredTenPFCheapWorse}{0}

\newcommand{\resVarChkFiftyHundredTenPFLightInfeasible}{2}
\newcommand{\resVarChkFiftyHundredTenPFLightMaxLoss}{0.0}
\newcommand{\resVarChkFiftyHundredTenPFLightWorse}{0}

\newcommand{\resVarChkFiftyHundredTenPFMaxInfeasible}{6}
\newcommand{\resVarChkFiftyHundredTenPFMaxMaxLoss}{11.3}

\newcommand{\resVarChkFiftyHundredTenPFMaxWorse}{1}

\newcommand{\resVarChkFiftyHundredTenPFMinInfeasible}{6}
\newcommand{\resVarChkFiftyHundredTenPFMinMaxLoss}{0.0}
\newcommand{\resVarChkFiftyHundredTenPFMinWorse}{0}

\newcommand{\resVarChkFiftyHundredTenSLInfeasible}{2}
\newcommand{\resVarChkFiftyHundredTenSLMaxLoss}{0.1}

\newcommand{\resVarChkFiftyHundredTenSLWorse}{1}
\newcommand{\resVarChkFiftyHundredTwentyFiveInstances}{9}

\newcommand{\resVarChkFiftyHundredTwentyFivePFCheapInfeasible}{2}
\newcommand{\resVarChkFiftyHundredTwentyFivePFCheapMaxLoss}{0.0}
\newcommand{\resVarChkFiftyHundredTwentyFivePFCheapWorse}{0}

\newcommand{\resVarChkFiftyHundredTwentyFivePFLightInfeasible}{2}
\newcommand{\resVarChkFiftyHundredTwentyFivePFLightMaxLoss}{0.0}
\newcommand{\resVarChkFiftyHundredTwentyFivePFLightWorse}{0}

\newcommand{\resVarChkFiftyHundredTwentyFivePFMaxInfeasible}{6}
\newcommand{\resVarChkFiftyHundredTwentyFivePFMaxMaxLoss}{18.4}

\newcommand{\resVarChkFiftyHundredTwentyFivePFMaxWorse}{1}

\newcommand{\resVarChkFiftyHundredTwentyFivePFMinInfeasible}{6}
\newcommand{\resVarChkFiftyHundredTwentyFivePFMinMaxLoss}{0.0}
\newcommand{\resVarChkFiftyHundredTwentyFivePFMinWorse}{0}

\newcommand{\resVarChkFiftyHundredTwentyFiveSLInfeasible}{0}
\newcommand{\resVarChkFiftyHundredTwentyFiveSLMaxLoss}{0.5}

\newcommand{\resVarChkFiftyHundredTwentyFiveSLWorse}{2}
\newcommand{\resVarChkFiftyNoBoundInstances}{9}

\newcommand{\resVarChkFiftyNoBoundPFCheapInfeasible}{0}
\newcommand{\resVarChkFiftyNoBoundPFCheapMaxLoss}{0.0}
\newcommand{\resVarChkFiftyNoBoundPFCheapWorse}{0}

\newcommand{\resVarChkFiftyNoBoundPFLightInfeasible}{0}
\newcommand{\resVarChkFiftyNoBoundPFLightMaxLoss}{15.4}

\newcommand{\resVarChkFiftyNoBoundPFLightWorse}{2}

\newcommand{\resVarChkFiftyNoBoundPFMaxInfeasible}{0}
\newcommand{\resVarChkFiftyNoBoundPFMaxMaxLoss}{8.8}

\newcommand{\resVarChkFiftyNoBoundPFMaxWorse}{5}

\newcommand{\resVarChkFiftyNoBoundPFMinInfeasible}{0}
\newcommand{\resVarChkFiftyNoBoundPFMinMaxLoss}{0.0}
\newcommand{\resVarChkFiftyNoBoundPFMinWorse}{0}

\newcommand{\resVarChkFiftyNoBoundSLInfeasible}{0}
\newcommand{\resVarChkFiftyNoBoundSLMaxLoss}{0.0}
\newcommand{\resVarChkFiftyNoBoundSLWorse}{0}
\newcommand{\resVarChkFiftyTwoHundredInstances}{9}

\newcommand{\resVarChkFiftyTwoHundredPFCheapInfeasible}{0}
\newcommand{\resVarChkFiftyTwoHundredPFCheapMaxLoss}{0.0}
\newcommand{\resVarChkFiftyTwoHundredPFCheapWorse}{0}

\newcommand{\resVarChkFiftyTwoHundredPFLightInfeasible}{0}
\newcommand{\resVarChkFiftyTwoHundredPFLightMaxLoss}{15.2}

\newcommand{\resVarChkFiftyTwoHundredPFLightWorse}{2}

\newcommand{\resVarChkFiftyTwoHundredPFMaxInfeasible}{0}
\newcommand{\resVarChkFiftyTwoHundredPFMaxMaxLoss}{9.0}

\newcommand{\resVarChkFiftyTwoHundredPFMaxWorse}{4}

\newcommand{\resVarChkFiftyTwoHundredPFMinInfeasible}{0}
\newcommand{\resVarChkFiftyTwoHundredPFMinMaxLoss}{0.0}
\newcommand{\resVarChkFiftyTwoHundredPFMinWorse}{0}

\newcommand{\resVarChkFiftyTwoHundredSLInfeasible}{0}
\newcommand{\resVarChkFiftyTwoHundredSLMaxLoss}{4.6}

\newcommand{\resVarChkFiftyTwoHundredSLWorse}{2}
\newcommand{\resVarChkTenHundredFiftyInstances}{9}

\newcommand{\resVarChkTenHundredFiftyPFCheapInfeasible}{2}
\newcommand{\resVarChkTenHundredFiftyPFCheapMaxLoss}{0.0}
\newcommand{\resVarChkTenHundredFiftyPFCheapWorse}{0}

\newcommand{\resVarChkTenHundredFiftyPFLightInfeasible}{0}
\newcommand{\resVarChkTenHundredFiftyPFLightMaxLoss}{12.4}

\newcommand{\resVarChkTenHundredFiftyPFLightWorse}{2}

\newcommand{\resVarChkTenHundredFiftyPFMaxInfeasible}{3}
\newcommand{\resVarChkTenHundredFiftyPFMaxMaxLoss}{10.8}

\newcommand{\resVarChkTenHundredFiftyPFMaxWorse}{3}

\newcommand{\resVarChkTenHundredFiftyPFMinInfeasible}{3}
\newcommand{\resVarChkTenHundredFiftyPFMinMaxLoss}{1.0}

\newcommand{\resVarChkTenHundredFiftyPFMinWorse}{2}

\newcommand{\resVarChkTenHundredFiftySLInfeasible}{0}
\newcommand{\resVarChkTenHundredFiftySLMaxLoss}{1.0}

\newcommand{\resVarChkTenHundredFiftySLWorse}{2}
\newcommand{\resVarChkTenHundredInstances}{9}

\newcommand{\resVarChkTenHundredPFCheapInfeasible}{2}
\newcommand{\resVarChkTenHundredPFCheapMaxLoss}{0.0}
\newcommand{\resVarChkTenHundredPFCheapWorse}{0}

\newcommand{\resVarChkTenHundredPFLightInfeasible}{0}
\newcommand{\resVarChkTenHundredPFLightMaxLoss}{0.0}
\newcommand{\resVarChkTenHundredPFLightWorse}{0}

\newcommand{\resVarChkTenHundredPFMaxInfeasible}{6}
\newcommand{\resVarChkTenHundredPFMaxMaxLoss}{11.4}

\newcommand{\resVarChkTenHundredPFMaxWorse}{1}

\newcommand{\resVarChkTenHundredPFMinInfeasible}{6}
\newcommand{\resVarChkTenHundredPFMinMaxLoss}{0.0}
\newcommand{\resVarChkTenHundredPFMinWorse}{0}

\newcommand{\resVarChkTenHundredSLInfeasible}{0}
\newcommand{\resVarChkTenHundredSLMaxLoss}{2.7}

\newcommand{\resVarChkTenHundredSLWorse}{2}
\newcommand{\resVarChkTenHundredTenInstances}{9}

\newcommand{\resVarChkTenHundredTenPFCheapInfeasible}{2}
\newcommand{\resVarChkTenHundredTenPFCheapMaxLoss}{0.0}
\newcommand{\resVarChkTenHundredTenPFCheapWorse}{0}

\newcommand{\resVarChkTenHundredTenPFLightInfeasible}{0}
\newcommand{\resVarChkTenHundredTenPFLightMaxLoss}{0.0}
\newcommand{\resVarChkTenHundredTenPFLightWorse}{0}

\newcommand{\resVarChkTenHundredTenPFMaxInfeasible}{6}
\newcommand{\resVarChkTenHundredTenPFMaxMaxLoss}{11.6}

\newcommand{\resVarChkTenHundredTenPFMaxWorse}{1}

\newcommand{\resVarChkTenHundredTenPFMinInfeasible}{6}
\newcommand{\resVarChkTenHundredTenPFMinMaxLoss}{0.0}
\newcommand{\resVarChkTenHundredTenPFMinWorse}{0}

\newcommand{\resVarChkTenHundredTenSLInfeasible}{0}
\newcommand{\resVarChkTenHundredTenSLMaxLoss}{4.8}

\newcommand{\resVarChkTenHundredTenSLWorse}{1}
\newcommand{\resVarChkTenHundredTwentyFiveInstances}{9}

\newcommand{\resVarChkTenHundredTwentyFivePFCheapInfeasible}{2}
\newcommand{\resVarChkTenHundredTwentyFivePFCheapMaxLoss}{0.0}
\newcommand{\resVarChkTenHundredTwentyFivePFCheapWorse}{0}

\newcommand{\resVarChkTenHundredTwentyFivePFLightInfeasible}{0}
\newcommand{\resVarChkTenHundredTwentyFivePFLightMaxLoss}{0.0}
\newcommand{\resVarChkTenHundredTwentyFivePFLightWorse}{0}

\newcommand{\resVarChkTenHundredTwentyFivePFMaxInfeasible}{6}
\newcommand{\resVarChkTenHundredTwentyFivePFMaxMaxLoss}{10.6}

\newcommand{\resVarChkTenHundredTwentyFivePFMaxWorse}{1}

\newcommand{\resVarChkTenHundredTwentyFivePFMinInfeasible}{6}
\newcommand{\resVarChkTenHundredTwentyFivePFMinMaxLoss}{0.0}
\newcommand{\resVarChkTenHundredTwentyFivePFMinWorse}{0}

\newcommand{\resVarChkTenHundredTwentyFiveSLInfeasible}{0}
\newcommand{\resVarChkTenHundredTwentyFiveSLMaxLoss}{1.0}

\newcommand{\resVarChkTenHundredTwentyFiveSLWorse}{2}
\newcommand{\resVarChkTenNoBoundInstances}{9}

\newcommand{\resVarChkTenNoBoundPFCheapInfeasible}{0}
\newcommand{\resVarChkTenNoBoundPFCheapMaxLoss}{0.0}
\newcommand{\resVarChkTenNoBoundPFCheapWorse}{0}

\newcommand{\resVarChkTenNoBoundPFLightInfeasible}{0}
\newcommand{\resVarChkTenNoBoundPFLightMaxLoss}{30.8}

\newcommand{\resVarChkTenNoBoundPFLightWorse}{2}

\newcommand{\resVarChkTenNoBoundPFMaxInfeasible}{0}
\newcommand{\resVarChkTenNoBoundPFMaxMaxLoss}{8.8}

\newcommand{\resVarChkTenNoBoundPFMaxWorse}{5}

\newcommand{\resVarChkTenNoBoundPFMinInfeasible}{0}
\newcommand{\resVarChkTenNoBoundPFMinMaxLoss}{0.0}
\newcommand{\resVarChkTenNoBoundPFMinWorse}{0}

\newcommand{\resVarChkTenNoBoundSLInfeasible}{0}
\newcommand{\resVarChkTenNoBoundSLMaxLoss}{0.0}
\newcommand{\resVarChkTenNoBoundSLWorse}{0}
\newcommand{\resVarChkTenTwoHundredInstances}{9}

\newcommand{\resVarChkTenTwoHundredPFCheapInfeasible}{0}
\newcommand{\resVarChkTenTwoHundredPFCheapMaxLoss}{0.0}
\newcommand{\resVarChkTenTwoHundredPFCheapWorse}{0}

\newcommand{\resVarChkTenTwoHundredPFLightInfeasible}{0}
\newcommand{\resVarChkTenTwoHundredPFLightMaxLoss}{30.2}

\newcommand{\resVarChkTenTwoHundredPFLightWorse}{2}

\newcommand{\resVarChkTenTwoHundredPFMaxInfeasible}{0}
\newcommand{\resVarChkTenTwoHundredPFMaxMaxLoss}{9.0}

\newcommand{\resVarChkTenTwoHundredPFMaxWorse}{6}

\newcommand{\resVarChkTenTwoHundredPFMinInfeasible}{0}
\newcommand{\resVarChkTenTwoHundredPFMinMaxLoss}{0.9}

\newcommand{\resVarChkTenTwoHundredPFMinWorse}{3}

\newcommand{\resVarChkTenTwoHundredSLInfeasible}{0}
\newcommand{\resVarChkTenTwoHundredSLMaxLoss}{4.6}

\newcommand{\resVarChkTenTwoHundredSLWorse}{3}
\newcommand{\resVarTokTwoHundredFiftyInstances}{8}

\newcommand{\resVarTokTwoHundredFiftyPFCheapInfeasible}{0}
\newcommand{\resVarTokTwoHundredFiftyPFCheapMaxLoss}{0.0}
\newcommand{\resVarTokTwoHundredFiftyPFCheapWorse}{0}

\newcommand{\resVarTokTwoHundredFiftyPFLightInfeasible}{0}
\newcommand{\resVarTokTwoHundredFiftyPFLightMaxLoss}{30.1}

\newcommand{\resVarTokTwoHundredFiftyPFLightWorse}{2}

\newcommand{\resVarTokTwoHundredFiftyPFMaxInfeasible}{0}
\newcommand{\resVarTokTwoHundredFiftyPFMaxMaxLoss}{10.1}

\newcommand{\resVarTokTwoHundredFiftyPFMaxWorse}{2}

\newcommand{\resVarTokTwoHundredFiftyPFMinInfeasible}{0}
\newcommand{\resVarTokTwoHundredFiftyPFMinMaxLoss}{0.0}
\newcommand{\resVarTokTwoHundredFiftyPFMinWorse}{0}

\newcommand{\resVarTokTwoHundredFiftySLInfeasible}{0}
\newcommand{\resVarTokTwoHundredFiftySLMaxLoss}{4.6}

\newcommand{\resVarTokTwoHundredFiftySLWorse}{2}
\newcommand{\resVarTokTwoHundredInstances}{9}

\newcommand{\resVarTokTwoHundredPFCheapInfeasible}{0}
\newcommand{\resVarTokTwoHundredPFCheapMaxLoss}{0.0}
\newcommand{\resVarTokTwoHundredPFCheapWorse}{0}

\newcommand{\resVarTokTwoHundredPFLightInfeasible}{0}
\newcommand{\resVarTokTwoHundredPFLightMaxLoss}{28.4}

\newcommand{\resVarTokTwoHundredPFLightWorse}{2}

\newcommand{\resVarTokTwoHundredPFMaxInfeasible}{0}
\newcommand{\resVarTokTwoHundredPFMaxMaxLoss}{8.0}

\newcommand{\resVarTokTwoHundredPFMaxWorse}{4}

\newcommand{\resVarTokTwoHundredPFMinInfeasible}{0}
\newcommand{\resVarTokTwoHundredPFMinMaxLoss}{0.3}

\newcommand{\resVarTokTwoHundredPFMinWorse}{1}

\newcommand{\resVarTokTwoHundredSLInfeasible}{0}
\newcommand{\resVarTokTwoHundredSLMaxLoss}{2.7}

\newcommand{\resVarTokTwoHundredSLWorse}{3}
\newcommand{\resVarTokTwoHundredTenInstances}{9}

\newcommand{\resVarTokTwoHundredTenPFCheapInfeasible}{0}
\newcommand{\resVarTokTwoHundredTenPFCheapMaxLoss}{0.0}
\newcommand{\resVarTokTwoHundredTenPFCheapWorse}{0}

\newcommand{\resVarTokTwoHundredTenPFLightInfeasible}{0}
\newcommand{\resVarTokTwoHundredTenPFLightMaxLoss}{28.4}

\newcommand{\resVarTokTwoHundredTenPFLightWorse}{2}

\newcommand{\resVarTokTwoHundredTenPFMaxInfeasible}{0}
\newcommand{\resVarTokTwoHundredTenPFMaxMaxLoss}{7.9}

\newcommand{\resVarTokTwoHundredTenPFMaxWorse}{5}

\newcommand{\resVarTokTwoHundredTenPFMinInfeasible}{0}
\newcommand{\resVarTokTwoHundredTenPFMinMaxLoss}{4.6}

\newcommand{\resVarTokTwoHundredTenPFMinWorse}{3}

\newcommand{\resVarTokTwoHundredTenSLInfeasible}{0}
\newcommand{\resVarTokTwoHundredTenSLMaxLoss}{3.6}

\newcommand{\resVarTokTwoHundredTenSLWorse}{3}
\newcommand{\resVarTokTwoHundredTwentyFiveInstances}{9}

\newcommand{\resVarTokTwoHundredTwentyFivePFCheapInfeasible}{0}
\newcommand{\resVarTokTwoHundredTwentyFivePFCheapMaxLoss}{0.0}
\newcommand{\resVarTokTwoHundredTwentyFivePFCheapWorse}{0}

\newcommand{\resVarTokTwoHundredTwentyFivePFLightInfeasible}{0}
\newcommand{\resVarTokTwoHundredTwentyFivePFLightMaxLoss}{28.4}

\newcommand{\resVarTokTwoHundredTwentyFivePFLightWorse}{2}

\newcommand{\resVarTokTwoHundredTwentyFivePFMaxInfeasible}{0}
\newcommand{\resVarTokTwoHundredTwentyFivePFMaxMaxLoss}{10.3}

\newcommand{\resVarTokTwoHundredTwentyFivePFMaxWorse}{6}

\newcommand{\resVarTokTwoHundredTwentyFivePFMinInfeasible}{0}
\newcommand{\resVarTokTwoHundredTwentyFivePFMinMaxLoss}{1.3}

\newcommand{\resVarTokTwoHundredTwentyFivePFMinWorse}{4}

\newcommand{\resVarTokTwoHundredTwentyFiveSLInfeasible}{0}
\newcommand{\resVarTokTwoHundredTwentyFiveSLMaxLoss}{1.3}

\newcommand{\resVarTokTwoHundredTwentyFiveSLWorse}{5}
\newcommand{\resVarTokTwoNoBoundInstances}{9}

\newcommand{\resVarTokTwoNoBoundPFCheapInfeasible}{0}
\newcommand{\resVarTokTwoNoBoundPFCheapMaxLoss}{0.0}
\newcommand{\resVarTokTwoNoBoundPFCheapWorse}{0}

\newcommand{\resVarTokTwoNoBoundPFLightInfeasible}{0}
\newcommand{\resVarTokTwoNoBoundPFLightMaxLoss}{30.8}

\newcommand{\resVarTokTwoNoBoundPFLightWorse}{2}

\newcommand{\resVarTokTwoNoBoundPFMaxInfeasible}{0}
\newcommand{\resVarTokTwoNoBoundPFMaxMaxLoss}{8.8}

\newcommand{\resVarTokTwoNoBoundPFMaxWorse}{5}

\newcommand{\resVarTokTwoNoBoundPFMinInfeasible}{0}
\newcommand{\resVarTokTwoNoBoundPFMinMaxLoss}{0.0}
\newcommand{\resVarTokTwoNoBoundPFMinWorse}{0}

\newcommand{\resVarTokTwoNoBoundSLInfeasible}{0}
\newcommand{\resVarTokTwoNoBoundSLMaxLoss}{0.0}
\newcommand{\resVarTokTwoNoBoundSLWorse}{0}
\newcommand{\resVarTokTwoTwoHundredInstances}{8}

\newcommand{\resVarTokTwoTwoHundredPFCheapInfeasible}{0}
\newcommand{\resVarTokTwoTwoHundredPFCheapMaxLoss}{0.0}
\newcommand{\resVarTokTwoTwoHundredPFCheapWorse}{0}

\newcommand{\resVarTokTwoTwoHundredPFLightInfeasible}{0}
\newcommand{\resVarTokTwoTwoHundredPFLightMaxLoss}{30.4}

\newcommand{\resVarTokTwoTwoHundredPFLightWorse}{2}

\newcommand{\resVarTokTwoTwoHundredPFMaxInfeasible}{0}
\newcommand{\resVarTokTwoTwoHundredPFMaxMaxLoss}{8.7}

\newcommand{\resVarTokTwoTwoHundredPFMaxWorse}{4}

\newcommand{\resVarTokTwoTwoHundredPFMinInfeasible}{0}
\newcommand{\resVarTokTwoTwoHundredPFMinMaxLoss}{0.0}
\newcommand{\resVarTokTwoTwoHundredPFMinWorse}{0}

\newcommand{\resVarTokTwoTwoHundredSLInfeasible}{0}
\newcommand{\resVarTokTwoTwoHundredSLMaxLoss}{0.0}
\newcommand{\resVarTokTwoTwoHundredSLWorse}{0}

\newcommand{\resXsolverCpSatOneProven}{8}

\newcommand{\resXsolverGecodeProven}{1}

\newcommand{\resXsolverHighsProven}{2}
\newcommand{\resXsolverHighsSolved}{6}
\newcommand{\resXsolverInstances}{10}

\usepackage{placeins}
\usepackage{etoolbox}
\newcommand{\tablefontsizes}{%
  \renewcommand{\small}{\fontsize{9}{10}\selectfont}%
  \renewcommand{\footnotesize}{\fontsize{8}{9}\selectfont}%
  \renewcommand{\scriptsize}{\fontsize{7}{8}\selectfont}}
\AtBeginEnvironment{table}{\tablefontsizes}
\AtBeginEnvironment{table*}{\tablefontsizes}
\newcommand{\artifacturl}{\url{https://github.com/SBU-CPS-Lab/SafeDSE}}
\newcommand{\experimentsurl}{\url{https://github.com/SBU-CPS-Lab/SafeDSE-experiments}}

\begin{document}

\title{SafeDSE: Joint Exploration of Safety Architecture Patterns, Mapping,
and Scheduling for Dataflow Applications With Safety Arguments
(Extended Version)}

\author{Fatemeh~Bahrami
        and~Seyed-Hosein~Attarzadeh-Niaki%
\thanks{This technical report is the extended version of a shorter paper
that has been submitted for publication. It adds the material that the page
limit of the shorter version leaves out: full model and verifier details,
the complete pattern catalog and parameters, the proof of the potential
characterization, all measured results per instance, and the negative
results of the development.}%
\thanks{The authors are with the Faculty of Computer Science and
Engineering, Shahid Beheshti University, Tehran, Iran (e-mail:
farnoosh3@gmail.com; h\_attarzadeh@sbu.ac.ir). Corresponding author:
Seyed-Hosein Attarzadeh-Niaki.}}

\markboth{SafeDSE: Extended Technical Report}%
{Bahrami and Attarzadeh-Niaki: SafeDSE (Extended Version)}

\maketitle
\suppressfloats[t]                %
\begin{abstract}
Embedded platforms increasingly host safety-related streaming functions next
to less critical ones. Whether a safety architecture pattern, such as a
doer/checker pair, protects a function depends on the cores that run
its components. Their schedule must meet the throughput requirement,
and mixing safety integrity levels (SILs) on one core raises development
cost. This report presents \safedse{}, a design space exploration framework
for synchronous dataflow applications. One declarative constraint model
decides the pattern of each actor, the placement that the fault model
requires, the static-order schedule, and the platform, and prices SIL
promotion. Guards activate the chosen patterns in a fixed graph that holds
all of them, and node potentials encode throughput. An independent verifier
re-checks every solution, and a generator turns its checks into evidence of
a safety argument in Goal Structuring Notation. On
\resBaselineInstances{} benchmark instances under period bounds, no rule
that fixes patterns before mapping was both always feasible and never more
expensive under conservative pricing, while joint exploration was. In the
generated arguments, \resGsnBothSolutionShare{} to
\resGsnNvpSolutionShare{}\,\% of the leaves cite verifier checks; the rest
are claims a mapping cannot discharge. Optimality within the model was
proven for all synthetic graphs with up to \resScalMaxProvenActors{}
actors. This extended version adds the full model, the verifier and
generator in detail, the complete catalog and parameters, and a proof of
the potential characterization. It also reports every experiment per
instance, a mutation analysis of the verifier, and the negative results of
the development.
\end{abstract}


\section{Introduction}\label{sec:intro}

\begin{figure*}[!t]
  \centering
  \includegraphics[width=\textwidth]{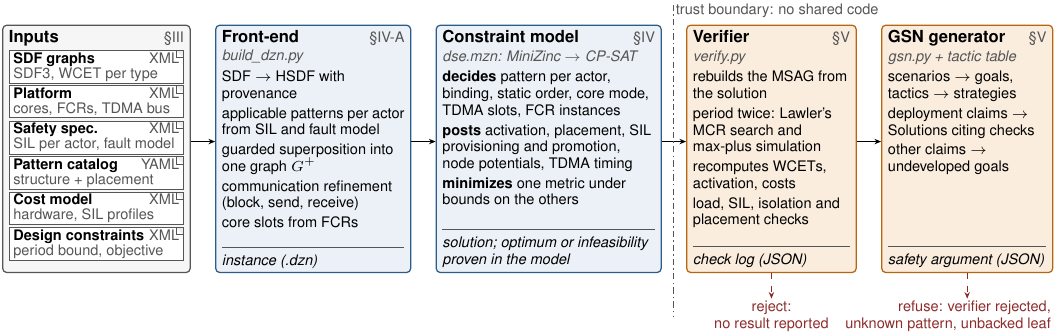}
  \caption{\safedse{} design flow. The front-end turns all design
  alternatives into one fixed constraint instance; the model decides and
  proves the optimum within the model. Only checks of an independent
  verifier, which shares no code with the model, become evidence in the
  generated argument; dashed arrows are rejections and refusals.}
  \label{fig:flow}
\end{figure*}

\IEEEPARstart{E}{mbedded} systems increasingly integrate functions of
different criticality on one multiprocessor
platform~\cite{burns2018survey}. Streaming functions, which process images,
audio, or sensor data, are often specified as synchronous dataflow (SDF)
graphs~\cite{lee1987synchronous} and must sustain a throughput. Under
IEC~61508, each safety function has a safety integrity level (SIL), and the
techniques and measures required for its software become stricter with the
SIL~\cite{iec61508part3,myklebust2015software}. Experts estimated that
moving from SIL~2 to SIL~3 may double the development
cost~\cite{myklebust2015software}. To integrate such a function, a designer
chooses a safety architecture pattern (e.g., a doer/checker pair or a
2-of-2 pair), the cores that run its components, a schedule that meets the
throughput requirement, and the cores to buy.

Each part of this problem has been studied. Design space exploration (DSE)
with constraint programming (CP) maps and schedules SDF graphs on
multiprocessors under throughput constraints and proves the result
optimal~\cite{bonfietti2013maximum,rosvall2017throughput,rosvall2018flexible},
or proves that no mapping meets the requirements~\cite{jordao2023design}.
Safety architecture patterns have been cataloged with their structure,
tactics, and
costs~\cite{armoush2010design,preschern2013building,koopman2016challenges}.
Optimization methods have allocated and decomposed SILs to reduce
development
cost~\cite{tamasselicean2015design,parker2013automatic,zaheri2025asil},
added redundancy for reliability~\cite{glass2007reliability,kang2014static},
and selected architectures from templates~\cite{nuzzo2020optimized}.

However, the choice of a pattern and its deployment depend on each other.
A checker covers random hardware faults of its doer's core only if it runs
on another core; on the same core, it covers only systematic faults
(\cref{sec:bg-safety}). Each checker or replica adds execution time and
often a cycle to the dataflow graph, so a pattern that meets the throughput
requirement under one mapping can miss it under another. Actors of
different SILs on one core without certified partitioning promote the
lower-SIL ones, so saving cores costs development effort. DSE tools take
redundancy as an input, and catalogs describe patterns without a platform.
The methods that combine both search
heuristically~\cite{tamasselicean2015design,kang2014static,parker2013automatic},
separate redundant parts by a fixed rule~\cite{zaheri2025asil}, or choose
patterns without a timing model~\cite{gildantas2023toolchain}. Safety
arguments are built for a given
design~\cite{hawkins2015weaving,denney2018tool}, and a DSE tool delivers no
evidence for the claims of the argument. It therefore remains unclear
whether pattern choice, placement, scheduling, and SIL cost can be decided
in one exact search, and whether the chosen design can carry its own safety
evidence.

This report presents \safedse{}, a DSE framework that takes these decisions
in one declarative constraint model (\cref{fig:flow}). A front-end places
every applicable pattern, every communication actor that a remote channel
may need, and every core that can be bought in one fixed graph. A
\mzn{}~\cite{nethercote2007minizinc} model, solved with
\cpsat{}~\cite{perron2026cpsat}, chooses the pattern of each actor, the core
and static order of each node, and the platform, and proves the result
optimal within the model. A verifier that shares no code with the model
re-checks each solution, and a generator turns the checks into evidence in a
safety argument in Goal Structuring Notation (GSN)~\cite{gsn2021standard}.
On \resBaselineInstances{} benchmark instances under period bounds, no rule
that fixes the pattern before mapping was both always feasible and never
more expensive under our conservative pricing. With a lower-SIL doer and
priced diverse development, the cost-aware rule reached the joint optimum.
Joint exploration was feasible and optimal in all cases.

The contributions are:
\begin{enumerate}
  \renewcommand{\labelenumi}{C\arabic{enumi})}
  \item\label{contrib:superposition} Guarded superposition: patterns with
  their explicit voters, communication refinement, and platform instances
  become guarded parts of one fixed graph, so pattern choice, replica
  placement, static-order scheduling, and platform selection form one CP
  problem (\cref{sec:model-superposition}).
  \item\label{contrib:safety} A placement-aware safety model: a pattern is a
  structure plus placement relations tied to fault classes. Each core is
  provisioned to a SIL under the IEC~61508 rule on mixing SILs or certified
  partitioning, and a cost profile prices SIL promotion
  (\cref{sec:bg-safety,sec:model-safety}).
  \item\label{contrib:throughput} Throughput on a guarded graph: the
  potential form of cyclic scheduling~\cite{hanen1995study} is posted on the
  mapping- and scheduling-aware graph under reified binding, order, and
  activation conditions. A CP solver thus decides mapping, static orders, and
  period together without a throughput propagator. The encoding excludes
  deadlock and couples applications that share a core to one period
  (\cref{sec:model-throughput}).
  \item\label{contrib:assurance} Evidence-backed safety arguments: an
  independent verifier computes each period by two disjoint methods and logs
  every check. A GSN generator cites these records for the deployment claims
  and refuses to emit an argument with an unsupported leaf
  (\cref{sec:verification}).
\end{enumerate}

\emph{Relation to the shorter version.} This report extends a shorter
paper that has a page limit. It keeps the structure of that paper and
restores what was cut. \Cref{sec:related} gives the related work in
more detail. \Cref{sec:background} adds the complete catalog with its
scenarios and the cost profiles. \Cref{sec:model} describes every
constraint family of the model, including communication timing, the static
order, the implied constraints, and the exact variant of the period
constraints, and \cref{app:proof} proves the potential characterization in
full. \Cref{sec:verification} lists every verifier record, the tactic table,
and the refusal conditions of the generator. \Cref{sec:eval} reports every
experiment per instance, including the head-to-head with \desyde{}, input
variants, a cost-model extension, model size, cross-solver runs, and a
mutation analysis of the verifier. \Cref{sec:lessons} collects the negative
results and the defects that the evaluation found in the tool, and
\cref{app:instances,app:artifact} list the instances, parameters, and the
scripts that produce each table and figure. All numbers come from the same
measured data as the shorter paper. The model and the tools are available
at \artifacturl, and the scripts and data of all experiments at
\experimentsurl.

\FloatBarrier
\section{Related Work}\label{sec:related}

\subsection{Dataflow Mapping and Constraint-Based DSE}\label{sec:related-dse}
Throughput-constrained mapping of SDF graphs has been studied with
heuristics and with exact methods. Stuijk et
al.~\cite{stuijk2007multiprocessor} allocated resources heuristically on top
of exact throughput analysis~\cite{ghamarian2006throughput}. Bonfietti et
al.\ posted a throughput constraint for SDF graphs in a CP
solver~\cite{bonfietti2009throughput} and later mapped and scheduled SDF
graphs for maximum throughput with complete CP~\cite{bonfietti2013maximum}.
\desyde{}~\cite{rosvall2014constraint,rosvall2018flexible} extended CP-based
exploration to heterogeneous platforms with time-division multiple access
(TDMA) communication and to trade-offs of throughput against buffer size
and power. Rosvall et al.~\cite{rosvall2017throughput} added throughput
propagators and hard constraints for some applications next to best-effort
ones. IDeSyDe~\cite{jordao2024idesyde} makes such models modular, and
Jord\~ao et al.~\cite{jordao2023design} used it to map avionics functions to
integrated modular avionics (IMA) platforms. Their CP model minimizes the
number of processing elements and either finds a mapping that meets all
requirements or proves that none exists. The redundant part of a flight
function is an input. Since the model is optimistic about partitioning, an
infeasibility carries over to partitioned designs.

Cyclic scheduling characterizes periodic schedules by node
potentials~\cite{hanen1995study}. Hanen and Munier showed that the least
period of a uniform graph is the ratio of its critical circuit. Fixing a
cyclic order of the tasks on each processor adds uniform
constraints (the tie-breaking graph) whose critical circuit gives the best
schedule for that order. This construction is the analogue of the
static-order and wrap edges of the mapping- and scheduling-aware graph
(MSAG) used in this report. Govindarajan et al.~\cite{govindarajan1996framework}
put Reiter's characterization into one integer linear program that schedules
loops and maps them to function units.

Reliability-aware synthesis adds redundancy during the search. Izosimov et
al.~\cite{izosimov2005design} chose re-execution and replication for
fault-tolerant distributed systems under time and cost constraints. Glaß et
al.~\cite{glass2007reliability} and Kang et al.~\cite{kang2014static}
replicated or re-executed tasks, and Bolchini and
Miele~\cite{bolchini2013reliability} hardened tasks by duplication with
comparison or triple modular redundancy. These works search evolutionarily
against transient faults. Kang et al.\ also drop low-criticality task graphs under
faults, use passive replication, derive response-time bounds, and search
with SPEA-II. In the CP line, the redundancy of an application is an input;
in the reliability line, the search is heuristic and certification cost is
not modeled. \safedse{} keeps exact CP search and the MSAG of \desyde{},
makes the pattern a decision, and posts the potential form on the MSAG under
reified conditions instead of a throughput propagator
(\cref{sec:model-throughput}).

\subsection{Safety Patterns, SIL Cost, and Mixed-Criticality Optimization}\label{sec:related-safety}
Pattern catalogs describe safety architectures with their effect on safety,
reliability, and cost~\cite{armoush2010design,preschern2013building}.
Douglass~\cite{douglass2002realtime} described, among others, the
Monitor-Actuator, Homogeneous Redundancy, Heterogeneous Redundancy, and
triple modular redundancy patterns, and Koopman and
Wagner~\cite{koopman2016challenges} explained how a high-SIL monitor lets a
low-SIL actuator carry a high-SIL function. The designer applies these
patterns by hand. Mixed-criticality scheduling, which started from Vestal's
criticality-dependent WCETs~\cite{vestal2007preemptive}, solves another
problem: it guarantees timing at several levels of
assurance~\cite{burns2018survey} and does not choose redundancy. In this
report, ``different criticality'' means different SILs, not Vestal's sense.

The closest work is by T\u{a}ma\c{s}-Selicean and
Pop~\cite{tamasselicean2011design,tamasselicean2015design}. They mapped
tasks of different SILs to a partitioned architecture with static-cyclic
schedules, assigned tasks to partitions, and sized the partition time slots.
Tasks that share a partition are elevated to the highest SIL among them,
and high-SIL tasks are decomposed into redundant lower-SIL tasks from a library
taken from ISO~26262-9, connected through distributor and collector tasks.
They minimized development cost with Tabu Search and cited a certification
overhead of 25 to 100\,\%. Fault tolerance is outside their problem. Our
promotion rule (\cref{sec:bg-safety}) uses the same certification argument,
applied per core. HiP-HOPS synthesizes fault trees from architecture
models~\cite{papadopoulos2011engineering}, and Parker et
al.~\cite{parker2013automatic} decomposed and allocated SILs on them with a
genetic algorithm and a SIL-dependent cost. Zaheri and
Reuss~\cite{zaheri2025asil} solved the decomposition of automotive SILs
(ASILs), mapping, and non-preemptive scheduling as one mixed-integer linear
program (MILP) that minimizes development cost and chain latency. Redundant
subtasks never share an electronic control unit, by a fixed constraint.
ArchEx~\cite{nuzzo2020optimized} selected a minimum-cost architecture from
a template by MILP under a reliability bound and an additive latency bound.
It has an algorithm that is optimal within the error of an approximate
reliability algebra and a lazy one that is sound but possibly sub-optimal;
its patterns are requirement patterns that generate MILP constraints, not
architecture patterns. SafPat~\cite{gildantas2023toolchain}, implemented in
answer set programming and validated by simulated fault injection,
recommends patterns, including the hardware separation of replicas, from
the fault type. The deployment of the extended architecture is then
optimized with the SMT-based synthesis of
AutoFOCUS3~\cite{terzimehic2023safety}, which minimizes safety-critical
coupling between control units. AutoFOCUS3 also respects core SIL limits
with time-triggered schedules, with a binary search on the schedule length
over an SMT solver~\cite{voss2013deployment}. It places given master and
standby replicas on nodes with different power supplies, choosing hot or
cold standby from the fault-tolerance and recovery times, without
optimization~\cite{becker2014formal}. As \cref{tab:related} shows, none of
these works decides the pattern, its placement for a given fault model, and
a throughput-feasible schedule in one exact search.

\begin{table}[!t]
  \caption{Comparison With Related Work}
  \label{tab:related}
  \centering\footnotesize
  \setlength{\tabcolsep}{2pt}
  \begin{tabular*}{\columnwidth}{@{\extracolsep{\fill}}l cccccc@{}}
    \toprule
    Approach & Thr. & Exact & Pat. & Place. & Cost & Arg. \\
    \midrule
    \desyde{}~\cite{rosvall2014constraint,rosvall2018flexible} & \cmark & \cmark & \xmark & \xmark & \xmark & \xmark \\
    Bonfietti et al.~\cite{bonfietti2009throughput,bonfietti2013maximum} & \cmark & \cmark & \xmark & \xmark & \xmark & \xmark \\
    Jord\~ao et al.~\cite{jordao2023design} & \pmark & \cmark & \xmark & \xmark & \xmark & \xmark \\
    Kang et al.~\cite{kang2014static} & \pmark & \xmark & \pmark & \pmark & \xmark & \xmark \\
    T\u{a}ma\c{s}-Selicean, Pop~\cite{tamasselicean2015design} & \pmark & \xmark & \pmark & \xmark & \cmark & \xmark \\
    Parker et al.~\cite{parker2013automatic} & \xmark & \xmark & \pmark & \xmark & \cmark & \xmark \\
    Zaheri, Reuss~\cite{zaheri2025asil} & \pmark & \cmark & \pmark & \pmark & \cmark & \xmark \\
    ArchEx~\cite{nuzzo2020optimized} & \pmark & \cmark & \pmark & \xmark & \xmark & \xmark \\
    AutoFOCUS3~\cite{voss2013deployment,becker2014formal} & \pmark & \cmark & \pmark & \pmark & \xmark & \xmark \\
    SafPat~\cite{gildantas2023toolchain,terzimehic2023safety} & \xmark & \pmark & \cmark & \cmark & \xmark & \xmark \\
    AdvoCATE~\cite{denney2018tool} & \xmark & \xmark & \xmark & \xmark & \xmark & \cmark \\
    \midrule
    \safedse{} & \cmark & \cmark & \cmark & \cmark & \cmark & \cmark \\
    \bottomrule
  \end{tabular*}

  \vspace{2pt}\parbox{\columnwidth}{\scriptsize Columns: the decisions
  and outputs of the problem in \cref{sec:bg-problem}. Thr.: \cmark{} SDF
  throughput, \pmark{} deadline or latency. Exact: optimum proven, \pmark{}
  for the deployment only. Pat.: \cmark{} catalog patterns chosen,
  \pmark{} redundancy or SIL decomposition only. Place.: \cmark{} separation
  from the fault model, \pmark{} fixed rule. Cost: SIL-dependent
  development cost. Arg.: safety argument produced.}
\end{table}

\subsection{Model-Based Safety Argument Generation}\label{sec:related-gsn}
GSN~\cite{gsn2021standard} is the common notation for safety arguments, and
Preschern et al.~\cite{preschern2013building} gave a GSN argument for each
pattern. Hawkins et al.~\cite{hawkins2015weaving} linked argument elements
to design models so that argument parts can be generated from them.
AdvoCATE~\cite{denney2018tool} supports the creation of arguments, the
integration of formal methods, the instantiation of argument patterns from
structured data, abstraction, queries, and the verification of the argument
structure. Resolute~\cite{gacek2014resolute} discharges claims with
analyses of an AADL architecture model through rules written in a
domain-specific language, so that the evidence is computed. These tools take
the design as given. Assurance~2.0~\cite{bloomfield2021assurance} asks for
deductive argument steps, a search for defeaters, and an assessment of
confidence. In \safedse{}, the argument describes the design that the
exploration chose, and each deployment claim cites a check of an
independent verifier (\cref{sec:verification}); the search for defeaters
stays with the engineer.

\FloatBarrier
\section{Preliminaries and Problem Formulation}\label{sec:background}

\subsection{Application and Platform Models}\label{sec:bg-models}
An application is a synchronous dataflow (SDF) graph~\cite{lee1987synchronous}:
actors fire with fixed rates on channels that may hold initial tokens. Like
\desyde{}~\cite{rosvall2018flexible}, we map and schedule the unfolded
homogeneous SDF (HSDF) graph, in which every node fires once per iteration
and records its SDF actor and copy index. Each core runs its nodes
self-timed in a static order.

\begin{definition}[Period]\label{def:period}
Let $G$ be an HSDF graph with execution time $T_i$ of node $i$ and $\tok_{ij}$
initial tokens on edge $(i,j)$. The cycle ratio of a cycle $C$ is
$\sum_{i\in C}T_i / \sum_{(i,j)\in C}\tok_{ij}$, and $\mathrm{MCR}(G)$ is the
maximum cycle ratio over all cycles of $G$. The period of $G$ is
$\mu=\mathrm{MCR}(G)$ and its throughput is
$1/\mu$~\cite{ghamarian2006throughput}.
\end{definition}

The platform offers core types. A core type has a worst-case execution time
(WCET) per actor and operating mode, an integer in the benchmark's time
unit. A missing entry forbids the binding: without a timing figure, the
binding is not certified. It also has the highest SIL
to which software on it can be certified, a memory size, and optionally
certified partitioning at a cost. Each mode has a price and a power. Cores
are bought as cards, each an instance of a card template. The cores on a
card share power and clock, so a card is a fault containment region (FCR).
Remote channels use a TDMA bus whose slots go to the sending
cores~\cite{rosvall2018flexible}. \Cref{fig:example-a} shows the running
example and the FCR templates of the platform, and \cref{tab:platforms}
lists all card templates of the evaluation.

\begin{table*}[!t]
  \caption{Card Templates of the Platforms in the Evaluation}
  \label{tab:platforms}
  \centering\footnotesize
  \setlength{\tabcolsep}{3pt}
\begin{tabular*}{\textwidth}{@{\extracolsep{\fill}}l l l c c r >{\raggedright\arraybackslash}p{1.55in} c c r@{}}
\toprule
Platform file & Card & Core type & Cores & Max. & Card price & Modes (cycle factor, core price) & $\silmax$ & Part. & Memory \\
\midrule
\code{mixed.xml} & R-card & cortexR & 2 & 3 & 120 & standard (1, 40); eco (1.3, 40) & 4 & yes, 15 & 64000 \\
 & M-card & cortexM & 4 & 4 & 25 & standard (1.6, 8); performance (1.2, 11) & 2 & no & 16000 \\
\addlinespace[1pt]
\code{mixed\_noiso.xml} & R-card & cortexR & 2 & 3 & 120 & standard (1, 40); eco (1.3, 40) & 4 & no & 64000 \\
 & M-card & cortexM & 4 & 4 & 25 & standard (1.6, 8); performance (1.2, 11) & 2 & no & 16000 \\
\addlinespace[1pt]
\code{mixed\_3type.xml} & R-card & cortexR & 2 & 3 & 120 & standard (1, 40); eco (1.3, 40) & 4 & no & 64000 \\
 & P-card & ppcE200 & 2 & 3 & 140 & standard (1.15, 45) & 4 & no & 48000 \\
 & M-card & cortexM & 4 & 4 & 25 & standard (1.6, 8); performance (1.2, 11) & 2 & no & 16000 \\
\addlinespace[1pt]
\code{mixed\_4type.xml} & R-card & cortexR & 2 & 3 & 120 & standard (1, 40); eco (1.3, 40) & 4 & no & 64000 \\
 & P-card & ppcE200 & 2 & 3 & 140 & standard (1.15, 45) & 4 & no & 48000 \\
 & M-card & cortexM & 4 & 4 & 25 & standard (1.6, 8); performance (1.2, 11) & 2 & no & 16000 \\
 & S-card & sparcLEON & 2 & 3 & 150 & standard (1.25, 48) & 4 & no & 48000 \\
\addlinespace[1pt]
\code{mixed\_costly.xml} & R-card & cortexR & 2 & 3 & 480 & standard (1, 160); eco (1.3, 160) & 4 & no & 64000 \\
 & M-card & cortexM & 4 & 4 & 100 & standard (1.6, 32); performance (1.2, 44) & 2 & no & 16000 \\
\bottomrule
\end{tabular*}

\vspace{2pt}\parbox{\textwidth}{\scriptsize Used by: \code{mixed.xml}: \code{p\_rasta}; \code{mixed\_noiso.xml}: \code{f\_random\_hw}, \code{d\_*}, \code{c\_pat}, RQ4; \code{mixed\_3type.xml}: \code{f\_sw3}, \code{f\_both3}; \code{mixed\_4type.xml}: \code{v\_nvp}; \code{mixed\_costly.xml}: RQ4 (old \cref{fig:profcons}). Max.: cards the front-end may instantiate. Cycle factor: multiplier of the nominal execution time in that mode. Part.: certified partitioning and its cost per partitioned core. TDMA bus of the card platforms: 16 slots, at most 8 per core. Prices in cost units; they are illustrative.}

\end{table*}

\subsection{Safety Model}\label{sec:bg-safety}
The safety specification assigns each SDF actor $a$ a required SIL
$\silreq_a\in\{0,\dots,4\}$ from a hazard and risk analysis (SIL~0: no
requirement)~\cite{iec61508part3}. It fixes a fault model: random hardware
faults (H), systematic faults (S) that differ between core types and their
toolchains or between algorithms, or both. Design faults of the application
software that reappear in every implementation are not in S; they need
diverse software development, which \safedse{} does not establish.

A safety architecture
pattern~\cite{armoush2010design,preschern2013building,koopman2016challenges,douglass2002realtime}
adds components, with edges and initial tokens, to the actor it protects,
the owner. The edge that returns a verdict to the owner holds one token, so
the fault reaction starts within one period. Each component fires once per
iteration and produces on all its outputs; a detected fault changes token
values, not the presence of tokens, so the timing analysis stays valid under
a fault. Placement relations link the roles by \code{SAME} or
\code{DIFFERENT} core, \code{DIFFERENT\_FCR}, or \code{DIVERSE} (different
core types), mostly per fault class. On a platform with a single core type,
\code{DIVERSE} cannot hold, so heterogeneity is a precondition for any
defense against systematic faults.

The catalog (\cref{tab:catalogue}) gives the SIL range for which a pattern
is admitted, the fault classes it covers, a recurring cost, and the tactics
and general scenarios used by the argument generator
(\cref{sec:verification}). \Cref{tab:catalog-full} lists the remaining
fields of each record, and \cref{tab:scenarios} the scenarios. The SIL
range is necessary, not sufficient: \safedse{} computes no failure rates,
hardware fault tolerance, or diagnostic coverage. The ceiling of the
pattern \code{none} at SIL~1 follows IEC~61508's expectation of diagnostic
coverage from SIL~2 upward; it is the most consequential number of the
catalog and therefore lives in the data, where it can be argued with. The
2-of-2 pair runs only while both channels agree, i.e., 1oo2 in IEC~61508
notation~\cite{iec61508part4,iec61508part6}. The pattern names are
Koopman's; Koopman and Wagner~\cite{koopman2016challenges} published the
doer/checker (monitor/actuator) pair, Douglass~\cite{douglass2002realtime}
has counterparts for four rows, and Armoush~\cite{armoush2010design} for
three (\cref{tab:catalogue}). Armoush calls the mixed-SIL doer/checker and
the diverse doer/checker homogeneous and heterogeneous duplex patterns, and
the low-SIL diverse doer/checker is our combination of two catalog
patterns. N-version programming is IEC~61508-7 C.4.4.

In \cref{fig:example-b}, the two doer/checker patterns have the same
structure and differ only in \code{DIFFERENT} versus \code{SAME}, and the
fault model selects between them. A pattern is applicable to $a$ if its SIL
range contains $\silreq_a$, it covers the fault model, and the platform has
enough core types certifiable at that SIL for its \code{DIVERSE} groups.
An inapplicable pattern is dropped as a candidate; an instance is rejected
only if an actor has no applicable pattern left.

If software implements safety functions of different SILs, IEC~61508
requires that all of it is developed to the highest SIL unless their
independence is shown~\cite{iec61508part3}. We apply this rule per core and
accept certified partitioning as evidence of independence; otherwise, a
low-SIL actor next to a high-SIL actor is promoted. This assumes that cores
interfere only through the bus, not through shared memory, caches, or
I/O~\cite{easa2022amc20193}. Communication actors are exempt under the
black-channel principle~\cite{iec61784part3}: the communicating actors
detect corruption, loss, and delay end to end. Whether a transfer is a
safety function is a system-level assumption; the front-end offers two
alternatives (\cref{sec:model-comm}), and the generated argument states the
one used. Development cost is a baseline cost times a multiplier $k(s)$
for SIL $s$. Three profiles of $k$, which disagree strongly, come from the
estimates that Myklebust et al.\ collected~\cite{myklebust2015software}:
their own (default), Klosterman's, and DO-178B ratios with levels D and C as
SIL~2 and~3 (\cref{tab:profiles}). Klosterman's table has two rows that map
to SIL~3; the profile takes the higher one, and the evaluation also uses
the lower one.

\begin{figure}[!t]
  \centering
  \subfloat[]{\includegraphics{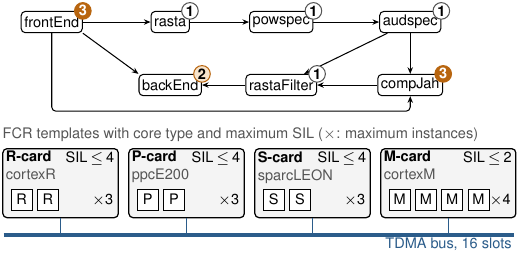}\label{fig:example-a}}\\[-2pt]
  \subfloat[]{\includegraphics{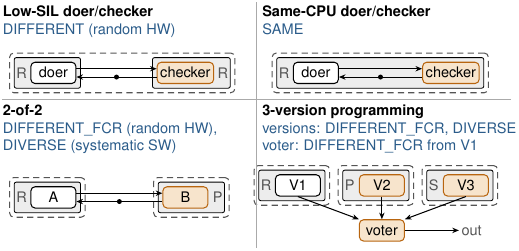}\label{fig:example-b}}
  \caption{(a) Running example: RASTA speech front end with the required
  SIL of each actor (badge), and the FCR templates of the platform. (b) Four
  catalog patterns with their placement relations (white: owner, shaded:
  added component, dot: initial token). Both doer/checker patterns have the
  same structure and differ only in placement, which the fault model
  selects; the evaluation instances use subsets of the catalog.}
  \label{fig:example}
\end{figure}

\begin{table}[!t]
  \caption{Pattern Catalog (Excerpt of the Data File)}
  \label{tab:catalogue}
  \centering\footnotesize
  \setlength{\tabcolsep}{3pt}
\begin{tabular*}{\columnwidth}{@{\extracolsep{\fill}}l l >{\raggedright\arraybackslash}p{1.22in} c c@{}}
\toprule
Pattern & Adds & Placement per fault class & SIL & Covers \\
\midrule
none & -- & -- & 0--1 & -- \\
Low-SIL D/C & $c$ & \mbox{H: $o\neq_{c}c$} & 1--2 & H \\
Same-CPU D/C & $c$ & \mbox{$o=c$} & 1--2 & S \\
2-of-2 & $b$ & \mbox{H: $o\neq_{f}b$}; \mbox{S: $o\neq_{t}b$} & 3--4 & H, S \\
High-SIL isolated D/C & $c$ & \mbox{H: $o\neq_{f}c$} & 3 & H \\
Mixed-SIL D/C & $c_1,c_2$ & \mbox{H: $o\neq_{f}c_1$}, \mbox{$c_1\neq_{f}c_2$}; \mbox{S: $c_1\neq_{t}c_2$} & 3--4 & H, S \\
Dual 2-of-2 & $b,s_a,s_b$ & \mbox{H: $o\neq_{f}b$}, \mbox{$s_a\neq_{f}s_b$}, \mbox{$o\neq_{f}s_a$} & 3--4 & H \\
1-of-2 failover & $b$ & \mbox{H: $o\neq_{f}b$} & 0--1 & H \\
Low-SIL diverse D/C & $c$ & \mbox{H: $o\neq_{f}c$}; \mbox{S: $o\neq_{t}c$} & 1--2 & H, S \\
3-version (NVP)$^\dagger$ & $v_2,v_3,w$ & \mbox{H: $\neq_{f}\{o,v_2,v_3\}$}, \mbox{$o\neq_{f}w$}; \mbox{S: $\neq_{t}\{o,v_2,v_3\}$} & 3--4 & H, S \\
\bottomrule
\end{tabular*}

\vspace{2pt}\parbox{\columnwidth}{\scriptsize Roles: $o$ owner (the application actor); $c,c_1,c_2$ checkers; $b$ second channel of the owner's pair; $s_a,s_b$ standby pair; $v_2,v_3$ versions; $w$ voter. $\neq_{c}$, $\neq_{f}$, $\neq_{t}$: different core, FCR, core type; $=$: same core. S is covered through different core types or, for Same-CPU D/C, the checker's algorithm. D/C: doer/checker. NVP: N-version programming. $^\dagger$Separate catalog, used in one case study; its voter $w$ is not replicated. Sources: Koopman (names; D/C pair in~\cite{koopman2016challenges}): all except 3-version (NVP); Douglass~\cite{douglass2002realtime}: Low-SIL D/C, High-SIL isolated D/C, 1-of-2 failover, Low-SIL diverse D/C; Armoush~\cite{armoush2010design}: Mixed-SIL D/C, Low-SIL diverse D/C, 3-version (NVP).}

\end{table}

\begin{table}[!t]
  \caption{Development-Cost Profiles}
  \label{tab:profiles}
  \centering\footnotesize
  \setlength{\tabcolsep}{4pt}
\begin{tabular*}{\columnwidth}{@{\extracolsep{\fill}}l rrrrr@{}}
\toprule
Profile & SIL~0 & 1 & 2 & 3 & 4 \\
\midrule
Myklebust et al. & 1.00 & 1.13 & 2.25 & 5.18 & 9.06 \\
Klosterman & 1.00 & 1.13 & 1.23 & 1.70 & 2.35 \\
Klosterman (lower row) & 1.00 & 1.13 & 1.23 & 1.40 & 1.59 \\
DO-178B & 1.00 & 2.00 & 3.00 & 5.00 & 9.00 \\
\bottomrule
\end{tabular*}

\vspace{2pt}\parbox{\columnwidth}{\scriptsize Multiplier $k(s)$ of the baseline development cost (100 units per node). SIL~0 and SIL~4 are extrapolated. Lower row: Klosterman's ASIL~C row for SIL~3, used only in the price sweep of \cref{sec:eval-rq4}.}

\end{table}

\begin{table*}[!t]
  \caption{Pattern Catalog: Structure and Attributes of Every Record}
  \label{tab:catalog-full}
  \centering\footnotesize
  \setlength{\tabcolsep}{3pt}
\begin{tabular*}{\textwidth}{@{\extracolsep{\fill}}l >{\raggedright\arraybackslash}p{1.0in} >{\raggedright\arraybackslash}p{1.15in} >{\raggedright\arraybackslash}p{0.75in} l c c >{\raggedright\arraybackslash}p{0.95in}@{}}
\toprule
Pattern & Components (WCET, $\Delta$SIL) & Edges (tokens) & Inputs to / output from & Failure & Rec. & Mult. & Tactics \\
\midrule
none & -- & -- & -- / -- & none & 0 & -- & -- \\
Low-SIL D/C & $c$ (checker, $-1$) & $o\!\to\!c$ (0), $c\!\to\!o$ (1) & $c$ / $o$ & active & 2 & 1 & SaC, RR \\
Same-CPU D/C & $c$ (checker, $-1$) & $o\!\to\!c$ (0), $c\!\to\!o$ (1) & $c$ / $o$ & active & 1 & 1 & SaC \\
2-of-2 & $b$ (full, 0) & $o\!\to\!b$ (0), $b\!\to\!o$ (1) & $b$ / $o$ & silent & 2 & 1 (2$^{d}$) & RR, Cmp, Vot, DR \\
High-SIL isolated D/C & $c$ (checker, 0) & $o\!\to\!c$ (0), $c\!\to\!o$ (1) & $c$ / $o$ & silent & 2 & 1 & SaC, RR \\
Mixed-SIL D/C & $c_1$ (checker, $-2$), $c_2$ (checker, $-2$) & $o\!\to\!c_1$ (0), $o\!\to\!c_2$ (0), $c_1\!\to\!c_2$ (1), $c_2\!\to\!c_1$ (1), $c_1\!\to\!o$ (1) & $c_1,c_2$ / $o$ & silent & 3 & 1 & SaC, Cmp, RR, Vot, DR \\
Dual 2-of-2 & $b$ (full, 0), $s_a$ (full, 0), $s_b$ (full, 0) & $o\!\to\!b$ (0), $b\!\to\!o$ (1), $s_a\!\to\!s_b$ (0), $s_b\!\to\!s_a$ (1), $s_a\!\to\!o$ (1) & $b,s_a,s_b$ / $o$ & operational & 4 & 1 & RR, Cmp, Vot, Rep \\
1-of-2 failover & $b$ (full, 0) & $b\!\to\!o$ (1) & $b$ / $o$ & operational & 2 & 1 & RR, Rep \\
Low-SIL diverse D/C & $c$ (checker, $-1$) & $o\!\to\!c$ (0), $c\!\to\!o$ (1) & $c$ / $o$ & active & 2 & 2 & SaC, DR, RR \\
3-version (NVP)$^\dagger$ & $v_2$ (full, 0), $v_3$ (full, 0), $w$ (voter, 0) & $o\!\to\!w$ (0), $v_2\!\to\!w$ (0), $v_3\!\to\!w$ (0) & $v_2,v_3$ / $w$ & masking & 4 & 0.4/1 & RR, DR, Cmp, Vot \\
\bottomrule
\end{tabular*}

\vspace{2pt}\parbox{\textwidth}{\scriptsize Roles as in \cref{tab:catalogue}. WCET: \emph{full}, the owner's WCET (a replica); \emph{checker} or \emph{voter}, the component's own WCET entry. $\Delta$SIL: the catalog's \code{sil\_offset}, the number of levels by which the doer may be developed below the actor's SIL; the model does not use it (\cref{sec:discussion}). Edges: from, to, and initial tokens. Inputs to: components that receive copies of the owner's input channels; output from: the role whose output leaves the pattern. Failure: failure mode of the protected function. Rec.: recurring cost units. Mult.: development-cost multipliers per component in the catalog ($^{d}$when the pattern places the component \code{DIVERSE}); read but not priced by the model. Tactics after Preschern et al.: SaC sanity check, Cmp comparison, RR replication redundancy, DR diverse redundancy, Vot voting, Rep repair. $^\dagger$Separate catalog.}

\end{table*}

\begin{table*}[!t]
  \caption{General Scenarios of the Catalog Patterns}
  \label{tab:scenarios}
  \centering\footnotesize
  \setlength{\tabcolsep}{3pt}
\begin{tabular*}{\textwidth}{@{\extracolsep{\fill}}l l c l >{\raggedright\arraybackslash}p{3.45in}@{}}
\toprule
Pattern & Scen. & Covers & Tactics & General scenario (catalog text) \\
\midrule
Low-SIL D/C & SC1 & -- & Sanity Check & An implausible output of the doer is detected by the checker. \\
 & SC2 & H & Replication Redundancy & A random hardware fault in the doer's processor does not prevent detection, because the checker runs on another processor. \\
\addlinespace[1pt]
Same-CPU D/C & SC1 & -- & Sanity Check & An implausible output of the doer is detected by the checker. \\
\addlinespace[1pt]
2-of-2 & SC1 & -- & Comparison & A fault in either channel is detected as a disagreement between the two channels. \\
 & SC2 & H & Replication Redundancy & A random hardware fault in one channel does not affect the other channel. \\
 & SC3 & S & Diverse Redundancy & A systematic fault is not common to both channels. \\
 & SC4 & -- & Voting & A detected disagreement does not propagate an incorrect result to the consumer. \\
\addlinespace[1pt]
High-SIL isolated D/C & SC1 & -- & Sanity Check & An implausible output of the doer is detected by the equally-high-SIL checker. \\
 & SC2 & H & Replication Redundancy & A random hardware fault in the doer's fault containment region does not disable the checker. \\
\addlinespace[1pt]
Mixed-SIL D/C & SC1 & -- & Sanity Check & An implausible output of the low-SIL doer is detected by the high-SIL checkers. \\
 & SC2 & -- & Comparison & A fault in one checker is detected by the other checker's cross-check. \\
 & SC3 & H & Replication Redundancy & A random hardware fault affects at most one of the doer and the two checkers. \\
 & SC4 & S & Diverse Redundancy & A systematic fault in one checker is not reproduced by the other checker. \\
 & SC5 & -- & Voting & A detected disagreement does not propagate an incorrect result to the consumer. \\
\addlinespace[1pt]
Dual 2-of-2 & SC1 & -- & Comparison & A fault in a channel is detected as a disagreement within its pair. \\
 & SC2 & H & Replication Redundancy & A random hardware fault affects at most one of the four channels. \\
 & SC3 & -- & Repair & The function remains available after the primary pair has failed silent. \\
 & SC4 & -- & Voting & The failed pair's result does not reach the consumer once the standby pair has taken over. \\
\addlinespace[1pt]
1-of-2 failover & SC1 & -- & Repair & The function remains available after one channel stops. \\
\addlinespace[1pt]
Low-SIL diverse D/C & SC1 & -- & Sanity Check & An implausible output of the doer is detected by the checker. \\
 & SC2 & H & Replication Redundancy & A random hardware fault in the doer's fault containment region does not prevent detection. \\
 & SC3 & S & Diverse Redundancy & A systematic fault in the doer's implementation is not reproduced by the checker. \\
\addlinespace[1pt]
3-version (NVP)$^\dagger$ & SC1 & -- & Voting & A fault in any single version is masked by the majority of the remaining versions. \\
 & SC2 & H & Replication Redundancy & A random hardware fault affects at most one of the three versions. \\
 & SC3 & S & Diverse Redundancy & A systematic fault in one version is not reproduced by the other two versions. \\
 & SC4 & -- & Comparison & A disagreement among the versions is detected before the result leaves the component. \\
\bottomrule
\end{tabular*}

\vspace{2pt}\parbox{\textwidth}{\scriptsize Covers: the fault class a scenario needs (H random hardware, S systematic); a scenario outside the fault model becomes an undeveloped goal with the fault model as the reason. --: in scope under every fault model. \code{none} has no scenario. $^\dagger$Separate catalog.}

\end{table*}

\subsection{Problem Statement}\label{sec:bg-problem}
Given the applications with period bounds, the platform, the safety
specification, the catalog, and a cost profile, the problem is to choose a
pattern per SDF actor and, per active HSDF node, a core, a static-order
position, and an implemented SIL. The choice also covers each core's
purchase, mode, provisioned SIL, and partitioning, and the TDMA slots and
buffer sizes. The placement relations of the chosen patterns and all period
bounds must hold, and one metric is minimized under upper bounds on the
others. Patterns are not nested: they apply to the actors of the
application only, which keeps the superposed graph linear in the catalog
size.

\FloatBarrier
\section{The \safedse{} Constraint Model}\label{sec:model}

\subsection{Decision Variables and Guarded Superposition}\label{sec:model-superposition}
The model decides a pattern $\pat_a$ per SDF actor $a$, and per node $i$ a
core $\proc_i$, a successor in that core's static order, and an implemented
SIL $\sil_i$, the systematic capability of its software. Per core slot $p$,
it decides the use, mode, provisioned SIL $\csil_p$, and partitioning
$\partn_p$; per application $z$, a period $\mu_z$; and globally, the TDMA
slots and buffer sizes.

The topology is not a variable. The front-end builds the superposed graph
$\Gplus$, which holds every component of every applicable pattern
(\cref{fig:model-a}). Patterns share a component slot of an actor only if
their components use the same WCET entry (a full replica uses its owner's,
a checker its own). Node $i$ carries the set $\Gamma_i$ of patterns that
use it (all patterns for an application node). Patterns are applied before
unfolding, because a pattern is a decision about a function, not about one
iteration of it, so copy $r$ of a component belongs to copy $r$ of its
owner, $o(i)$. A component slot is itself an SDF actor of $\Gplus$; its
pattern variable is tied to its owner's, which removes one redundant
variable and one symmetry per slot. With $a(i)$ the actor of node $i$ (for
a component, the owner),
\begin{equation}\label{eq:activation}
  \begin{gathered}
  \act_i \Leftrightarrow \pat_{a(i)}\in\Gamma_i,\\
  \neg\act_i \rightarrow T_i=0 \,\wedge\, \proc_i=\proc_{o(i)}.
  \end{gathered}
\end{equation}
An inactive node is neutralized, not deleted: it takes no time, stays on its
owner's core, and no cost or memory sum counts it. Pinning matters as much
as the zero time; an unpinned inactive node could take any core and
multiply every solution by the number of cores. The graph thus has a fixed
shape, and the order and throughput constraints are posted once. The
pattern of each actor is restricted to the applicable ones, which the
front-end computes from the SIL ranges, the fault model, and the platform.

The execution time of a computation node follows its binding and mode,
\begin{equation}\label{eq:wcet}
  T_i=\act_i\,W\bigl(i,\ctype_{\proc_i},\mathit{mode}_{\proc_i}\bigr),
\end{equation}
where $W$ is the WCET table; a missing entry is a forbidden value of
$\proc_i$. The nested lookup looks expensive, but a flattened encoding with
one reified implication per node and core was never faster in the
development measurements, so the model keeps the direct form.

A negative result motivated guarded superposition. A variable adjacency
matrix over all possible replicas needs a quadratic number of topology
variables and an \code{element} constraint on a variable index for every
access. In the early analysis of the tool, it did not solve a five-actor
instance within minutes. With a fixed graph, only the selection is a
variable, and every guard is a lookup into a constant table.

\emph{Explicit voters.} Most catalog patterns fold the comparison into the
owner (\cref{tab:catalog-full}). A pattern with an explicit voter, such as
N-version programming, routes the owner's output through the voter. The
front-end re-emits every channel that leaves the owner from the voter slot,
guarded by the voter patterns. The original channel survives only under the
other patterns. An edge into a downstream pattern can then depend on
the patterns of two actors (the voter pattern of its source and the input
fan-out of its destination), so a guarded edge $e$ carries two conditions,
\begin{equation}\label{eq:guard2}
  \mathit{on}_e \Leftrightarrow
  \pat_{a_1(e)}\in\Gamma^{1}_e \wedge \pat_{a_2(e)}\in\Gamma^{2}_e .
\end{equation}
An edge with one condition repeats its owner as the second one with an
all-true row, so the conjunction degenerates exactly and no sentinel index
is needed. Dropping either condition would let the pre-voter result reach
a downstream checker. The guarded edge constraints themselves turned out to
be redundant: the front-end also writes every pattern channel into the
token matrix, which \cref{eq:pot} posts unconditionally, and zero execution
time of inactive nodes is what neutralizes their edges
(\cref{sec:model-throughput}). They are kept because this redundancy
depends on how the front-end fills the token matrix.

\begin{figure}[!t]
  \centering
  \subfloat[]{\includegraphics{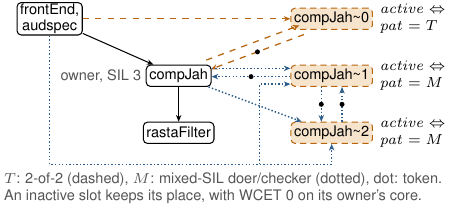}\label{fig:model-a}}\\[-2pt]
  \subfloat[]{\includegraphics{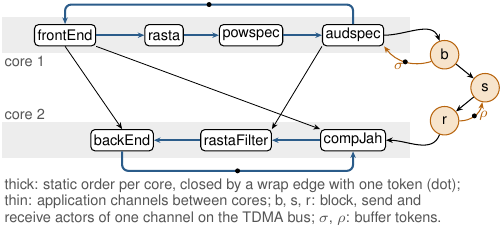}\label{fig:model-b}}
  \caption{(a) Guarded superposition of \code{compJah} under both fault
  classes: each slot is guarded by the patterns that use it, the graph is
  fixed, and only $\pat$ is a variable. (b) MSAG of a two-core binding of
  the running example, with one remote channel refined for the TDMA bus.
  Not drawn: the one-token self-edge of every actor and the refinement of
  the other remote channels.}
  \label{fig:model}
\end{figure}

\subsection{Communication Refinement}\label{sec:model-comm}
Communication refinement uses the same mechanism. For each application
channel $c=(u,v)$, $\Gplus$ holds a block, a send, and a receive actor.
They take zero time unless the channel is remote, which the model defines as
\begin{equation}\label{eq:remote}
  \mathit{remote}_c \Leftrightarrow
  \act_u \wedge \act_v \wedge \proc_u\neq\proc_v .
\end{equation}
A remote channel replaces its direct edge by the path through these actors
(\cref{fig:model-b}); the front-end removes the direct edge from the token
matrix, so exactly one of the two forms is posted. A direct edge kept next
to the path would be a shortcut past the whole transfer delay. Block and
send run on the sending core's network interface and receive in the
receiving core's memory, so their placement is derived:
$\proc_{\mathit{blk}(c)}=\proc_{\mathit{snd}(c)}=\proc_u$ and
$\proc_{\mathit{rcv}(c)}=\proc_v$.

The times follow \desyde{}'s TDMA model~\cite{rosvall2018flexible}. With
$S$ slots per round of cycle length $L$, $a_p$ slots allocated to the
sending core $p$, and $n_c$ slots needed for the message of $c$,
\begin{equation}\label{eq:tdma}
  \begin{aligned}
  T_{\mathit{blk}(c)}&=\bigl(S-(a_p-1)\bigr)L,\\
  T_{\mathit{snd}(c)}&=\bigl(\rho_c(a_p)\,S+S\bigr)L,\qquad
  T_{\mathit{rcv}(c)}=0,
  \end{aligned}
\end{equation}
where $\rho_c(s)=\lceil n_c/s\rceil-1$ is the number of full extra rounds.
Blocking is the worst-case wait for the core's slot, and sending is the
transfer in whole rounds. Receiving takes no time because the token lands in
local memory; the receive actor still exists because its buffer closes a
cycle. The slots of all cores sum to at most $S$, a core has slots
iff it sends, and it has no more slots than its largest outgoing message
needs, a restriction that \desyde{} also posts.

The path of a remote channel $c=(u,v)$ in the MSAG is
\begin{equation}\label{eq:commpath}
  \begin{gathered}
  u\xrightarrow{\tok_{uv}}\mathit{blk}\xrightarrow{0}\mathit{snd}
  \xrightarrow{0}\mathit{rcv}\xrightarrow{0}v,\\
  \mathit{blk}\xrightarrow{\beta^{s}_c}u,\qquad
  \mathit{rcv}\xrightarrow{\beta^{r}_c}\mathit{snd},
  \end{gathered}
\end{equation}
with the channel's initial tokens on the first edge, so the pipeline depth
is kept. The send and receive buffer sizes $\beta^{s}_c,\beta^{r}_c\ge1$
are the tokens on the back-edges; a local channel has $\beta^{s}_c=0$ and
$\beta^{r}_c=1$. The send buffers of a core must fit its network interface.
Because the buffer sizes are variables, their back-edges in \cref{eq:pot}
multiply two variables, which \cpsat{} handles natively.

Communication actors are scheduled on the bus, not in a core's static order,
and they are not counted in a core's load. A transfer overlaps computation;
it lengthens the latency, and the period only when the bus is the
bottleneck. Contention for the bus is already priced into
\cref{eq:tdma}, so this exclusion does not make the model optimistic. An
early version that put them into the static order reported wrong periods;
\cref{sec:lessons} returns to it.

Only application channels are refined (\code{--comm-scope app}).
Cross-checks between pattern components are assumed to use a dedicated
link; \code{--comm-scope all} refines them too, and RQ4 measures the effect
(\cref{sec:eval-rq4}). Whether a transfer is a safety function is set by
\code{--comm-sil}. The default \code{exempt} leaves communication actors
out of the SIL-mixing rule. \code{inherit} gives each transfer at least its
sender's SIL, and \code{core} subjects it to the rule of the core it runs
on, like any software there. A replica channel pair can also be required to
leave from different cores, which separates their TDMA slots; on a single
bus this is weaker than redundant communication, and the argument says so.

\subsection{Safety and Placement Constraints}\label{sec:model-safety}
A core's SIL is bounded by its type, only a type $t$ with certified
partitioning ($\mathit{cp}_t$) may be partitioned, and every active node
meets its actor's required SIL:
\begin{equation}\label{eq:sil}
  \begin{gathered}
  \csil_p\le\silmax_{\ctype_p},\qquad
  \partn_p\rightarrow\mathit{cp}_{\ctype_p},\\
  \act_i \rightarrow \silreq_{a(i)}\le\sil_i\le\csil_{\proc_i}.
  \end{gathered}
\end{equation}
An unused core has $\csil_p=0$ and no partitioning, and an inactive node
has $\sil_i=0$. The IEC~61508 rule on mixing SILs becomes, for every active
computation node (communication actors are exempt by default),
\begin{equation}\label{eq:rule}
  \act_i \wedge \neg\partn_{\proc_i} \rightarrow \sil_i=\csil_{\proc_i},
\end{equation}
and $\csil_p$ equals the highest $\sil_i$ of an active, non-exempt node on a
used core $p$. So all nodes on a core without partitioning share that SIL,
and a partitioned core pays the partitioning cost of its type. With the
baseline development cost $b_i$ and the profile's multiplier $k$, the
development cost and its part spent on promotion (included in
$C_{\mathrm{dev}}$) are
\begin{equation}\label{eq:devcost}
  \begin{aligned}
  C_{\mathrm{dev}}&=\textstyle\sum_{i}\act_i\, b_i\, k(\sil_i),\\
  C_{\mathrm{pro}}&=\textstyle\sum_{i}\act_i\, b_i
  \bigl(k(\sil_i)-k(\silreq_{a(i)})\bigr).
  \end{aligned}
\end{equation}
The model stores $100\,k$ and divides each sum by 100, so all arithmetic
stays integral. An input flag forbids promotion: every active node is then
implemented exactly at its required SIL (\cref{sec:eval-rq2}).

For each placement relation of an applicable pattern whose fault class is in
the fault model, the front-end emits a tuple $(u,v,R,a,\Gamma)$, and the
model posts
\begin{equation}\label{eq:placement}
  \pat_a\in\Gamma \rightarrow R(\proc_u,\proc_v).
\end{equation}
Relation $R$ is one of $\proc_u=\proc_v$ (\code{SAME}), $\proc_u\neq\proc_v$
(\code{DIFFERENT}), $\fcr_{\proc_u}\neq\fcr_{\proc_v}$
(\code{DIFFERENT\_FCR}), or $\ctype_{\proc_u}\neq\ctype_{\proc_v}$
(\code{DIVERSE}). As $\fcr$ and $\ctype$ are constant arrays, each relation
is an \code{element} constraint. Relations hold per HSDF copy: copy $r$ of a
checker avoids copy $r$ of its doer. The set-wise reading (no copy of the
checker on a core of any copy of the doer) is stronger and forbids safe
interleavings. Under a single fault the pairwise reading suffices: a fault
on a core with doer copy $r+1$ and checker copy $r$ is caught by checker
copy $r+1$ elsewhere.

\subsection{Static Order}\label{sec:model-order}
Each core runs its computation nodes in one static order, encoded by a
successor $\mathit{succ}_i\in\{0,\dots,n\}$ ($0$: last on its core), a
position $\mathit{ord}_i$, and a head flag $h_i$:
\begin{equation}\label{eq:order}
  \begin{gathered}
  \mathit{succ}_i=j>0 \rightarrow \proc_j=\proc_i \wedge
  \mathit{ord}_j=\mathit{ord}_i+1,\\
  h_j \Leftrightarrow \nexists i:\mathit{succ}_i=j,\qquad
  h_j\rightarrow\mathit{ord}_j=0,
  \end{gathered}
\end{equation}
with all non-zero successors distinct, no node its own successor, distinct
positions on one core, and exactly one last node per used core. The wrap
edge runs from the last node of a core to its head. This form replaces
Rosvall's encoding with a dummy node and a \code{circuit}
constraint~\cite{rosvall2018flexible}, a Gecode idiom that \mzn{}
decomposes anyway. The anchor $h_j\rightarrow\mathit{ord}_j=0$ is not
cosmetic: without it the positions float, the head cannot be identified,
and the wrap edge, and with it the bound ``period at least the core's
load'', silently disappears. A core is used iff an active node is on it,
and a used core requires its card. Deadlock is caught by the potentials of
\cref{sec:model-throughput}, so the positions only have to make each static
order acyclic.

\subsection{Throughput by Node Potentials}\label{sec:model-throughput}
The MSAG~\cite{rosvall2018flexible} of a solution adds to $\Gplus$ a
token-free edge from each node to its successor and a one-token wrap edge
from the last to the first node of each core. Remote channels add their
paths \cref{eq:commpath}, whose back-edges carry the buffer sizes as tokens.
Each computation node has a one-token self-edge and the time
\cref{eq:wcet}. With an integer potential $\pot_i$ per node, the model
posts for every MSAG edge $(i,j)$, with $z(i)$ the application of $i$,
\begin{equation}\label{eq:pot}
  \pot_j \ge \pot_i + T_i - \tok_{ij}\,\mu_{z(i)} .
\end{equation}
Order, wrap, and channel edges are posted under reified conditions on the
successor and $\mathit{remote}_c$ variables. One potential is anchored at
zero, which removes the free additive constant.

\begin{proposition}\label{prop:potentials}
Fix all decisions except $\mu$ and $\pot$. Potentials that satisfy
\cref{eq:pot} on every MSAG edge exist iff every cycle $C$ of the MSAG
satisfies $\sum_{i\in C}T_i\le\mu\sum_{(i,j)\in C}\tok_{ij}$.
\end{proposition}
\begin{IEEEproof}
If $\pot$ is feasible, summing \cref{eq:pot} along $C$ cancels the
potentials and yields the inequality. Conversely, weight each edge with
$w_{ij}=T_i-\mu\,\tok_{ij}$; no cycle has positive weight. Add a source with
a zero-weight edge to every node. The longest-path lengths $\ell_j$ from it
are then finite and satisfy $\ell_j\ge\ell_i+w_{ij}$ on every edge, so
$\pot=\ell$ is feasible; it is integral for integral $T$, $\tok$, and $\mu$.
As $0\le\ell_j\le\sum_iT_i$, the bounded potential domain of the model
holds $\ell$ after anchoring. \Cref{app:proof} gives the proof in full,
with its corollaries.
\end{IEEEproof}

\Cref{prop:potentials} is the linear-programming form of the cycle ratio
problem, with the potentials as node distances~\cite{dasdan2004experimental},
as in cyclic scheduling~\cite{hanen1995study}. Minimizing $\mu$ thus yields
the period of \cref{def:period} without a throughput
propagator~\cite{rosvall2017throughput}, which a \mzn{} model cannot host;
as $\mu$ is an integer, the model reports $\lceil\mathrm{MCR}\rceil$. A
token-free cycle with positive execution time makes \cref{eq:pot}
infeasible, so the same inequalities replace separate deadlock machinery,
such as Rosvall's rank constraints or a test for positive throughput. As
remote buffers hold a token and every WCET in our inputs is at least one,
each token-free cycle of active nodes has positive time, so deadlock is
excluded. The encoding needs one variable per node, against a quadratic
number for an edge-based residual formulation. For fixed token counts, the
inequality is linear; only the buffer back-edges are products of two
variables.

Neutralized nodes keep their channels with zero execution time. A
token-free cycle through them has zero weight and does not constrain $\mu$.
A path through an inactive slot, however, can raise the least $\mu$ of a
design above its MCR, never below it. A channel of an unselected pattern
into an inactive slot, followed by the slot's static-order edge, is a path
that the deployed design does not have (\cref{app:proof}, Corollary~4). The
model thus stays sound but can be pessimistic for a design, and under a
period bound it can reject a design that meets the bound. We wrote two
exact variants. Variant~A posts pattern channels only under their guard and
changes the verifier in step. Variant~B adds a second potential for the
static-order and wrap edges, equal to $\pot$ on active nodes, and keeps
under the guard only the channels whose guard differs from the activity of
their endpoints. Both remove the gap on the affected designs but solve
more slowly. Variant~B served as a cross-check of every RQ2 and RQ4 result
(\cref{sec:eval-rq2}), with a median run-time ratio of
\resExactTimeRatioMedian{} and a heavy tail at loose bounds; the default
model keeps the faster form. For other objectives than the period, the
model guarantees only $\mathrm{MCR}\le\mu\le$ the period bound, and the
evaluation fixes the design and minimizes $\mu$ afterwards
(\cref{sec:eval-setup}).

Rosvall et al.~\cite{rosvall2017throughput} gave each application the MCR of
its connected MSAG component. Applications that share a core lie in one
component and share its period, since the core repeats one static order in
which each of their nodes fires once. The model equates their periods
pairwise:
\begin{equation}\label{eq:coupling}
  \mathit{sh}_{y,z}\Leftrightarrow\exists p:\ \mathit{on}_{y,p}\wedge
  \mathit{on}_{z,p},\qquad
  \mathit{sh}_{y,z}\rightarrow\mu_y=\mu_z,
\end{equation}
where $\mathit{on}_{z,p}$ holds iff an active node of $z$ is on core $p$. By
transitivity, each component gets one period, and its potentials share one
time base, which is the condition for summing \cref{eq:pot} along a cycle
that crosses applications. \desyde{} states the coupling over pairs of
actors of different applications, a quadratic number of reified
implications. In an early version of \safedse{}, this timed out at 32
actors; the occupancy matrix $\mathit{on}$ makes it
$O(|\mathit{APP}|^2+n|\mathit{APP}|)$ and is equivalent.

Four implied constraints from \desyde{}~\cite{rosvall2018flexible} prune at
the root, although the potentials imply them only after search. With the
load $\lambda_p$ of core $p$ over its computation nodes,
\begin{equation}\label{eq:load}
  \lambda_p=\textstyle\sum_{j}[\proc_j=p]\,T_j,\qquad
  \proc_i=p \rightarrow \mu_{z(i)}\ge\lambda_p .
\end{equation}
Second, an active node's time is at most its application's period bound.
Third, two active nodes of one application whose least WCETs sum to more
than its period bound share no core and are not adjacent in a static order.
The activity condition matters, because an inactive slot is pinned to its
owner's core (\cref{sec:lessons}). Fourth, an application uses at least
$\lceil\sum_i\min_t W(i,t)/B_z\rceil$ distinct cores, with the sum over
its always-active nodes and $B_z$ its period bound. The wrap edge already
implies \cref{eq:load}, but only through a cycle that the solver must
find; stated as a sum, it propagates at once.

\subsection{Further Metrics and Constraints}\label{sec:model-other}
The memory of the active nodes on a core is bounded by its type's memory,
\begin{equation}\label{eq:memory}
  \textstyle\sum_{i}[\proc_i=p\wedge\act_i]\,m_i\le M_{\ctype_p}.
\end{equation}
The hardware cost sums the prices of the used cards and of the used cores
in their modes, and the power sums the power of the used cores in their
modes. A latency bound between a source and a sink actor uses the
potentials again: in the periodic regime, the $k$-th firing of node $i$
starts at about $\pot_i+k\mu$, so
\begin{equation}\label{eq:latency}
  \mathit{lat}(s,d)=\pot_d-\pot_s+\rho(s,d)\,\mu+T_d ,
\end{equation}
with $\rho(s,d)$ the least number of initial tokens on a path from $s$ to
$d$, computed by the front-end. The estimate is exact in the periodic phase
but ignores the transient before it, so it serves as a cheap filter; a
separate checker simulates the transient, and a reported latency is not a
worst-case bound without it.

\subsection{Objectives and Solving}\label{sec:model-solving}
The model computes eight metrics: the worst period, the largest latency,
the hardware, development, and total cost, the power, the number of cores,
and the promotion cost. One metric is minimized, chosen by a parameter, and
every metric has an upper bound, so switching the objective needs no model
edit. The total cost
\begin{equation}\label{eq:totalcost}
  C_{\mathrm{tot}}=C_{\mathrm{hw}}+C_{\mathrm{dev}}+C_{\mathrm{part}}
  +C_{\mathrm{pat}}
\end{equation}
adds the price of cards and cores in their modes, \cref{eq:devcost}, the
cost of each partitioned core, and the recurring cost of the selected
patterns, summed over the actors of the applications only, all in one cost
unit.

Symmetry breaking orders only interchangeable cores. Cores of one type on
one card are interchangeable, so a value-precedence constraint orders their
use. Cores on different cards are not, once FCR relations exist; only whole
cards of the same template are, so the model buys a prefix of them and
orders successive identical cards by their node membership. Cores of one
type that are each alone on their card are interchangeable outright and get
value precedence too. The order of Rosvall and
Sander~\cite{rosvall2018flexible} on the HSDF copies of an SDF actor (copies
ordered by core index, and by position on a shared core) acts only on
multi-rate actors, which our pattern instances lack. With per-copy
placement relations, it could exclude an optimum in principle; in our
probes with multi-rate owners, it pruned designs but moved no optimum.

The model is written in \mzn{}~\cite{nethercote2007minizinc} and solved
with \cpsat{}~\cite{perron2026cpsat}. The search annotation branches on the
binding (first fail, smallest value first), then on the successors, then on
the periods (smallest value first). With one worker, \cpsat{} follows this
annotation; with several, its portfolio does not, which is why the period
of a non-period optimum is not pinned to its least value without the extra
step of \cref{sec:eval-setup}. The model also has a two-step strategy
inherited from Rosvall and Sander: step~1 forbids applications to share a
core, and its objective value bounds the full model in step~2, which stays
complete. With four workers it is no longer needed on the multi-application
benchmark (\cref{sec:eval-rq1}).

\FloatBarrier
\section{Independent Verification and Safety-Argument Generation}\label{sec:verification}

\subsection{Independent Verifier}\label{sec:ver-verifier}
A wrong constraint does not fail; it returns a plausible number. A verifier
that shares no code with the model therefore re-checks every solution. It
reads the front-end's instance and the solution, recomputes every field
that the model derives, and appends one record to a check log per check.
No derived field of the solution is taken on trust. Each record has a kind,
a pass flag, and a detail. An \code{input} record checks that every
required field is present, so a missing field is a failed record, not a
crash. A \code{wcet} record checks that every node time equals the table
entry of its core type and mode, zero when inactive, and that no binding is
forbidden. An \code{order} record checks one acyclic static order per used
core, over exactly the nodes bound to it. An \code{activation} record checks
that activation follows the selected pattern, that the pattern is admissible
for the actor, and that an actor's slots share it. The \code{memory} and
\code{cost} records recompute memory use and the hardware, development,
promotion, partitioning, and pattern cost from the instance tables. A
\code{sil\_actor} record per actor and per component checks that the
implemented SIL is at least the required one and at most the core's. An
\code{isolation} record per core checks the SIL ceiling of the type, the
SIL-mixing rule, and partitioning only on a type that provides it. Each
relation instance of the selected patterns gets a \code{placement} record.
Finally, the \code{period} record checks the period as described next, and
the \code{load\_bound} record that it is at least the load of the busiest
core.

For the period, the verifier rebuilds the MSAG of the deployed design. It
removes inactive nodes and splices each static order around them:
consecutive live nodes get a token-free edge, and the last live node a
one-token wrap edge to the first. A refined channel that stays local gets
back its direct edge with its initial tokens; a remote one gets the path
\cref{eq:commpath}, with the tokens on its first edge and the two buffer
back-edges. A token-free cycle in this graph is a deadlock and fails the
record. For each connected component (applications that share a core share
a period, \cref{eq:coupling}), the verifier computes the MCR twice. The
first oracle is Lawler's parametric search~\cite{dasdan2004experimental}:
bisection on $\mu$ with a Bellman-Ford test for a positive cycle, then the
rational with the smallest denominator in the final interval. The second
fires the actors in max-plus time, i.e., it simulates the self-timed
execution, until the differences of the start times repeat over a window of
up to eight iterations. The two share no reasoning: one is a cycle analysis,
the other a simulation. Both must agree, and the reported period must equal
$\lceil\mathrm{MCR}\rceil$; for other objectives than the period, we first
minimize $\mu$ for the fixed design (\cref{sec:eval-setup}). If a latency
bound is set, the verifier also simulates the transient and flags a worst
case above the model's periodic estimate \cref{eq:latency}.

The log records what was checked, not what is true. The generator treats a
missing record like a failed one, so ``nobody looked'' differs from
``looked and found it fine''. The verifier reads the front-end's instance,
so front-end and specification errors are common to model and verifier, as
every argument states in an assumption. It is an unqualified T2 tool in the
sense of IEC~61508~\cite{iec61508part4}. Its first version took WCETs,
static orders, activation, and costs from the solution; a mutation analysis
(\cref{sec:eval-rq4}) showed that it then accepted corrupted solutions, and
the current version re-derives all of them.

\subsection{Evidence-Backed Pattern Arguments}\label{sec:ver-gsn}
Preschern et al.~\cite{preschern2013building} attached a
GSN~\cite{gsn2021standard} fragment to each pattern: its general scenarios
are goals, the tactics that achieve a scenario are strategies below it, and
IEC~61508 methods realize the tactics. A scenario that depends on another
is nested below it. Their fragments are per pattern and generic; our
generator instantiates them for a solved and verified deployment, which is
what makes them checkable.

The top goal claims that the deployed architecture meets, for every
safety-related actor, the architectural preconditions that the catalog
states for its SIL. Its contexts name the standard and the meaning of
SIL~0, the hazard analysis as the source of the required SILs, and the
catalog with its sources. Further contexts give the applications, the
platform with the instantiated cores, the fault model with the excluded
fault classes and the meaning of S, and the SHA-256 hashes of instance,
solution, and check log. A
justification limits the argument to the architecture and its deployment,
and three assumptions cover sound WCETs, token preservation, and the shared
front-end. A strategy then argues over each actor with
$\silreq_a\ge1$ and over the platform-wide properties.

Per actor, the goal claims that the deployment meets the catalog's
preconditions for $\silreq_a$; no goal claims an attained SIL. Contexts give
the HSDF copies with their cores and SILs, and the components of the
selected pattern with theirs. A provisioning goal has a Solution for the
owner and one per component, each citing a \code{sil\_actor} record. A
strategy ``by application of'' the selected pattern has one goal per
general scenario (\cref{tab:scenarios}) and below it one strategy per
tactic. A tactic table, transcribed from Preschern et al.'s appendix, links
each tactic to IEC~61508-7 methods and to the placement relations whose
records discharge its deployment precondition (\cref{tab:tactics}). It
links Replication Redundancy to \code{DIFFERENT\_FCR} or \code{DIFFERENT},
Diverse Redundancy to \code{DIVERSE}, and value-domain tactics (Sanity
Check, Condition Monitoring, Comparison, Voting, Override, Repair) to none.
This column is the exact boundary between what a mapping can and cannot
establish, and the one to be most suspicious of. A relation listed there
that the model does not post would turn a hand-wave into a claim.

For a tactic with relations, a deployment goal claims exactly what the
relation establishes. It claims ``the replicated components run in different fault
containment regions'' for \code{DIFFERENT\_FCR}, ``on different cores,
possibly in one fault containment region'' for \code{DIFFERENT}, and ``the
diverse components run on different core types'' for \code{DIVERSE}. Its
Solutions cite the \code{placement} records of that relation by index
(\cref{fig:gsn,fig:gsn-compjah}). An undeveloped goal next to a replication
goal covers the common causes that remain: the interconnect and cross-check
link, the shared inputs, the comparison or voting point, and for
\code{DIFFERENT} also the shared card with its power and clock. The goal
that an IEC~61508 method realizes the tactic, and is effective for the
faults of concern, stays undeveloped with the reason ``value-domain or
implementation evidence; the mapping cannot establish this'' and the
candidate methods as context. A \code{DIVERSE} goal carries the assumption
that distinct core types give distinct toolchains and object code, not an
independent design. A scenario whose fault class lies outside the fault
model becomes an undeveloped goal with the fault model as the reason. A
placement relation that no tactic consumes (e.g., \code{SAME}, which
confers no independence) appears as a goal ``remaining placement
constraints'' with its records.

Four platform-wide goals follow. The isolation goal claims that no core
hosts software of differing SIL without certified partitioning, also from
SIL-0 actors. It is argued by the Barrier tactic with a Solution citing the
\code{isolation} records, under the assumption that cores interfere only
through the interconnect; if a core is partitioned, an undeveloped goal
claims that the mechanism is certified. The timing goal has two Solutions,
the \code{period} and the \code{load\_bound} records. The development
process goal stays undeveloped, and so does the goal that each detected
fault brings the function to its safe state within the process safety
time, with the verdict tokens of the selected patterns as context. With \code{--comm-sil exempt}, an assumption
states the black-channel condition for transfers.

\begin{definition}[Well-formed argument]\label{def:argument}
Every leaf is a Solution that cites at least one passed record of the check
log of the argued solution, or a goal marked undeveloped with a stated
reason.
\end{definition}

\Cref{def:argument} is structural: nothing unchecked appears as checked,
but the argument may still be inadequate (\cref{sec:related-gsn}). The
generator's audit enforces it, including that every cited record exists and
passed, and the generator refuses (exit code~3) to emit an argument that
violates it. It also refuses if the verifier rejected the solution (exit
code~2), if the catalog lacks a pattern selected in the instance (exit
code~4), or if a given fault model contradicts the one it infers (exit
code~2). The fault model is inferred, not trusted from a flag. The
admissible patterns per actor are a function of the fault model and the
required SILs, so the generator tries all three fault models and keeps those
that reproduce them. If none or several fit, it names none, and the
out-of-scope analysis is unavailable. A wrong flag would otherwise yield a
fluent argument that claims coverage the design never had. The catalog
check exists because of an observed failure: against a catalog without the
NVP record, an early generator version produced a well-formed,
audit-passing argument that argued an actor with three-version programming
``by development process alone''.

\begin{figure}[!t]
  \centering
  \includegraphics{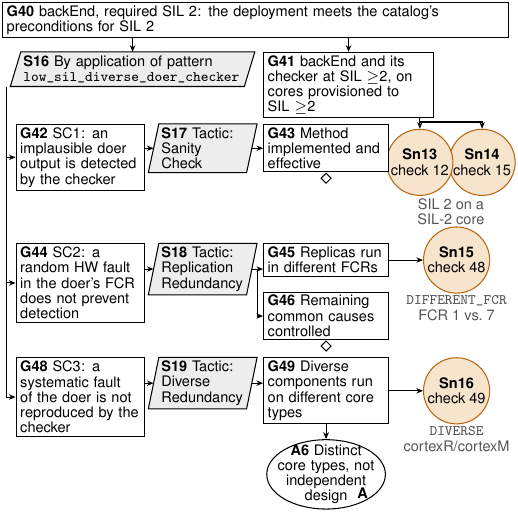}
  \caption{Excerpt of the argument generated for a verified optimum of the
  running example under both fault classes: actor \code{backEnd} (SIL~2,
  low-SIL diverse doer/checker), ids as generated, texts shortened.
  Deployment goals claim exactly what their log records establish; the
  remaining common causes and the method goals (diamonds) stay with the
  engineer. The method goal is drawn once.}
  \label{fig:gsn}
\end{figure}

\begin{figure}[!t]
  \centering
  \includegraphics{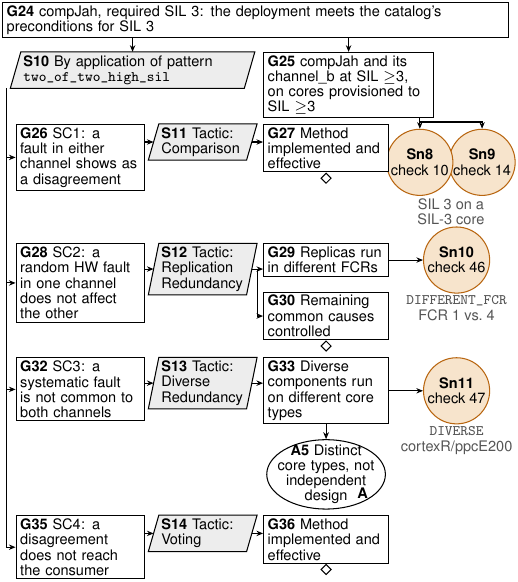}
  \caption{Second excerpt of the same argument: actor \code{compJah} (SIL~3,
  2-of-2). The replication goal cites the \code{DIFFERENT\_FCR} record and
  the diversity goal the \code{DIVERSE} record of the same pair; Comparison
  and Voting are value-domain tactics and stay undeveloped. Generated by the
  script of \cref{fig:gsn}.}
  \label{fig:gsn-compjah}
\end{figure}

\begin{table*}[!t]
  \caption{Tactic Table: Tactics, IEC 61508-7 Methods, and Deployment Relations}
  \label{tab:tactics}
  \centering\footnotesize
  \setlength{\tabcolsep}{3pt}
\begin{tabular*}{\textwidth}{@{\extracolsep{\fill}}l l >{\raggedright\arraybackslash}p{1.75in} l >{\raggedright\arraybackslash}p{1.2in}@{}}
\toprule
Tactic & Category & IEC 61508-7 methods & Relations & Patterns using it \\
\midrule
Sanity Check & Checking & A.1.2 monitoring relay contacts; A.2.7 analog signal monitoring; A.3.1-A.3.3 self-tests; A.4.1-A.4.4 checksums; A.5.1-A.5.5 RAM tests; A.6.1 test pattern; A.9 temporal and logical program monitoring; C.3.3 assertion programming; C.5.3 interface checking & -- & Low-SIL D/C, Same-CPU D/C, High-SIL isolated D/C, Mixed-SIL D/C, Low-SIL diverse D/C \\
Condition Monitoring & Checking & A.1.1 failure detection by online monitoring; A.6.4 monitored outputs; A.9 temporal and logical program monitoring; A.13.1 monitoring & -- & -- \\
Comparison & Checking & A.1.3 comparator; A.6.5 input comparison/voting & -- & 2-of-2, Mixed-SIL D/C, Dual 2-of-2, 3-version (NVP) \\
Replication Redundancy & Redundancy & A.2.1 tests by redundant hardware; A.2.5 monitored redundancy; A.3.5 reciprocal comparison by software; A.4.5 block replication; A.6.3 multi-channel output; A.7.3 complete hardware redundancy; A.7.5 transmission redundancy & \code{DIFFERENT\_FCR}, \code{DIFFERENT} & Low-SIL D/C, 2-of-2, High-SIL isolated D/C, Mixed-SIL D/C, Dual 2-of-2, 1-of-2 failover, Low-SIL diverse D/C, 3-version (NVP) \\
Diverse Redundancy & Redundancy & A.7.6 information redundancy; A.13.2 cross-monitoring of multiple actuators; B.1.4 diverse hardware; C.4.4 diverse programming & \code{DIVERSE} & 2-of-2, Mixed-SIL D/C, Low-SIL diverse D/C, 3-version (NVP) \\
Voting & Masking & A.1.4 voter; A.6.5 input comparison/voting & -- & 2-of-2, Mixed-SIL D/C, Dual 2-of-2, 3-version (NVP) \\
Override & Masking & IEC 61508 fail-safe principle; A.1.3 comparator & -- & -- \\
Repair & Recovery & C.3.9 error correction; C.3.10 dynamic reconfiguration & -- & Dual 2-of-2, 1-of-2 failover \\
Degradation & Recovery & A.8 voltage supply error handling; C.3.8 degraded function limitation & -- & -- \\
Barrier & Isolation & A.11 separation of energy lines from information lines; B.1.3 separation of safety functions from non-safety functions; B.3.4 modularisation; C.2.8 information hiding/encapsulation; C.3.11 time-triggered architecture & -- & -- \\
Simplicity & Avoidance & B.2.1 structured specification; B.3.2 structured design; C.2.7 structured programming; C.4.2 programming language subset & -- & -- \\
Substitution & Avoidance & B.3.3 usage of well-proven components; B.5.4 field experience; C.2.10 usage of well-proven/verified software elements; C.4.3 certified tools and compilers & -- & -- \\
\bottomrule
\end{tabular*}

\vspace{2pt}\parbox{\textwidth}{\scriptsize Methods as listed by Preschern et al.\ (annex clause of IEC~61508-7 and its title). Relations: the placement relations whose verifier records discharge the tactic's deployment precondition (\code{deployment\_relations} of the tactic table); a tactic without one is realized in the value domain and stays undeveloped.}

\end{table*}

\section{Experimental Evaluation}\label{sec:eval}

\subsection{Setup}\label{sec:eval-setup}
All runs used one machine (\resSetupCpu, \resSetupLogicalCpus{} logical
CPUs, \resSetupRamGiB\,GiB), \mzn{}~\resSetupMiniZinc{}, and
\cpsat{}~\resSetupCpSat{} with \resSetupThreads{} workers, one per physical
core, and a time limit of \resSetupTimeLimit{}\,s per solve unless stated
otherwise. When $\mu$ is not the objective, an optimum may report any
period between the MCR and the bound (\cref{sec:model-throughput}), so we
then fix the design and minimize $\mu$ in a wrapper model (the period step).
Every solution below passed the verifier; a run that did not is left out
of every number and named where it matters. Optima are often not unique
(on \resPatternTieInstances{} instances, two runs chose different SIL-3
patterns at equal cost), so we report cost, and periods only where the
period is the objective.

\Cref{tab:results} lists the pattern instances. They include the running
example on platforms with two to four core types (one with certified
partitioning), the fault model changed, two with the TDMA bus,
Sobel~\cite{rosvall2018flexible} with a SIL-3 actor, and three applications
together; \cref{app:instances} gives their inputs. WCETs on the card
platforms are the benchmarks' nominal times scaled per core mode, and a
checker or voter takes \resSetupCheckerScale\,\% of its owner's. All nodes
have the same baseline development cost, and prices are illustrative
(\cref{tab:platforms,tab:profiles}). The joint optimum of every row of
\cref{tab:results} was solved \resRepsPerRow{} times. All \resRepsRuns{} runs of this kind
returned the same optimum per row, and the run time varied by a median of \resRepsMedianSpread\,\% and at most
\resRepsMaxSpread\,\% around the median of a group. Front-end, period step,
verifier, and generator take at most \resPipePinMax{}\,s per instance; the
solve dominates (\cref{tab:pipe}).

\begin{table}[!t]
  \caption{Pattern Instances and Joint Optima (All Solutions Verified)}
  \label{tab:results}
  \centering\footnotesize
  \setlength{\tabcolsep}{1.5pt}
\begin{tabular*}{\columnwidth}{@{\extracolsep{\fill}}l l c r r r c r r@{}}
\toprule
Instance & Platform & FM & $|\Gplus|$ & Vars & Cost & Cores & $\mu^{*}$ & $t$ (s) \\
\midrule
\code{p\_rasta} & 2T, part. & H & 26 & 11k & 3330 & 2 & 235 & 6.0 \\
\code{f\_random\_hw} & 2T & H & 26 & 11k & 3349 & 5 & 235 & 8.7 \\
\code{d\_random\_hw}$^{a}$ & 2T & H & 18 & 7k & 1857 & 3 & 235 & 6.0 \\
\code{d\_systematic\_sw}$^{a}$ & 2T & S & 14 & 5k & 1846 & 2 & 235 & 3.3 \\
\code{f\_sw3} & 3T & S & 18 & 8k & 3365 & 4 & 235 & 5.2 \\
\code{f\_both3} & 3T & H+S & 18 & 8k & 3406 & 5 & 235 & 3.2 \\
\code{c\_pat} & 2T, TDMA & H & 71 & 15k & 3349 & 5 & 235 & 14.3 \\
\code{pc\_sobel}$^{b}$ & 2T, TDMA & H & 57 & 5k & 1401 & 3 & 333 & 3.0 \\
\code{v\_nvp}$^{c}$ & 4T & H+S & 18 & 9k & 5680 & 6 & 235 & 4.9 \\
\code{m\_3app}$^{d}$ & 2T & H & 44 & 22k & 5665 & 6 & -- & 44.0 \\
\bottomrule
\end{tabular*}

  \vspace{2pt}\parbox{\columnwidth}{\scriptsize Application RASTA unless
  noted; $n$T: $n$ core types; part.: certified partitioning; FM: fault
  model (H random hardware, S systematic). $|\Gplus|$: nodes of the
  superposed graph incl.\ communication actors; Vars: FlatZinc variables;
  Cost: proven minimum of $C_{\mathrm{tot}}$ of the cost model; $\mu^{*}$:
  minimum period; $t$: median solve time of the cost optimum over
  \resRepsPerRow{} runs. $^{a}$SIL~3 actors lowered to SIL~2. $^{b}$Sobel
  on Rosvall's platform. $^{c}$Separate catalog with the NVP pattern.
  $^{d}$Sobel (SIL~2), SUSAN (SIL~0), and RASTA, one period each.}
\end{table}

\begin{table}[!t]
  \caption{Time per Pipeline Stage (s)}
  \label{tab:pipe}
  \centering\footnotesize
  \setlength{\tabcolsep}{3pt}
\begin{tabular*}{\columnwidth}{@{\extracolsep{\fill}}l rrrrr@{}}
\toprule
Instance & Front-end & Solve & Period step & Verifier & Generator \\
\midrule
\code{p\_rasta} & 0.14 & 6.2 & 0.53 & 0.09 & 0.17 \\
\code{f\_random\_hw} & 0.14 & 9.5 & 0.55 & 0.09 & 0.17 \\
\code{d\_random\_hw} & 0.16 & 7.1 & 0.56 & 0.10 & 0.17 \\
\code{d\_systematic\_sw} & 0.15 & 3.5 & 0.50 & 0.10 & 0.17 \\
\code{f\_sw3} & 0.14 & 5.7 & 0.46 & 0.08 & 0.15 \\
\code{f\_both3} & 0.14 & 3.4 & 0.48 & 0.10 & 0.16 \\
\code{c\_pat} & 0.15 & 15.1 & 1.17 & 0.19 & 0.18 \\
\code{pc\_sobel} & 0.15 & 3.2 & 0.92 & 0.10 & 0.17 \\
\code{v\_nvp} & 0.12 & 5.6 & 0.52 & 0.10 & 0.14 \\
\code{m\_3app} & 0.15 & 80.7 & 0.79 & 0.12 & 0.18 \\
\bottomrule
\end{tabular*}

\vspace{2pt}\parbox{\columnwidth}{\scriptsize Wall-clock seconds of one run of each stage on the joint TOTALCOST instance (\resSetupThreads{} workers for the solve). Period step: the design fixed and $\mu$ minimized (\cref{sec:eval-setup}).}

\end{table}

\subsection{RQ1: Comparison With \desyde{}}\label{sec:eval-rq1}
Without a safety specification, the model only maps and schedules. We ran
it and \desyde{} on the files of Experiments~3 and~4 of Rosvall and
Sander's article~\cite{rosvall2018flexible}: four applications on eight
identical cores with a TDMA bus (Experiment~3) and the same with a JPEG
accelerator (Experiment~4). \desyde{} was built at the version of the
article. It needed compile fixes for a current compiler and one fix of
undefined behavior: the constructor used by these experiments left the
system constraints uninitialized, which failed every scenario at the root.
After the fix it reproduced the article's first result exactly. Both tools
had \resDsdThreads{} threads and a wall budget of \resDsdSafeTimeLimit{}\,s
per scenario. \desyde{} ran its published configuration, a two-step
heuristic with its throughput propagator, and minimizes one application's
period under bounds for the others; \safedse{} minimizes the worst period,
and every solution was pinned and verified. The objectives coincide when
the optimized application's period is bounded below by more than every
other bound, which holds in \resDsdComparable{} of the
\resDsdScenarios{} scenarios; the other \resDsdSkipped{} (Experiment~3
SoSuRa and the Experiment~4 scenarios with SUSAN) were not compared.

\begin{table}[!t]
  \caption{Head-to-Head With \desyde{} (Periods, Times in s)}
  \label{tab:desyde}
  \centering\footnotesize
  \setlength{\tabcolsep}{2.5pt}
  \begin{tabular*}{\columnwidth}{@{\extracolsep{\fill}}l rr rrc rr@{}}
    \toprule
    & \multicolumn{2}{c}{Article} & \multicolumn{3}{c}{\desyde{} here}
    & \multicolumn{2}{c}{\safedse{}} \\
    \cmidrule(lr){2-3}\cmidrule(lr){4-6}\cmidrule(l){7-8}
    Scenario & $\mu$ & $t$ & $\mu$ & $t$ & pr. & $\mu$ & $t$ \\
    \midrule
    3: RaJp & \resDsdThreeRaJpPubPeriod & \resDsdThreeRaJpPubTime & \resDsdThreeRaJpDesydePeriod & \resDsdThreeRaJpDesydeTime & \resDsdThreeRaJpDesydeProven & \resDsdThreeRaJpSafePeriod & \resDsdThreeRaJpSafeTime \\
    3: SoSuJp & \resDsdThreeSoSuJpPubPeriod & \resDsdThreeSoSuJpPubTime & \resDsdThreeSoSuJpDesydePeriod & \resDsdThreeSoSuJpDesydeTime & \resDsdThreeSoSuJpDesydeProven & \resDsdThreeSoSuJpSafePeriod & \resDsdThreeSoSuJpSafeTime \\
    3: SoRaJp & \resDsdThreeSoRaJpPubPeriod & \resDsdThreeSoRaJpPubTime & \resDsdThreeSoRaJpDesydePeriod & \resDsdThreeSoRaJpDesydeTime & \resDsdThreeSoRaJpDesydeProven & \resDsdThreeSoRaJpSafePeriod & \resDsdThreeSoRaJpSafeTime$^{\ast}$ \\
    3: SuRaJp & \resDsdThreeSuRaJpPubPeriod & \resDsdThreeSuRaJpPubTime & \resDsdThreeSuRaJpDesydePeriod & \resDsdThreeSuRaJpDesydeTime & \resDsdThreeSuRaJpDesydeProven & \resDsdThreeSuRaJpSafePeriod & \resDsdThreeSuRaJpSafeTime \\
    3: SoSuRaJp & \resDsdThreeSoSuRaJpPubPeriod & \resDsdThreeSoSuRaJpPubTime & \resDsdThreeSoSuRaJpDesydePeriod & \resDsdThreeSoSuRaJpDesydeTime & \resDsdThreeSoSuRaJpDesydeProven & \resDsdThreeSoSuRaJpSafePeriod & \resDsdThreeSoSuRaJpSafeTime \\
    4: RaJp & \resDsdFourRaJpPubPeriod & \resDsdFourRaJpPubTime & \resDsdFourRaJpDesydePeriod & \resDsdFourRaJpDesydeTime & \resDsdFourRaJpDesydeProven & \resDsdFourRaJpSafePeriod & \resDsdFourRaJpSafeTime \\
    4: SoRaJp & \resDsdFourSoRaJpPubPeriod & \resDsdFourSoRaJpPubTime$^{\ast}$ & \resDsdFourSoRaJpDesydePeriod & \resDsdFourSoRaJpDesydeTime & \resDsdFourSoRaJpDesydeProven & \resDsdFourSoRaJpSafePeriod & \resDsdFourSoRaJpSafeTime \\
    \bottomrule
  \end{tabular*}

  \vspace{2pt}\parbox{\columnwidth}{\scriptsize Experiment: scenario
  (So Sobel, Su SUSAN, Ra RASTA, Jp JPEG). Article: \desyde{} with 8 threads
  as published~\cite{rosvall2018flexible}; here: \resDsdThreads{} threads,
  \resDsdSafeTimeLimit{}\,s wall budget; a time at the budget means the
  search was stopped with the best value so far. pr.: proven within the
  budget. $^{\ast}$Not proven. \safedse{} times are medians over up to
  \resDsdMaxReps{} runs for the fast scenarios and one run otherwise.}
\end{table}

\Cref{tab:desyde} gives the results. \safedse{} reached the article's
period in all \resDsdComparable{} scenarios and proved it in
\resDsdSafeProven{}, including one optimum that the article left open after
\resDsdFourSoRaJpPubTime{}\,s (here \resDsdFourSoRaJpSafeTime{}\,s).
\desyde{} proved \resDsdDesydeProven{} within the budget on our machine. In
the other scenarios, its first step mostly ended at its inactivity time-out
with the article's value, and its second step did not finish. On
Experiment~4 RaJp, it reached only \resDsdFourRaJpDesydePeriod{} (article:
\resDsdFourRaJpPubPeriod). \safedse{} is not faster everywhere: RASTA with
JPEG took \resDsdThreeRaJpSafeTime{}\,s against
\resDsdThreeRaJpDesydeTime{}\,s, and two scenarios took close to the
budget. Both searches are nondeterministic, and the one-hour scenarios ran
once; the comparison covers only scenarios where the two objectives agree.
An uncapped first \desyde{} run of one scenario was still in its second
step long after the budget, where only \desyde{}'s own six-hour limit
applies, so all runs got the same cap.

\begin{table}[!t]
  \caption{Single-Application Benchmarks, Minimum Hardware Cost}
  \label{tab:rq1-hw}
  \centering\footnotesize
  \setlength{\tabcolsep}{4pt}
  \begin{tabular*}{\columnwidth}{@{\extracolsep{\fill}}l r r r r@{}}
    \toprule
    Application & Hardware cost & $\mu$ & Published bound & $t$ (s) \\
    \midrule
    Sobel & \resOneStepSobelHwCost & \resOneStepSobelMuMax & \resOneStepSobelBound & \resOneStepSobelTime \\
    SUSAN & \resOneStepSusanHwCost & \resOneStepSusanMuMax & \resOneStepSusanBound & \resOneStepSusanTime \\
    RASTA & \resOneStepRastaHwCost & \resOneStepRastaMuMax & \resOneStepRastaBound & \resOneStepRastaTime \\
    \bottomrule
  \end{tabular*}

  \vspace{2pt}\parbox{\columnwidth}{\scriptsize Rosvall and Sander's
  platform and period bounds; $\mu$: largest period over the repetitions
  (the period of a cost optimum is not unique); $t$: median time to prove
  the optimum.}
\end{table}

On the single-application benchmarks of the same article, minimizing the
hardware cost under the published period bound gives the optima of
\cref{tab:rq1-hw}, each proven within a second, and every period meets the
published bound. With all four applications together, the optimum has
\resRosvallAllTwoStepHwCost{} cost units on \resRosvallAllTwoStepNprocs{}
cores. The two-step strategy of \cref{sec:model-solving} proved it in
\resRosvallAllTwoStepTime{}\,s (step~1 \resRosvallAllTwoStepStepOneTime{}\,s,
step~2 \resRosvallAllTwoStepStepTwoTime{}\,s), and one step in
\resRosvallAllOnestepTime{}\,s (status \resRosvallAllOnestepStatus). With
four workers, the two steps are thus not needed on this benchmark; with one
worker in an early run, only the two-step strategy proved the optimum
within the time limit.

\subsection{RQ2: Value of Joint Exploration}\label{sec:eval-rq2}
We compare joint exploration with decoupled flows, each the joint model plus
constraints, so none can be cheaper. Safety-last (SL) maps the instance
without patterns at least cost (ties by period, then by core index), fixes
the core of every base actor, and leaves patterns, replicas, and platform
free. Pattern-first fixes each actor's pattern before mapping. The rules take
the fewest components (PF$^{-}$), the widest fault coverage (PF$^{+}$, akin to
SafPat~\cite{gildantas2023toolchain}), the least added WCET on the fastest
type (PF$^{t}$), or the least added cost (PF$^{c}$), with the catalog order
as the last tie-break. No-promotion (NP) forbids promotion. All minimize
$C_{\mathrm{tot}}$ without a period bound and with the bound
$\lceil f\mu^{*}\rceil$ for $f\in\{1.00,1.10,1.25,1.50,2.00\}$, where
$\mu^{*}$ is the instance's own minimum period. No other tool solves this
problem (\cref{tab:related}). A pattern-first flow that backtracks over its
pattern choices would be joint exploration by enumeration; the flows here
model a design process that fixes the architecture first.

\Cref{tab:rq2-nobound} shows the results without a bound. SL, PF$^{-}$,
PF$^{c}$, and NP reach the joint optimum on all \resBaselineInstances{}
instances; PF$^{+}$ loses up to \resBaselinePFMaxMaxLoss{}\,\% and PF$^{t}$
up to \resBaselinePFLightMaxLoss{}\,\%. Without a bound, throughput never
binds, and SL can still co-optimize patterns, replicas, and platform after
fixing the base actors, since a replica can always go to another core.
Under a bound, PF$^{-}$ and PF$^{+}$ are infeasible on
\resTightHundredTwentyFivePFMinInfeasible{} instances up to
$1.25\,\mu^{*}$ and on \resTightHundredFiftyPFMinInfeasible{} at
$1.50\,\mu^{*}$, and PF$^{c}$ on \resTightHundredFiftyPFCheapInfeasible{}
up to $1.50\,\mu^{*}$ (\cref{fig:joint}, \cref{tab:rq2-all}). PF$^{t}$
stays feasible but loses up to \resTightTwoHundredPFLightMaxLoss{}\,\% at
$2.00\,\mu^{*}$, and SL fails on \resTightHundredSLInfeasible{} at
$1.00\,\mu^{*}$ and loses up to \resTightTwoHundredSLMaxLoss{}\,\%. NP
equals joint everywhere. No pattern-first rule is thus both always feasible
and never more expensive. \Cref{fig:worse} adds how often each flow is
feasible but dearer.

\begin{figure}[!t]
  \centering
  \includegraphics{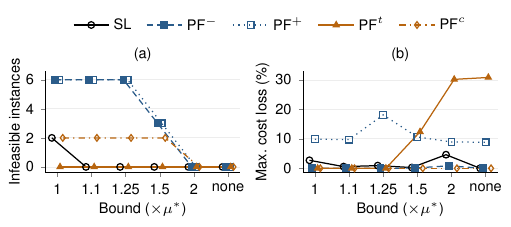}
  \caption{Decoupled flows against joint exploration per period bound: (a)
  instances proven to have no feasible design, (b) largest cost loss among
  the feasible ones (NP, equal to joint, not drawn). No pattern-first rule
  is both always feasible and never more expensive.}
  \label{fig:joint}
\end{figure}

\begin{figure}[!t]
  \centering
  \includegraphics{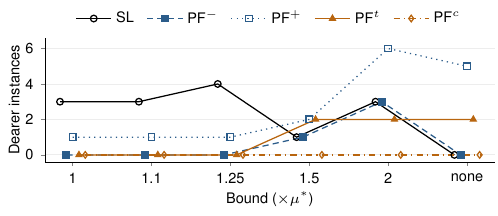}
  \caption{Instances on which a decoupled flow is feasible but dearer than
  joint exploration, per period bound (the third count of
  \cref{tab:rq2-all}). SL loses cost mostly at tight bounds, the
  timing-agnostic PF$^{+}$ mostly at loose ones.}
  \label{fig:worse}
\end{figure}

\begin{table}[!t]
  \caption{RQ2 Without a Period Bound (Extra Cost Over Joint)}
  \label{tab:rq2-nobound}
  \centering\footnotesize
  \setlength{\tabcolsep}{2.5pt}
\begin{tabular*}{\columnwidth}{@{\extracolsep{\fill}}l r r r r r r r c@{}}
\toprule
Instance & Joint & SL & PF$^{-}$ & PF$^{+}$ & PF$^{t}$ & PF$^{c}$ & NP & Cores \\
\midrule
\code{p\_rasta} & 3330 & = & = & = & = & = & = & 2 \\
\code{f\_random\_hw} & 3349 & = & = & +25 & = & = & = & 5 \\
\code{d\_random\_hw} & 1857 & = & = & +25 & = & = & = & 3 \\
\code{d\_systematic\_sw} & 1846 & = & = & +163 & = & = & = & 2 \\
\code{f\_sw3} & 3365 & = & = & +41 & +1038 & = & = & 4 \\
\code{f\_both3} & 3406 & = & = & = & +1038 & = & = & 5 \\
\code{c\_pat} & 3349 & = & = & +25 & = & = & = & 5 \\
\code{pc\_sobel} & 1401 & = & = & = & = & = & = & 3 \\
\code{v\_nvp} & 5680 & = & = & = & = & = & = & 6 \\
\bottomrule
\end{tabular*}

\vspace{2pt}\parbox{\columnwidth}{\scriptsize No period bound. Joint: proven minimum of $C_{\mathrm{tot}}$. Flows: extra cost over joint (=: equal). Cores: of the joint optimum.}

\end{table}

\begin{table*}[!t]
  \caption{RQ2 per Input Variant and Period Bound: Infeasible / Dearer /
  Largest Loss (\%)}
  \label{tab:rq2-all}
  \centering\footnotesize
  \setlength{\tabcolsep}{3pt}
\begin{tabular*}{\textwidth}{@{\extracolsep{\fill}}l l c c c c c c c @{}}
\toprule
Inputs & Bound & $n$ & SL & PF$^{-}$ & PF$^{+}$ & PF$^{t}$ & PF$^{c}$ & NP \\
\midrule
base & none & \resBaselineInstances{} & \resBaselineSLInfeasible{}/\resBaselineSLWorse{}/\resBaselineSLMaxLoss{} & \resBaselinePFMinInfeasible{}/\resBaselinePFMinWorse{}/\resBaselinePFMinMaxLoss{} & \resBaselinePFMaxInfeasible{}/\resBaselinePFMaxWorse{}/\resBaselinePFMaxMaxLoss{} & \resBaselinePFLightInfeasible{}/\resBaselinePFLightWorse{}/\resBaselinePFLightMaxLoss{} & \resBaselinePFCheapInfeasible{}/\resBaselinePFCheapWorse{}/\resBaselinePFCheapMaxLoss{} & \resBaselineNPInfeasible{}/\resBaselineNPWorse{}/\resBaselineNPMaxLoss{} \\
 & $1.00\,\mu^{*}$ & \resTightHundredInstances{} & \resTightHundredSLInfeasible{}/\resTightHundredSLWorse{}/\resTightHundredSLMaxLoss{} & \resTightHundredPFMinInfeasible{}/\resTightHundredPFMinWorse{}/\resTightHundredPFMinMaxLoss{} & \resTightHundredPFMaxInfeasible{}/\resTightHundredPFMaxWorse{}/\resTightHundredPFMaxMaxLoss{} & \resTightHundredPFLightInfeasible{}/\resTightHundredPFLightWorse{}/\resTightHundredPFLightMaxLoss{} & \resTightHundredPFCheapInfeasible{}/\resTightHundredPFCheapWorse{}/\resTightHundredPFCheapMaxLoss{} & \resTightHundredNPInfeasible{}/\resTightHundredNPWorse{}/\resTightHundredNPMaxLoss{} \\
 & $1.10\,\mu^{*}$ & \resTightHundredTenInstances{} & \resTightHundredTenSLInfeasible{}/\resTightHundredTenSLWorse{}/\resTightHundredTenSLMaxLoss{} & \resTightHundredTenPFMinInfeasible{}/\resTightHundredTenPFMinWorse{}/\resTightHundredTenPFMinMaxLoss{} & \resTightHundredTenPFMaxInfeasible{}/\resTightHundredTenPFMaxWorse{}/\resTightHundredTenPFMaxMaxLoss{} & \resTightHundredTenPFLightInfeasible{}/\resTightHundredTenPFLightWorse{}/\resTightHundredTenPFLightMaxLoss{} & \resTightHundredTenPFCheapInfeasible{}/\resTightHundredTenPFCheapWorse{}/\resTightHundredTenPFCheapMaxLoss{} & \resTightHundredTenNPInfeasible{}/\resTightHundredTenNPWorse{}/\resTightHundredTenNPMaxLoss{} \\
 & $1.25\,\mu^{*}$ & \resTightHundredTwentyFiveInstances{} & \resTightHundredTwentyFiveSLInfeasible{}/\resTightHundredTwentyFiveSLWorse{}/\resTightHundredTwentyFiveSLMaxLoss{} & \resTightHundredTwentyFivePFMinInfeasible{}/\resTightHundredTwentyFivePFMinWorse{}/\resTightHundredTwentyFivePFMinMaxLoss{} & \resTightHundredTwentyFivePFMaxInfeasible{}/\resTightHundredTwentyFivePFMaxWorse{}/\resTightHundredTwentyFivePFMaxMaxLoss{} & \resTightHundredTwentyFivePFLightInfeasible{}/\resTightHundredTwentyFivePFLightWorse{}/\resTightHundredTwentyFivePFLightMaxLoss{} & \resTightHundredTwentyFivePFCheapInfeasible{}/\resTightHundredTwentyFivePFCheapWorse{}/\resTightHundredTwentyFivePFCheapMaxLoss{} & \resTightHundredTwentyFiveNPInfeasible{}/\resTightHundredTwentyFiveNPWorse{}/\resTightHundredTwentyFiveNPMaxLoss{} \\
 & $1.50\,\mu^{*}$ & \resTightHundredFiftyInstances{} & \resTightHundredFiftySLInfeasible{}/\resTightHundredFiftySLWorse{}/\resTightHundredFiftySLMaxLoss{} & \resTightHundredFiftyPFMinInfeasible{}/\resTightHundredFiftyPFMinWorse{}/\resTightHundredFiftyPFMinMaxLoss{} & \resTightHundredFiftyPFMaxInfeasible{}/\resTightHundredFiftyPFMaxWorse{}/\resTightHundredFiftyPFMaxMaxLoss{} & \resTightHundredFiftyPFLightInfeasible{}/\resTightHundredFiftyPFLightWorse{}/\resTightHundredFiftyPFLightMaxLoss{} & \resTightHundredFiftyPFCheapInfeasible{}/\resTightHundredFiftyPFCheapWorse{}/\resTightHundredFiftyPFCheapMaxLoss{} & \resTightHundredFiftyNPInfeasible{}/\resTightHundredFiftyNPWorse{}/\resTightHundredFiftyNPMaxLoss{} \\
 & $2.00\,\mu^{*}$ & \resTightTwoHundredInstances{} & \resTightTwoHundredSLInfeasible{}/\resTightTwoHundredSLWorse{}/\resTightTwoHundredSLMaxLoss{} & \resTightTwoHundredPFMinInfeasible{}/\resTightTwoHundredPFMinWorse{}/\resTightTwoHundredPFMinMaxLoss{} & \resTightTwoHundredPFMaxInfeasible{}/\resTightTwoHundredPFMaxWorse{}/\resTightTwoHundredPFMaxMaxLoss{} & \resTightTwoHundredPFLightInfeasible{}/\resTightTwoHundredPFLightWorse{}/\resTightTwoHundredPFLightMaxLoss{} & \resTightTwoHundredPFCheapInfeasible{}/\resTightTwoHundredPFCheapWorse{}/\resTightTwoHundredPFCheapMaxLoss{} & \resTightTwoHundredNPInfeasible{}/\resTightTwoHundredNPWorse{}/\resTightTwoHundredNPMaxLoss{} \\
\midrule
\code{tok2} & none & \resVarTokTwoNoBoundInstances{} & \resVarTokTwoNoBoundSLInfeasible{}/\resVarTokTwoNoBoundSLWorse{}/\resVarTokTwoNoBoundSLMaxLoss{} & \resVarTokTwoNoBoundPFMinInfeasible{}/\resVarTokTwoNoBoundPFMinWorse{}/\resVarTokTwoNoBoundPFMinMaxLoss{} & \resVarTokTwoNoBoundPFMaxInfeasible{}/\resVarTokTwoNoBoundPFMaxWorse{}/\resVarTokTwoNoBoundPFMaxMaxLoss{} & \resVarTokTwoNoBoundPFLightInfeasible{}/\resVarTokTwoNoBoundPFLightWorse{}/\resVarTokTwoNoBoundPFLightMaxLoss{} & \resVarTokTwoNoBoundPFCheapInfeasible{}/\resVarTokTwoNoBoundPFCheapWorse{}/\resVarTokTwoNoBoundPFCheapMaxLoss{} & n/r \\
 & $1.00\,\mu^{*}$ & \resVarTokTwoHundredInstances{} & \resVarTokTwoHundredSLInfeasible{}/\resVarTokTwoHundredSLWorse{}/\resVarTokTwoHundredSLMaxLoss{} & \resVarTokTwoHundredPFMinInfeasible{}/\resVarTokTwoHundredPFMinWorse{}/\resVarTokTwoHundredPFMinMaxLoss{} & \resVarTokTwoHundredPFMaxInfeasible{}/\resVarTokTwoHundredPFMaxWorse{}/\resVarTokTwoHundredPFMaxMaxLoss{} & \resVarTokTwoHundredPFLightInfeasible{}/\resVarTokTwoHundredPFLightWorse{}/\resVarTokTwoHundredPFLightMaxLoss{} & \resVarTokTwoHundredPFCheapInfeasible{}/\resVarTokTwoHundredPFCheapWorse{}/\resVarTokTwoHundredPFCheapMaxLoss{} & n/r \\
 & $1.10\,\mu^{*}$ & \resVarTokTwoHundredTenInstances{} & \resVarTokTwoHundredTenSLInfeasible{}/\resVarTokTwoHundredTenSLWorse{}/\resVarTokTwoHundredTenSLMaxLoss{} & \resVarTokTwoHundredTenPFMinInfeasible{}/\resVarTokTwoHundredTenPFMinWorse{}/\resVarTokTwoHundredTenPFMinMaxLoss{} & \resVarTokTwoHundredTenPFMaxInfeasible{}/\resVarTokTwoHundredTenPFMaxWorse{}/\resVarTokTwoHundredTenPFMaxMaxLoss{} & \resVarTokTwoHundredTenPFLightInfeasible{}/\resVarTokTwoHundredTenPFLightWorse{}/\resVarTokTwoHundredTenPFLightMaxLoss{} & \resVarTokTwoHundredTenPFCheapInfeasible{}/\resVarTokTwoHundredTenPFCheapWorse{}/\resVarTokTwoHundredTenPFCheapMaxLoss{} & n/r \\
 & $1.25\,\mu^{*}$ & \resVarTokTwoHundredTwentyFiveInstances{} & \resVarTokTwoHundredTwentyFiveSLInfeasible{}/\resVarTokTwoHundredTwentyFiveSLWorse{}/\resVarTokTwoHundredTwentyFiveSLMaxLoss{} & \resVarTokTwoHundredTwentyFivePFMinInfeasible{}/\resVarTokTwoHundredTwentyFivePFMinWorse{}/\resVarTokTwoHundredTwentyFivePFMinMaxLoss{} & \resVarTokTwoHundredTwentyFivePFMaxInfeasible{}/\resVarTokTwoHundredTwentyFivePFMaxWorse{}/\resVarTokTwoHundredTwentyFivePFMaxMaxLoss{} & \resVarTokTwoHundredTwentyFivePFLightInfeasible{}/\resVarTokTwoHundredTwentyFivePFLightWorse{}/\resVarTokTwoHundredTwentyFivePFLightMaxLoss{} & \resVarTokTwoHundredTwentyFivePFCheapInfeasible{}/\resVarTokTwoHundredTwentyFivePFCheapWorse{}/\resVarTokTwoHundredTwentyFivePFCheapMaxLoss{} & n/r \\
 & $1.50\,\mu^{*}$ & \resVarTokTwoHundredFiftyInstances{} & \resVarTokTwoHundredFiftySLInfeasible{}/\resVarTokTwoHundredFiftySLWorse{}/\resVarTokTwoHundredFiftySLMaxLoss{} & \resVarTokTwoHundredFiftyPFMinInfeasible{}/\resVarTokTwoHundredFiftyPFMinWorse{}/\resVarTokTwoHundredFiftyPFMinMaxLoss{} & \resVarTokTwoHundredFiftyPFMaxInfeasible{}/\resVarTokTwoHundredFiftyPFMaxWorse{}/\resVarTokTwoHundredFiftyPFMaxMaxLoss{} & \resVarTokTwoHundredFiftyPFLightInfeasible{}/\resVarTokTwoHundredFiftyPFLightWorse{}/\resVarTokTwoHundredFiftyPFLightMaxLoss{} & \resVarTokTwoHundredFiftyPFCheapInfeasible{}/\resVarTokTwoHundredFiftyPFCheapWorse{}/\resVarTokTwoHundredFiftyPFCheapMaxLoss{} & n/r \\
 & $2.00\,\mu^{*}$ & \resVarTokTwoTwoHundredInstances{} & \resVarTokTwoTwoHundredSLInfeasible{}/\resVarTokTwoTwoHundredSLWorse{}/\resVarTokTwoTwoHundredSLMaxLoss{} & \resVarTokTwoTwoHundredPFMinInfeasible{}/\resVarTokTwoTwoHundredPFMinWorse{}/\resVarTokTwoTwoHundredPFMinMaxLoss{} & \resVarTokTwoTwoHundredPFMaxInfeasible{}/\resVarTokTwoTwoHundredPFMaxWorse{}/\resVarTokTwoTwoHundredPFMaxMaxLoss{} & \resVarTokTwoTwoHundredPFLightInfeasible{}/\resVarTokTwoTwoHundredPFLightWorse{}/\resVarTokTwoTwoHundredPFLightMaxLoss{} & \resVarTokTwoTwoHundredPFCheapInfeasible{}/\resVarTokTwoTwoHundredPFCheapWorse{}/\resVarTokTwoTwoHundredPFCheapMaxLoss{} & n/r \\
\midrule
\code{chk10} & none & \resVarChkTenNoBoundInstances{} & \resVarChkTenNoBoundSLInfeasible{}/\resVarChkTenNoBoundSLWorse{}/\resVarChkTenNoBoundSLMaxLoss{} & \resVarChkTenNoBoundPFMinInfeasible{}/\resVarChkTenNoBoundPFMinWorse{}/\resVarChkTenNoBoundPFMinMaxLoss{} & \resVarChkTenNoBoundPFMaxInfeasible{}/\resVarChkTenNoBoundPFMaxWorse{}/\resVarChkTenNoBoundPFMaxMaxLoss{} & \resVarChkTenNoBoundPFLightInfeasible{}/\resVarChkTenNoBoundPFLightWorse{}/\resVarChkTenNoBoundPFLightMaxLoss{} & \resVarChkTenNoBoundPFCheapInfeasible{}/\resVarChkTenNoBoundPFCheapWorse{}/\resVarChkTenNoBoundPFCheapMaxLoss{} & n/r \\
 & $1.00\,\mu^{*}$ & \resVarChkTenHundredInstances{} & \resVarChkTenHundredSLInfeasible{}/\resVarChkTenHundredSLWorse{}/\resVarChkTenHundredSLMaxLoss{} & \resVarChkTenHundredPFMinInfeasible{}/\resVarChkTenHundredPFMinWorse{}/\resVarChkTenHundredPFMinMaxLoss{} & \resVarChkTenHundredPFMaxInfeasible{}/\resVarChkTenHundredPFMaxWorse{}/\resVarChkTenHundredPFMaxMaxLoss{} & \resVarChkTenHundredPFLightInfeasible{}/\resVarChkTenHundredPFLightWorse{}/\resVarChkTenHundredPFLightMaxLoss{} & \resVarChkTenHundredPFCheapInfeasible{}/\resVarChkTenHundredPFCheapWorse{}/\resVarChkTenHundredPFCheapMaxLoss{} & n/r \\
 & $1.10\,\mu^{*}$ & \resVarChkTenHundredTenInstances{} & \resVarChkTenHundredTenSLInfeasible{}/\resVarChkTenHundredTenSLWorse{}/\resVarChkTenHundredTenSLMaxLoss{} & \resVarChkTenHundredTenPFMinInfeasible{}/\resVarChkTenHundredTenPFMinWorse{}/\resVarChkTenHundredTenPFMinMaxLoss{} & \resVarChkTenHundredTenPFMaxInfeasible{}/\resVarChkTenHundredTenPFMaxWorse{}/\resVarChkTenHundredTenPFMaxMaxLoss{} & \resVarChkTenHundredTenPFLightInfeasible{}/\resVarChkTenHundredTenPFLightWorse{}/\resVarChkTenHundredTenPFLightMaxLoss{} & \resVarChkTenHundredTenPFCheapInfeasible{}/\resVarChkTenHundredTenPFCheapWorse{}/\resVarChkTenHundredTenPFCheapMaxLoss{} & n/r \\
 & $1.25\,\mu^{*}$ & \resVarChkTenHundredTwentyFiveInstances{} & \resVarChkTenHundredTwentyFiveSLInfeasible{}/\resVarChkTenHundredTwentyFiveSLWorse{}/\resVarChkTenHundredTwentyFiveSLMaxLoss{} & \resVarChkTenHundredTwentyFivePFMinInfeasible{}/\resVarChkTenHundredTwentyFivePFMinWorse{}/\resVarChkTenHundredTwentyFivePFMinMaxLoss{} & \resVarChkTenHundredTwentyFivePFMaxInfeasible{}/\resVarChkTenHundredTwentyFivePFMaxWorse{}/\resVarChkTenHundredTwentyFivePFMaxMaxLoss{} & \resVarChkTenHundredTwentyFivePFLightInfeasible{}/\resVarChkTenHundredTwentyFivePFLightWorse{}/\resVarChkTenHundredTwentyFivePFLightMaxLoss{} & \resVarChkTenHundredTwentyFivePFCheapInfeasible{}/\resVarChkTenHundredTwentyFivePFCheapWorse{}/\resVarChkTenHundredTwentyFivePFCheapMaxLoss{} & n/r \\
 & $1.50\,\mu^{*}$ & \resVarChkTenHundredFiftyInstances{} & \resVarChkTenHundredFiftySLInfeasible{}/\resVarChkTenHundredFiftySLWorse{}/\resVarChkTenHundredFiftySLMaxLoss{} & \resVarChkTenHundredFiftyPFMinInfeasible{}/\resVarChkTenHundredFiftyPFMinWorse{}/\resVarChkTenHundredFiftyPFMinMaxLoss{} & \resVarChkTenHundredFiftyPFMaxInfeasible{}/\resVarChkTenHundredFiftyPFMaxWorse{}/\resVarChkTenHundredFiftyPFMaxMaxLoss{} & \resVarChkTenHundredFiftyPFLightInfeasible{}/\resVarChkTenHundredFiftyPFLightWorse{}/\resVarChkTenHundredFiftyPFLightMaxLoss{} & \resVarChkTenHundredFiftyPFCheapInfeasible{}/\resVarChkTenHundredFiftyPFCheapWorse{}/\resVarChkTenHundredFiftyPFCheapMaxLoss{} & n/r \\
 & $2.00\,\mu^{*}$ & \resVarChkTenTwoHundredInstances{} & \resVarChkTenTwoHundredSLInfeasible{}/\resVarChkTenTwoHundredSLWorse{}/\resVarChkTenTwoHundredSLMaxLoss{} & \resVarChkTenTwoHundredPFMinInfeasible{}/\resVarChkTenTwoHundredPFMinWorse{}/\resVarChkTenTwoHundredPFMinMaxLoss{} & \resVarChkTenTwoHundredPFMaxInfeasible{}/\resVarChkTenTwoHundredPFMaxWorse{}/\resVarChkTenTwoHundredPFMaxMaxLoss{} & \resVarChkTenTwoHundredPFLightInfeasible{}/\resVarChkTenTwoHundredPFLightWorse{}/\resVarChkTenTwoHundredPFLightMaxLoss{} & \resVarChkTenTwoHundredPFCheapInfeasible{}/\resVarChkTenTwoHundredPFCheapWorse{}/\resVarChkTenTwoHundredPFCheapMaxLoss{} & n/r \\
\midrule
\code{chk50} & none & \resVarChkFiftyNoBoundInstances{} & \resVarChkFiftyNoBoundSLInfeasible{}/\resVarChkFiftyNoBoundSLWorse{}/\resVarChkFiftyNoBoundSLMaxLoss{} & \resVarChkFiftyNoBoundPFMinInfeasible{}/\resVarChkFiftyNoBoundPFMinWorse{}/\resVarChkFiftyNoBoundPFMinMaxLoss{} & \resVarChkFiftyNoBoundPFMaxInfeasible{}/\resVarChkFiftyNoBoundPFMaxWorse{}/\resVarChkFiftyNoBoundPFMaxMaxLoss{} & \resVarChkFiftyNoBoundPFLightInfeasible{}/\resVarChkFiftyNoBoundPFLightWorse{}/\resVarChkFiftyNoBoundPFLightMaxLoss{} & \resVarChkFiftyNoBoundPFCheapInfeasible{}/\resVarChkFiftyNoBoundPFCheapWorse{}/\resVarChkFiftyNoBoundPFCheapMaxLoss{} & n/r \\
 & $1.00\,\mu^{*}$ & \resVarChkFiftyHundredInstances{} & \resVarChkFiftyHundredSLInfeasible{}/\resVarChkFiftyHundredSLWorse{}/\resVarChkFiftyHundredSLMaxLoss{} & \resVarChkFiftyHundredPFMinInfeasible{}/\resVarChkFiftyHundredPFMinWorse{}/\resVarChkFiftyHundredPFMinMaxLoss{} & \resVarChkFiftyHundredPFMaxInfeasible{}/\resVarChkFiftyHundredPFMaxWorse{}/\resVarChkFiftyHundredPFMaxMaxLoss{} & \resVarChkFiftyHundredPFLightInfeasible{}/\resVarChkFiftyHundredPFLightWorse{}/\resVarChkFiftyHundredPFLightMaxLoss{} & \resVarChkFiftyHundredPFCheapInfeasible{}/\resVarChkFiftyHundredPFCheapWorse{}/\resVarChkFiftyHundredPFCheapMaxLoss{} & n/r \\
 & $1.10\,\mu^{*}$ & \resVarChkFiftyHundredTenInstances{} & \resVarChkFiftyHundredTenSLInfeasible{}/\resVarChkFiftyHundredTenSLWorse{}/\resVarChkFiftyHundredTenSLMaxLoss{} & \resVarChkFiftyHundredTenPFMinInfeasible{}/\resVarChkFiftyHundredTenPFMinWorse{}/\resVarChkFiftyHundredTenPFMinMaxLoss{} & \resVarChkFiftyHundredTenPFMaxInfeasible{}/\resVarChkFiftyHundredTenPFMaxWorse{}/\resVarChkFiftyHundredTenPFMaxMaxLoss{} & \resVarChkFiftyHundredTenPFLightInfeasible{}/\resVarChkFiftyHundredTenPFLightWorse{}/\resVarChkFiftyHundredTenPFLightMaxLoss{} & \resVarChkFiftyHundredTenPFCheapInfeasible{}/\resVarChkFiftyHundredTenPFCheapWorse{}/\resVarChkFiftyHundredTenPFCheapMaxLoss{} & n/r \\
 & $1.25\,\mu^{*}$ & \resVarChkFiftyHundredTwentyFiveInstances{} & \resVarChkFiftyHundredTwentyFiveSLInfeasible{}/\resVarChkFiftyHundredTwentyFiveSLWorse{}/\resVarChkFiftyHundredTwentyFiveSLMaxLoss{} & \resVarChkFiftyHundredTwentyFivePFMinInfeasible{}/\resVarChkFiftyHundredTwentyFivePFMinWorse{}/\resVarChkFiftyHundredTwentyFivePFMinMaxLoss{} & \resVarChkFiftyHundredTwentyFivePFMaxInfeasible{}/\resVarChkFiftyHundredTwentyFivePFMaxWorse{}/\resVarChkFiftyHundredTwentyFivePFMaxMaxLoss{} & \resVarChkFiftyHundredTwentyFivePFLightInfeasible{}/\resVarChkFiftyHundredTwentyFivePFLightWorse{}/\resVarChkFiftyHundredTwentyFivePFLightMaxLoss{} & \resVarChkFiftyHundredTwentyFivePFCheapInfeasible{}/\resVarChkFiftyHundredTwentyFivePFCheapWorse{}/\resVarChkFiftyHundredTwentyFivePFCheapMaxLoss{} & n/r \\
 & $1.50\,\mu^{*}$ & \resVarChkFiftyHundredFiftyInstances{} & \resVarChkFiftyHundredFiftySLInfeasible{}/\resVarChkFiftyHundredFiftySLWorse{}/\resVarChkFiftyHundredFiftySLMaxLoss{} & \resVarChkFiftyHundredFiftyPFMinInfeasible{}/\resVarChkFiftyHundredFiftyPFMinWorse{}/\resVarChkFiftyHundredFiftyPFMinMaxLoss{} & \resVarChkFiftyHundredFiftyPFMaxInfeasible{}/\resVarChkFiftyHundredFiftyPFMaxWorse{}/\resVarChkFiftyHundredFiftyPFMaxMaxLoss{} & \resVarChkFiftyHundredFiftyPFLightInfeasible{}/\resVarChkFiftyHundredFiftyPFLightWorse{}/\resVarChkFiftyHundredFiftyPFLightMaxLoss{} & \resVarChkFiftyHundredFiftyPFCheapInfeasible{}/\resVarChkFiftyHundredFiftyPFCheapWorse{}/\resVarChkFiftyHundredFiftyPFCheapMaxLoss{} & n/r \\
 & $2.00\,\mu^{*}$ & \resVarChkFiftyTwoHundredInstances{} & \resVarChkFiftyTwoHundredSLInfeasible{}/\resVarChkFiftyTwoHundredSLWorse{}/\resVarChkFiftyTwoHundredSLMaxLoss{} & \resVarChkFiftyTwoHundredPFMinInfeasible{}/\resVarChkFiftyTwoHundredPFMinWorse{}/\resVarChkFiftyTwoHundredPFMinMaxLoss{} & \resVarChkFiftyTwoHundredPFMaxInfeasible{}/\resVarChkFiftyTwoHundredPFMaxWorse{}/\resVarChkFiftyTwoHundredPFMaxMaxLoss{} & \resVarChkFiftyTwoHundredPFLightInfeasible{}/\resVarChkFiftyTwoHundredPFLightWorse{}/\resVarChkFiftyTwoHundredPFLightMaxLoss{} & \resVarChkFiftyTwoHundredPFCheapInfeasible{}/\resVarChkFiftyTwoHundredPFCheapWorse{}/\resVarChkFiftyTwoHundredPFCheapMaxLoss{} & n/r \\
\bottomrule
\end{tabular*}

\vspace{2pt}\parbox{\textwidth}{\scriptsize Each cell: instances where the flow is infeasible / feasible but dearer than joint / largest extra cost among the feasible ones in percent (--: none feasible). $n$: instances with a verified joint optimum at that bound. Inputs: base catalog and WCETs; \code{tok2}: two tokens on every verdict edge; \code{chk10}, \code{chk50}: checker WCETs at 10\,\% and 50\,\% of the owner's (base: \resSetupCheckerScale{}\,\%). $\mu^{*}$ is the minimum period of each instance under the same inputs. n/r: not run.}

\end{table*}

Each failure is a proven infeasibility. The
\resTightHundredPFMinInfeasible{} instances concerned have a SIL-3 actor,
for which PF$^{-}$ and PF$^{+}$ fix 2-of-2 (PF$^{-}$ by catalog order, as
the isolated checker ties on its key), and PF$^{c}$ fixes it where the
mixed-SIL pair costs more (\cref{tab:pfchoice}). The 2-of-2 verdict edge
closes a one-token cycle through the owner and a full replica, so the period
is at least the sum of their WCETs. This sum is \resCertTwoOfTwoRasta{} on
RASTA,
where $\mu^{*}$ is \resThrBaseMuStarMin{}, \resCertTwoOfTwoDiverse{} when
\code{DIVERSE} puts the replica on a slower type, and
\resCertTwoOfTwoSobel{} on Sobel, where $\mu^{*}$ is
\resThrBaseMuStarMax{}. With a checker, the cycle is shorter. On all
RASTA instances, $\mu^{*}$ is the WCET of one actor on its self-edge.

\begin{table*}[!t]
  \caption{Pattern Fixed by Each Pattern-First Rule (Base Inputs)}
  \label{tab:pfchoice}
  \centering\footnotesize
  \setlength{\tabcolsep}{3pt}
\begin{tabular*}{\textwidth}{@{\extracolsep{\fill}}c l >{\raggedright\arraybackslash}p{1.4in} >{\raggedright\arraybackslash}p{1.4in} >{\raggedright\arraybackslash}p{1.4in} >{\raggedright\arraybackslash}p{1.4in} @{}}
\toprule
SIL & FM & PF$^{-}$ & PF$^{+}$ & PF$^{t}$ & PF$^{c}$ \\
\midrule
2 & H & Low-SIL D/C & Low-SIL diverse D/C & Low-SIL D/C$^{\ddagger}$ & Low-SIL D/C$^{\ddagger}$ \\
2 & S & Same-CPU D/C & Low-SIL diverse D/C & Same-CPU D/C & Same-CPU D/C \\
2 & H+S & Low-SIL diverse D/C & Low-SIL diverse D/C & Low-SIL diverse D/C & Low-SIL diverse D/C \\
3 & H & 2-of-2$^{\ddagger}$ & 2-of-2 & High-SIL isolated D/C & High-SIL isolated D/C \\
3 & S & 2-of-2 & 2-of-2 & Mixed-SIL D/C & 2-of-2 \\
3 & H+S & 2-of-2 (2); 3-version (NVP) (2) & 2-of-2 (2); 3-version (NVP) (2) & Mixed-SIL D/C (2); 3-version (NVP) (2) & 2-of-2 (2); 3-version (NVP) (2) \\
\bottomrule
\end{tabular*}

\vspace{2pt}\parbox{\textwidth}{\scriptsize Pattern that each pattern-first rule fixes for an actor of the given required SIL and fault model (FM) of the instance, base inputs; in parentheses the number of actors when an instance set mixes choices (the platform decides which patterns are applicable). SIL~0--1 actors always take \code{none}. $^{\ddagger}$The rule's key ties, and the catalog order decides for at least one actor.}

\end{table*}

\Cref{fig:thr,tab:rq2-thr} give the least period that each flow reaches
under its own restriction. That of PF$^{-}$ and PF$^{+}$ is
\resThrBasePFMinMinAboveRatio{} to \resThrBasePFMinMaxRatio{} times
$\mu^{*}$ on \resThrBasePFMinAbove{} instances, that of PF$^{c}$
\resThrBasePFCheapMaxRatio{} times on \resThrBasePFCheapAbove{}, and that
of SL at most \resThrBaseSLMaxRatio{} times (for SL, step~1 minimizes the
period too). A lower bound on the period of a pattern-first flow needs no solver: the
MCR of the graph that the fixed patterns imply, with every node at its least
WCET. Per actor, it is strengthened by enumerating the core types that
satisfy the actor's \code{SAME} and \code{DIVERSE} relations. This bound
certifies all \resCertUnsatRows{} pattern-first infeasibilities of the base
inputs and the variants below (\resCertCertified{} without the placement
step). In all \resCertThrAbove{} cases where a pattern-first flow's
least period exceeds $\mu^{*}$, the bound equals that period exactly
(\resCertThrAboveExact). So every pattern-first infeasibility holds for the
deployed system, independently of the solver.

\begin{figure}[!t]
  \centering
  \includegraphics{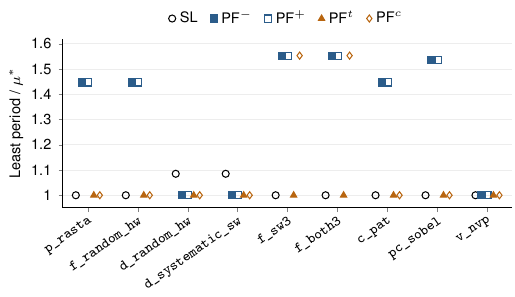}
  \caption{Least period of each decoupled flow relative to the joint
  minimum period $\mu^{*}$ (base inputs). The pattern-first rules that fix
  2-of-2 for a SIL-3 actor need \resThrBasePFMinMinAboveRatio{} to
  \resThrBasePFMinMaxRatio{} times $\mu^{*}$; SL loses
  period only on the two instances without a SIL-3 actor.}
  \label{fig:thr}
\end{figure}

\begin{table}[!t]
  \caption{Least Period per Flow: Instances Above $\mu^{*}$, Largest Ratio}
  \label{tab:rq2-thr}
  \centering\footnotesize
  \setlength{\tabcolsep}{3pt}
\begin{tabular*}{\columnwidth}{@{\extracolsep{\fill}}l c c c c c c c @{}}
\toprule
Inputs & $\mu^{*}$ & SL & PF$^{-}$ & PF$^{+}$ & PF$^{t}$ & PF$^{c}$ & NP \\
\midrule
base & \resThrBaseMuStarMin{}--\resThrBaseMuStarMax{} & \resThrBaseSLAbove{}, \resThrBaseSLMaxRatio{} & \resThrBasePFMinAbove{}, \resThrBasePFMinMaxRatio{} & \resThrBasePFMaxAbove{}, \resThrBasePFMaxMaxRatio{} & \resThrBasePFLightAbove{}, \resThrBasePFLightMaxRatio{} & \resThrBasePFCheapAbove{}, \resThrBasePFCheapMaxRatio{} & \resThrBaseNPAbove{}, \resThrBaseNPMaxRatio{} \\
\code{tok2} & \resThrTokTwoMuStarMin{}--\resThrTokTwoMuStarMax{} & \resThrTokTwoSLAbove{}, \resThrTokTwoSLMaxRatio{} & \resThrTokTwoPFMinAbove{}, \resThrTokTwoPFMinMaxRatio{} & \resThrTokTwoPFMaxAbove{}, \resThrTokTwoPFMaxMaxRatio{} & \resThrTokTwoPFLightAbove{}, \resThrTokTwoPFLightMaxRatio{} & \resThrTokTwoPFCheapAbove{}, \resThrTokTwoPFCheapMaxRatio{} & \resThrTokTwoNPAbove{}, \resThrTokTwoNPMaxRatio{} \\
\code{chk10} & \resThrChkTenMuStarMin{}--\resThrChkTenMuStarMax{} & \resThrChkTenSLAbove{}, \resThrChkTenSLMaxRatio{} & \resThrChkTenPFMinAbove{}, \resThrChkTenPFMinMaxRatio{} & \resThrChkTenPFMaxAbove{}, \resThrChkTenPFMaxMaxRatio{} & \resThrChkTenPFLightAbove{}, \resThrChkTenPFLightMaxRatio{} & \resThrChkTenPFCheapAbove{}, \resThrChkTenPFCheapMaxRatio{} & \resThrChkTenNPAbove{}, \resThrChkTenNPMaxRatio{} \\
\code{chk50} & \resThrChkFiftyMuStarMin{}--\resThrChkFiftyMuStarMax{} & \resThrChkFiftySLAbove{}, \resThrChkFiftySLMaxRatio{} & \resThrChkFiftyPFMinAbove{}, \resThrChkFiftyPFMinMaxRatio{} & \resThrChkFiftyPFMaxAbove{}, \resThrChkFiftyPFMaxMaxRatio{} & \resThrChkFiftyPFLightAbove{}, \resThrChkFiftyPFLightMaxRatio{} & \resThrChkFiftyPFCheapAbove{}, \resThrChkFiftyPFCheapMaxRatio{} & \resThrChkFiftyNPAbove{}, \resThrChkFiftyNPMaxRatio{} \\
\bottomrule
\end{tabular*}

\vspace{2pt}\parbox{\columnwidth}{\scriptsize Least period of each flow under its own restriction (THROUGHPUT), over the \resThrBaseInstances{} instances. Cell: instances whose least period exceeds the joint $\mu^{*}$, and the largest ratio to $\mu^{*}$. Range of $\mu^{*}$ over the instances in the second column.}

\end{table}

The result depends on the catalog's timing and pricing, so we varied both
(\cref{tab:rq2-all}). With two tokens on each verdict edge (\code{tok2},
fault reaction within two periods), no flow is infeasible at any bound, and
PF$^{c}$ matches joint exploration everywhere. The remaining gaps are costs,
up to \resVarTokTwoTwoHundredPFLightMaxLoss{}\,\% for PF$^{t}$ at
$2.00\,\mu^{*}$. With checkers at 10\,\% of their owner's WCET
(\code{chk10}), the picture is that of the base inputs. With checkers at
half their owner's WCET (\code{chk50}), $\mu^{*}$ of the RASTA instances
rises, PF$^{t}$ ties 2-of-2 on added WCET for one actor and picks it, and
it fails on \resVarChkFiftyHundredTwentyFivePFLightInfeasible{} instances
up to $1.25\,\mu^{*}$.

Consider a doer one SIL lower than its actor, as the sources
allow~\cite{koopman2016challenges}, and priced diverse development (a model
variant, \cref{sec:eval-ext}). Then PF$^{c}$ and PF$^{t}$ reach the joint
optimum everywhere, while PF$^{-}$ and PF$^{+}$ stay infeasible as above and
lose up to \resModExtBothPFMinMaxLossPct{} and
\resModExtBothPFMaxMaxLossPct{}\,\% where feasible (\cref{tab:modelext-flows}).
Which rule works thus depends on the cost model, while joint exploration
was feasible and optimal in every case.

\emph{Exactness.} Because of the gap of \cref{sec:model-throughput}, we
re-solved all \resExactRows{} RQ2 and RQ4 runs with exact variant~B, with
the same builds, bounds, objectives, fixings, and workers, and verified
every solution. \resExactConclusive{} runs were conclusive in both, and
\resExactAgree{} of them agree in status and value; all
\resExactUnsatRows{} infeasibilities, including the \resExactSLUnsatRows{}
SL ones that the cycle bound does not cover, are infeasible under variant~B
too. The one difference is an SL run in which variant~B finds a slightly
cheaper design on the same step-1 binding. \resExactInconclusive{} runs
stayed open because variant~B reached the time limit, all at loose bounds
where the model proves optimality in seconds; in all of them, variant~B's
incumbent equals the model's value. SL depends on which step-1 optimum the
solver returns, as step~1 has several optima with different bindings; the
canonical tie-break above makes it reproducible.

\emph{Three applications.} With Sobel (SIL~2), SUSAN (SIL~0), and RASTA on
one platform (\code{m\_3app}), the joint optimum costs \resMultiCost{} on
\resMultiNprocs{} cores (median \resMultiTime{}\,s, range
\resMultiTimeMin{} to \resMultiTimeMax{}\,s). Bounds are
$\lceil f\mu^{*}_z\rceil$ per application, with $\mu^{*}_z$ the least
period of the application alone (\cref{tab:multi}). Nothing is feasible for
$f\le1.10$: alone, each application reaches these periods; together, they
compete for the SIL-4 cores. PF$^{-}$ and PF$^{+}$ fail at $1.25$, and SL
fails at every bound. The variant-B cross-check of the bounded runs agrees
on all \resMultiExactAgree{} that are conclusive in both runs, with no
better incumbent among the \resMultiExactOpen{} open ones. A longer
variant-B run confirmed the joint infeasibility at $1.10$
(\resMultiHundredTenJointExactLong{} in
\resMultiHundredTenJointExactLongTime{}\,s), so every infeasibility holds
for the deployed system.

\begin{table}[!t]
  \caption{Three Applications (\code{m\_3app}) per Bound and Flow}
  \label{tab:multi}
  \centering\footnotesize
  \setlength{\tabcolsep}{3pt}
\begin{tabular*}{\columnwidth}{@{\extracolsep{\fill}}l c c c c c c c @{}}
\toprule
Bound & Joint & SL & PF$^{-}$ & PF$^{+}$ & PF$^{t}$ & PF$^{c}$ & NP \\
\midrule
none & \resMultiNoneJointCost{} & +\resMultiNoneSLLoss{}\,\% & +\resMultiNonePFMinLoss{}\,\% & +\resMultiNonePFMaxLoss{}\,\% & +\resMultiNonePFLightLoss{}\,\% & +\resMultiNonePFCheapLoss{}\,\% & +\resMultiNoneNPLoss{}\,\% \\
$1.00\,\mu^{*}_z$ & $\infty$ & $\infty$ & $\infty$ & $\infty$ & $\infty$ & $\infty$ & $\infty$ \\
$1.10\,\mu^{*}_z$ & $\infty$ & $\infty$ & $\infty$ & $\infty$ & $\infty$ & $\infty$ & $\infty$ \\
$1.25\,\mu^{*}_z$ & \resMultiHundredTwentyFiveJointCost{} & $\infty$ & $\infty$ & $\infty$ & +\resMultiHundredTwentyFivePFLightLoss{}\,\% & +\resMultiHundredTwentyFivePFCheapLoss{}\,\% & +\resMultiHundredTwentyFiveNPLoss{}\,\% \\
$1.50\,\mu^{*}_z$ & \resMultiHundredFiftyJointCost{} & $\infty$ & +\resMultiHundredFiftyPFMinLoss{}\,\% & +\resMultiHundredFiftyPFMaxLoss{}\,\% & +\resMultiHundredFiftyPFLightLoss{}\,\% & +\resMultiHundredFiftyPFCheapLoss{}\,\% & +\resMultiHundredFiftyNPLoss{}\,\% \\
$2.00\,\mu^{*}_z$ & \resMultiTwoHundredJointCost{} & $\infty$ & +\resMultiTwoHundredPFMinLoss{}\,\% & +\resMultiTwoHundredPFMaxLoss{}\,\% & +\resMultiTwoHundredPFLightLoss{}\,\% & +\resMultiTwoHundredPFCheapLoss{}\,\% & +\resMultiTwoHundredNPLoss{}\,\% \\
\bottomrule
\end{tabular*}

\vspace{2pt}\parbox{\columnwidth}{\scriptsize \code{m\_3app}. Joint: proven minimum of $C_{\mathrm{tot}}$; flows: extra cost over joint; $\infty$: proven infeasible. $\mu^{*}_z$: least period of application $z$ alone (Sobel \resMultiMuAloneSobel{}, SUSAN \resMultiMuAloneSusan{}, RASTA \resMultiMuAloneRasta{}).}

\end{table}

\emph{Promotion.} No optimum above promotes an actor, which is why NP equals
joint. With the core count bounded below that of the unbounded optimum
(\cref{tab:patcons}), every feasible joint optimum promotes
(\resPatConsPromoMin{} to \resPatConsPromoMax{} units), and NP is
infeasible in all \resPatConsNPInfeasiblePairs{} cases: promotion is what
allows consolidation with patterns. One run, \code{m\_3app} with at most four
cores, reached the time limit, and its incumbent fails the verifier by
period only (the gap of \cref{sec:model-throughput}); it is left out.

\begin{table}[!t]
  \caption{Pattern Instances With a Bounded Core Count}
  \label{tab:patcons}
  \centering\footnotesize
  \setlength{\tabcolsep}{3pt}
\begin{tabular*}{\columnwidth}{@{\extracolsep{\fill}}l c c r r c@{}}
\toprule
Instance & Cores & $k$ & Joint & Promotion & NP \\
\midrule
\code{p\_rasta} & 2 & 1 & $\infty$ & -- & $\infty$ \\
\code{f\_random\_hw} & 5 & 4 & 3634 & 293 & $\infty$ \\
 &  & 3 & 3919 & 586 & $\infty$ \\
 &  & 2 & 5506 & 2206 & $\infty$ \\
 &  & 1 & $\infty$ & -- & $\infty$ \\
\code{d\_random\_hw} & 3 & 2 & 2297 & 448 & $\infty$ \\
 &  & 1 & $\infty$ & -- & $\infty$ \\
\code{d\_systematic\_sw} & 2 & 1 & 2286 & 448 & $\infty$ \\
\code{f\_sw3} & 4 & 3 & 3805 & 448 & $\infty$ \\
 &  & 2 & 5530 & 2206 & $\infty$ \\
 &  & 1 & $\infty$ & -- & $\infty$ \\
\code{f\_both3} & 5 & 4 & 3659 & 293 & $\infty$ \\
 &  & 3 & 3944 & 586 & $\infty$ \\
 &  & 2 & 5531 & 2206 & $\infty$ \\
 &  & 1 & $\infty$ & -- & $\infty$ \\
\code{c\_pat} & 5 & 4 & 3634 & 293 & $\infty$ \\
 &  & 3 & 3919 & 586 & $\infty$ \\
 &  & 2 & 5506 & 2206 & $\infty$ \\
 &  & 1 & $\infty$ & -- & $\infty$ \\
\code{pc\_sobel} & 3 & 2 & 2608 & 1215 & $\infty$ \\
 &  & 1 & $\infty$ & -- & $\infty$ \\
\code{v\_nvp} & 6 & 5 & 5933 & 293 & $\infty$ \\
 &  & 4 & 6218 & 586 & $\infty$ \\
 &  & 3 & 7805 & 2206 & $\infty$ \\
 &  & 2 & $\infty$ & -- & $\infty$ \\
\code{m\_3app} & 6 & 5 & 6108 & 448 & $\infty$ \\
 &  & 4 & t/o & -- & $\infty$ \\
 &  & 3 & 10191 & 4550 & $\infty$ \\
 &  & 2 & $\infty$ & -- & $\infty$ \\
\bottomrule
\end{tabular*}

\vspace{2pt}\parbox{\columnwidth}{\scriptsize Cores: of the unbounded joint optimum. $k$: at most $k$ cores. Joint: proven minimum of $C_{\mathrm{tot}}$, Promotion: its $C_{\mathrm{pro}}$; NP: the no-promotion flow. $\infty$: proven infeasible; t/o: time limit, incumbent not verified (\cref{sec:discussion}).}

\end{table}

\subsection{Model Extensions: Doer SIL, Priced Diversity, Verdict Tokens}\label{sec:eval-ext}
The model develops every component and the doer at the actor's SIL and
prices no diverse development; verdict tokens are catalog data. To measure
what this conservative approximation costs, we evaluated three extensions
in a copy of the model, not in the tool. E1 develops the doer at the actor's
SIL minus the pattern's \code{sil\_offset}. E2 applies the catalog's
development multipliers per role of the selected pattern, with the diverse
multiplier where the pattern places the component \code{DIVERSE}. E3 makes
the verdict tokens a decision $r\in\{1,2\}$ per actor under
$r\mu\le\kappa B$, a stand-in for a process safety time $\kappa B$ that no
input states. Each solution was verified by the unmodified verifier on an
instance into which the decisions were projected. With neutral data, the
copy reproduces the stored optima, and E3 with $\kappa=2$ reproduces the
two-token catalog variant on all \resModExtTokTwoVariantCompared{} runs
with the same bound.

\begin{table}[!t]
  \caption{Model Extensions Against the Model (Joint Optima)}
  \label{tab:modelext}
  \centering\footnotesize
  \setlength{\tabcolsep}{3pt}
\begin{tabular*}{\columnwidth}{@{\extracolsep{\fill}}l c c c c c@{}}
\toprule
Extension & Rows & Cheaper & Dearer & Pat. & Time \\
\midrule
E1: doer one SIL lower & \resModExtDoerJointRows{} & \resModExtDoerCheaper{} (\resModExtDoerMinSavingPct{}--\resModExtDoerMaxSavingPct{}) & \resModExtDoerDearer{} & \resModExtDoerPatternsChanged{} & \resModExtDoerTimeRatioMed{} \\
E2: priced diversity & \resModExtDiverseJointRows{} & \resModExtDiverseCheaper{} (\resModExtDiverseMinSavingPct{}--\resModExtDiverseMaxSavingPct{}) & \resModExtDiverseDearer{} (\resModExtDiverseMinExtraPct{}--\resModExtDiverseMaxExtraPct{}) & \resModExtDiversePatternsChanged{} & \resModExtDiverseTimeRatioMed{} \\
E1 and E2 & \resModExtBothJointRows{} & \resModExtBothCheaper{} (\resModExtBothMinSavingPct{}--\resModExtBothMaxSavingPct{}) & \resModExtBothDearer{} (\resModExtBothMinExtraPct{}--\resModExtBothMaxExtraPct{}) & \resModExtBothPatternsChanged{} & \resModExtBothTimeRatioMed{} \\
E3, $\kappa=1$ & \resModExtTokOneJointRows{} & \resModExtTokOneCheaper{} & \resModExtTokOneDearer{} & \resModExtTokOnePatternsChanged{} & \resModExtTokOneTimeRatioMed{} \\
E3, $\kappa=1.5$ & \resModExtTokOneHalfJointRows{} & \resModExtTokOneHalfCheaper{} (\resModExtTokOneHalfMinSavingPct{}--\resModExtTokOneHalfMaxSavingPct{}) & \resModExtTokOneHalfDearer{} & \resModExtTokOneHalfPatternsChanged{} & \resModExtTokOneHalfTimeRatioMed{} \\
E3, $\kappa=2$ & \resModExtTokTwoJointRows{} & \resModExtTokTwoCheaper{} (\resModExtTokTwoMinSavingPct{}--\resModExtTokTwoMaxSavingPct{}) & \resModExtTokTwoDearer{} & \resModExtTokTwoPatternsChanged{} & \resModExtTokTwoTimeRatioMed{} \\
\bottomrule
\end{tabular*}

\vspace{2pt}\parbox{\columnwidth}{\scriptsize Joint optima in a model copy against the model, on the \resModExtDoerInstances{} RQ2 instances, without a bound and at the bounds of \cref{sec:eval-rq2}. Cheaper/Dearer: rows and the range of the change in percent. Pat.: rows whose selected patterns change. Time: median ratio of the solve time. E3: verdict tokens $r\in\{1,2\}$ per actor with $r\mu\le\kappa B$.}

\end{table}

\begin{table}[!t]
  \caption{Decoupled Flows Under the Extended Models}
  \label{tab:modelext-flows}
  \centering\footnotesize
  \setlength{\tabcolsep}{1.5pt}
\begin{tabular*}{\columnwidth}{@{\extracolsep{\fill}}l c c c c c @{}}
\toprule
Model & SL & PF$^{-}$ & PF$^{+}$ & PF$^{t}$ & PF$^{c}$ \\
\midrule
E1+E2, inf. & \resModExtBothSLHundredInfeasible{}/\resModExtBothSLHundredTenInfeasible{}/\resModExtBothSLHundredTwentyFiveInfeasible{}/\resModExtBothSLHundredFiftyInfeasible{}/\resModExtBothSLTwoHundredInfeasible{} & \resModExtBothPFMinHundredInfeasible{}/\resModExtBothPFMinHundredTenInfeasible{}/\resModExtBothPFMinHundredTwentyFiveInfeasible{}/\resModExtBothPFMinHundredFiftyInfeasible{}/\resModExtBothPFMinTwoHundredInfeasible{} & \resModExtBothPFMaxHundredInfeasible{}/\resModExtBothPFMaxHundredTenInfeasible{}/\resModExtBothPFMaxHundredTwentyFiveInfeasible{}/\resModExtBothPFMaxHundredFiftyInfeasible{}/\resModExtBothPFMaxTwoHundredInfeasible{} & \resModExtBothPFLightHundredInfeasible{}/\resModExtBothPFLightHundredTenInfeasible{}/\resModExtBothPFLightHundredTwentyFiveInfeasible{}/\resModExtBothPFLightHundredFiftyInfeasible{}/\resModExtBothPFLightTwoHundredInfeasible{} & \resModExtBothPFCheapHundredInfeasible{}/\resModExtBothPFCheapHundredTenInfeasible{}/\resModExtBothPFCheapHundredTwentyFiveInfeasible{}/\resModExtBothPFCheapHundredFiftyInfeasible{}/\resModExtBothPFCheapTwoHundredInfeasible{} \\
\phantom{x} loss (\%) & \resModExtBothSLMaxLossPct{} & \resModExtBothPFMinMaxLossPct{} & \resModExtBothPFMaxMaxLossPct{} & \resModExtBothPFLightMaxLossPct{} & \resModExtBothPFCheapMaxLossPct{} \\
E3, $\kappa=1$, inf. & \resModExtTokOneSLHundredInfeasible{}/\resModExtTokOneSLHundredTenInfeasible{}/\resModExtTokOneSLHundredTwentyFiveInfeasible{}/\resModExtTokOneSLHundredFiftyInfeasible{}/\resModExtTokOneSLTwoHundredInfeasible{} & \resModExtTokOnePFMinHundredInfeasible{}/\resModExtTokOnePFMinHundredTenInfeasible{}/\resModExtTokOnePFMinHundredTwentyFiveInfeasible{}/\resModExtTokOnePFMinHundredFiftyInfeasible{}/\resModExtTokOnePFMinTwoHundredInfeasible{} & \resModExtTokOnePFMaxHundredInfeasible{}/\resModExtTokOnePFMaxHundredTenInfeasible{}/\resModExtTokOnePFMaxHundredTwentyFiveInfeasible{}/\resModExtTokOnePFMaxHundredFiftyInfeasible{}/\resModExtTokOnePFMaxTwoHundredInfeasible{} & \resModExtTokOnePFLightHundredInfeasible{}/\resModExtTokOnePFLightHundredTenInfeasible{}/\resModExtTokOnePFLightHundredTwentyFiveInfeasible{}/\resModExtTokOnePFLightHundredFiftyInfeasible{}/\resModExtTokOnePFLightTwoHundredInfeasible{} & \resModExtTokOnePFCheapHundredInfeasible{}/\resModExtTokOnePFCheapHundredTenInfeasible{}/\resModExtTokOnePFCheapHundredTwentyFiveInfeasible{}/\resModExtTokOnePFCheapHundredFiftyInfeasible{}/\resModExtTokOnePFCheapTwoHundredInfeasible{} \\
\phantom{x} loss (\%) & \resModExtTokOneSLMaxLossPct{} & \resModExtTokOnePFMinMaxLossPct{} & \resModExtTokOnePFMaxMaxLossPct{} & \resModExtTokOnePFLightMaxLossPct{} & \resModExtTokOnePFCheapMaxLossPct{} \\
E3, $\kappa=1.5$, inf. & \resModExtTokOneHalfSLHundredInfeasible{}/\resModExtTokOneHalfSLHundredTenInfeasible{}/\resModExtTokOneHalfSLHundredTwentyFiveInfeasible{}/\resModExtTokOneHalfSLHundredFiftyInfeasible{}/\resModExtTokOneHalfSLTwoHundredInfeasible{} & \resModExtTokOneHalfPFMinHundredInfeasible{}/\resModExtTokOneHalfPFMinHundredTenInfeasible{}/\resModExtTokOneHalfPFMinHundredTwentyFiveInfeasible{}/\resModExtTokOneHalfPFMinHundredFiftyInfeasible{}/\resModExtTokOneHalfPFMinTwoHundredInfeasible{} & \resModExtTokOneHalfPFMaxHundredInfeasible{}/\resModExtTokOneHalfPFMaxHundredTenInfeasible{}/\resModExtTokOneHalfPFMaxHundredTwentyFiveInfeasible{}/\resModExtTokOneHalfPFMaxHundredFiftyInfeasible{}/\resModExtTokOneHalfPFMaxTwoHundredInfeasible{} & \resModExtTokOneHalfPFLightHundredInfeasible{}/\resModExtTokOneHalfPFLightHundredTenInfeasible{}/\resModExtTokOneHalfPFLightHundredTwentyFiveInfeasible{}/\resModExtTokOneHalfPFLightHundredFiftyInfeasible{}/\resModExtTokOneHalfPFLightTwoHundredInfeasible{} & \resModExtTokOneHalfPFCheapHundredInfeasible{}/\resModExtTokOneHalfPFCheapHundredTenInfeasible{}/\resModExtTokOneHalfPFCheapHundredTwentyFiveInfeasible{}/\resModExtTokOneHalfPFCheapHundredFiftyInfeasible{}/\resModExtTokOneHalfPFCheapTwoHundredInfeasible{} \\
\phantom{x} loss (\%) & \resModExtTokOneHalfSLMaxLossPct{} & \resModExtTokOneHalfPFMinMaxLossPct{} & \resModExtTokOneHalfPFMaxMaxLossPct{} & \resModExtTokOneHalfPFLightMaxLossPct{} & \resModExtTokOneHalfPFCheapMaxLossPct{} \\
\bottomrule
\end{tabular*}

\vspace{2pt}\parbox{\columnwidth}{\scriptsize inf.: infeasible instances at $1.00/1.10/1.25/1.50/2.00\,\mu^{*}$; max.\ loss: largest extra cost over the extended model's joint optimum.}

\end{table}

\Cref{tab:modelext} summarizes the joint optima. E1 makes \resModExtDoerCheaper{}
of \resModExtDoerJointRows{} optima cheaper, by up to
\resModExtDoerMaxSavingPct{}\,\%; E2 makes \resModExtDiverseDearer{} dearer,
by up to \resModExtDiverseMaxExtraPct{}\,\% (diverse 2-of-2 and checker),
and \resModExtDiverseCheaper{} cheaper (the NVP voter's multiplier is below
one). Together, the patterns of \resModExtBothPatternsChanged{} optima
change, at \resModExtBothTimeRatioMed{} times the median solve time. E3 with
$\kappa=1$ changes nothing, since two tokens need $\mu\le B/2$; with
$\kappa=2$, \resModExtTokTwoCheaper{} optima become cheaper, by up to
\resModExtTokTwoMaxSavingPct{}\,\%, at \resModExtTokTwoVariantTimeRatioMed{}
times the time of the fixed two-token catalog. Under E1 and E2 together,
the mixed-SIL doer/checker (on S and H+S) and the isolated checker (on H)
become the cheapest and the lightest SIL-3 choice. So PF$^{c}$ and
PF$^{t}$ reach the joint optimum on every instance and bound, while
PF$^{-}$ and PF$^{+}$ keep their infeasibilities (\cref{tab:modelext-flows}).
The cost beyond solve time is in the claims. E1 needs a goal that a lower
doer SIL is acceptable, E3 needs a process safety time per function, and
each needs the verifier's requirement and token semantics changed in step.

\subsection{RQ3: Scalability}\label{sec:eval-rq3}
We generated consistent, single-rate SDF graphs with the random generator of
SDF3~\cite{stuijk2006sdf3}, \resSetupScalReps{} per size from
\resSetupScalMinActors{} to \resSetupScalMaxActors{} actors, with WCETs for
\resSetupScalTypes{} core types. The WCETs are the generator's own
execution times per processor type; a type it leaves out for an actor
forbids that binding. The generator has no seed option, so the generated graphs are
stored in the artifact; reproducibility comes from them. Each instance has
\resSetupScalSlots{} core slots in FCRs of one or two cores (each type
spans at least two FCRs, so \code{DIFFERENT\_FCR} can hold), fault model H,
and a quarter of its actors each at SIL~3 and SIL~2. A throughput solve of
at most \resSetupScalMuTimeLimit{}\,s gives the best period $\mu^{*}$ it
finds as the period bound; the cost solve then has
\resSetupScalTimeLimit{}\,s. A second sweep varies the catalog size at
\resSetupScalCatActors{} actors, with prefixes of one fixed pattern order,
so the catalogs are nested.

\Cref{fig:scalability,tab:scal} show where exact search stops. All graphs
with 8 and 16 actors are proven optimal (medians \resScalEightTime{}\,s and
\resScalSixteenTime{}\,s). With 24 actors, \resScalTwentyFourProven{} of
\resSetupScalReps{} are proven (median \resScalTwentyFourTime{}\,s) and
\resScalTwentyFourTimeout{} reach the time limit. With 32 actors,
\resScalThirtyTwoProven{} is proven, \resScalThirtyTwoTimeout{} reach the
limit, and for \resScalThirtyTwoNoMu{} the throughput solve finds no period
in time, so there is no cost solve; with 48 and 64 actors, this holds for
all graphs. $\mu^{*}$ is proven for \resScalSixteenMuProven{},
\resScalTwentyFourMuProven{}, and \resScalThirtyTwoMuProven{} of the graphs
with 16, 24, and 32 actors. Only \resScalGapCount{} timed-out cost solve has
a verified incumbent (gap \resScalGapMax\,\%); the two 32-actor timeouts
fail the verifier by period (\cref{sec:model-throughput}), so their gaps
are not reported.

Model size limits as much as search (\cref{fig:flat}). Flattening the
throughput model takes \resFlatTwentyFourTime{}\,s at 24 actors,
\resFlatFortyEightTime{}\,s at 48, and \resFlatSixtyFourTime{}\,s at 64.
There, it writes \resFlatSixtyFourMB{}\,MB of FlatZinc with
\resFlatSixtyFourVars{}k variables, the compiler needs up to
\resFlatSixtyFourRssMB{}\,MB, and flattening alone uses almost half of the
throughput budget. The instances of \cref{tab:results} flatten
in at most \resFlatTableTimeMax{}\,s to \resFlatTableVarsMin{}k to
\resFlatTableVarsMax{}k variables.

The catalog size matters less. With only the pattern \code{none}, the
front-end rejects all \resScalCatOneBuildFail{} instances, as no SIL-3 actor
can be served. With 3, 5, and 9 patterns, \resScalCatThreeProven{},
\resScalCatFiveProven{}, and \resScalCatNineProven{} of
\resSetupScalReps{} instances are proven, with medians of
\resScalCatThreeTime{}, \resScalCatFiveTime{}, and
\resScalCatNineTime{}\,s, and the FlatZinc grows from
\resScalCatThreeFlatVars{} to \resScalCatNineFlatVars{} variables.

\begin{figure}[!t]
  \centering
  \includegraphics{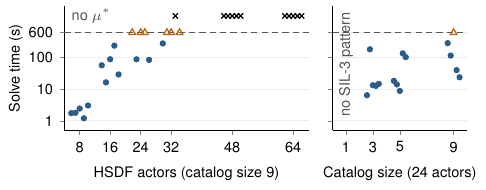}
  \caption{Solve time on synthetic instances, one mark per graph (counts in
  \cref{tab:scal}). Filled: proven optimal; open: time limit reached;
  cross: no period bound $\mu^{*}$ found, so no cost solve. Beyond 16
  actors, most searches reach a limit.}
  \label{fig:scalability}
\end{figure}

\begin{figure}[!t]
  \centering
  \includegraphics{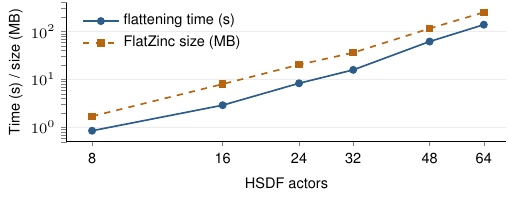}
  \caption{Cost of flattening the throughput model of the synthetic graphs
  (median per actor count, logarithmic axes): compile time and FlatZinc
  size both grow steeply with the number of HSDF actors.}
  \label{fig:flat}
\end{figure}

\begin{table*}[!t]
  \caption{RQ3 Actor Sweep: Outcome, Solve Time, and Model Size}
  \label{tab:scal}
  \centering\footnotesize
  \setlength{\tabcolsep}{2.5pt}
\begin{tabular*}{\textwidth}{@{\extracolsep{\fill}}r c c c c c c r r r r@{}}
\toprule
Actors & Graphs & $\mu^{*}$ found & $\mu^{*}$ proven & Proven & Time limit & No $\mu^{*}$ & Median (s) & Flatten (s) & FlatZinc (MB) & Memory (MB) \\
\midrule
8 & \resScalEightReps{} & \resScalEightMuFound{} & \resScalEightMuProven{} & \resScalEightProven{} & \resScalEightTimeout{} & \resScalEightNoMu{} & \resScalEightTime{} & \resFlatEightTime{} & \resFlatEightMB{} & \resFlatEightRssMB{} \\
16 & \resScalSixteenReps{} & \resScalSixteenMuFound{} & \resScalSixteenMuProven{} & \resScalSixteenProven{} & \resScalSixteenTimeout{} & \resScalSixteenNoMu{} & \resScalSixteenTime{} & \resFlatSixteenTime{} & \resFlatSixteenMB{} & \resFlatSixteenRssMB{} \\
24 & \resScalTwentyFourReps{} & \resScalTwentyFourMuFound{} & \resScalTwentyFourMuProven{} & \resScalTwentyFourProven{} & \resScalTwentyFourTimeout{} & \resScalTwentyFourNoMu{} & \resScalTwentyFourTime{} & \resFlatTwentyFourTime{} & \resFlatTwentyFourMB{} & \resFlatTwentyFourRssMB{} \\
32 & \resScalThirtyTwoReps{} & \resScalThirtyTwoMuFound{} & \resScalThirtyTwoMuProven{} & \resScalThirtyTwoProven{} & \resScalThirtyTwoTimeout{} & \resScalThirtyTwoNoMu{} & \resScalThirtyTwoTime{} & \resFlatThirtyTwoTime{} & \resFlatThirtyTwoMB{} & \resFlatThirtyTwoRssMB{} \\
48 & -- & \resScalFortyEightMuFound{} & \resScalFortyEightMuProven{} & -- & -- & \resScalFortyEightNoMu{} & -- & \resFlatFortyEightTime{} & \resFlatFortyEightMB{} & \resFlatFortyEightRssMB{} \\
64 & -- & \resScalSixtyFourMuFound{} & \resScalSixtyFourMuProven{} & -- & -- & \resScalSixtyFourNoMu{} & -- & \resFlatSixtyFourTime{} & \resFlatSixtyFourMB{} & \resFlatSixtyFourRssMB{} \\
\bottomrule
\end{tabular*}

\vspace{2pt}\parbox{\textwidth}{\scriptsize Actor sweep, full catalog. $\mu^{*}$ found/proven: the throughput solve (\resSetupScalMuTimeLimit{}\,s) found or proved a period. Proven / Time limit: cost solve (\resSetupScalTimeLimit{}\,s) proved optimal / reached the limit. No $\mu^{*}$: no period found, hence no cost solve. Median: solve time of the proven graphs. Flatten, FlatZinc: median compile time and FlatZinc size of the throughput model; Memory: largest peak memory of the compiler. Graphs: graphs that built (--: all).}

\end{table*}

\subsection{RQ4: Sensitivity and Generated Arguments}\label{sec:eval-rq4}
\emph{Cost profile and hardware price.} For \cref{fig:sensitivity}, we
scaled all hardware prices by \resPriceMinFactor{} to \resPriceMaxFactor{}
and solved the running example without patterns, unbounded and with at
most three, two, and one cores (\cref{tab:price}). With Klosterman's flat
profile, one core is optimal from factor
\resPriceKlostermanConsolidateFactor{}
(\resPriceKlostermanLowConsolidateFactor{} with his lower SIL-3 row). With
the steep profiles of Myklebust et al.\ and DO-178B, the optimum keeps
\resPriceMyklebustSixteenNprocs{} cores, while the extra cost of one core
falls from \resPriceMyklebustOneOneExtra{} to
\resPriceMyklebustSixteenOneExtra{}\,\% and from
\resPriceDoOneSeventyEightBOneOneExtra{} to
\resPriceDoOneSeventyEightBSixteenOneExtra{}\,\%. The promotion cost of the
one-core design depends on the profile, not on the factor
(\resPriceMyklebustOneOnePromotion, \resPriceKlostermanOneOnePromotion,
\resPriceKlostermanLowOneOnePromotion, and
\resPriceDoOneSeventyEightBOneOnePromotion{} units).
\Cref{fig:profcons} shows the same trade-off at factor four over the core
bound: the profile sets the price of consolidation, not only the cost
level. Before the Klosterman profile took the higher of the two SIL-3 rows,
its optimum at factor four used one core; with the higher row it uses
\resPriceKlostermanFourNprocs{}, and with the lower row
\resPriceKlostermanLowFourNprocs{}. A flip of the architecture between
profiles therefore rests on one row of a published table.
With the original prices and the default profile, the promotion cost of the
optimum with at most three, two, and one cores is
\resConsolidationThreePromotion, \resConsolidationTwoPromotion, and
\resConsolidationOnePromotion{} units; the price of each further step
grows steeply.

\begin{figure}[!t]
  \centering
  \includegraphics{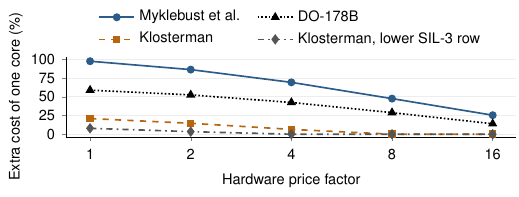}
  \caption{Extra cost of the cheapest one-core design of the running
  example without patterns over its optimum, per cost profile and hardware
  price factor. Only the flat profiles reach zero (one core optimal).}
  \label{fig:sensitivity}
\end{figure}

\begin{figure}[!t]
  \centering
  \includegraphics{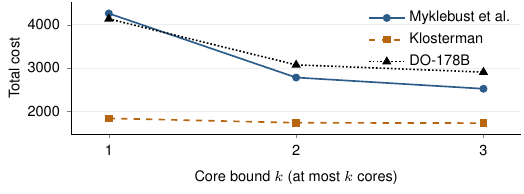}
  \caption{Total cost of the optimum with at most $k$ cores, running example
  without patterns at four times the hardware prices. The steep profiles
  pay heavily for the third consolidation step; Klosterman's profile hardly
  at all.}
  \label{fig:profcons}
\end{figure}

\begin{table}[!t]
  \caption{Hardware Price Sweep per Cost Profile (Total Cost)}
  \label{tab:price}
  \centering\footnotesize
  \setlength{\tabcolsep}{3pt}
\begin{tabular*}{\columnwidth}{@{\extracolsep{\fill}}l r r r r r r@{}}
\toprule
Profile & $\times$ & Opt.\ (cores) & $k\le3$ & $k\le2$ & $k\le1$ & Pro.\ \\
\midrule
Myklebust et al. & 1 & 1914 (3) & 1914 & 2199 & 3786 & 1913 \\
 & 2 & 2115 (3) & 2115 & 2392 & 3946 & 1913 \\
 & 4 & 2517 (3) & 2517 & 2778 & 4266 & 1913 \\
 & 8 & 3321 (3) & 3321 & 3550 & 4906 & 1913 \\
 & 16 & 4929 (3) & 4929 & 5094 & 6186 & 1913 \\
\addlinespace[2pt]
Klosterman & 1 & 1116 (3) & 1116 & 1148 & 1350 & 275 \\
 & 2 & 1317 (3) & 1317 & 1341 & 1510 & 275 \\
 & 4 & 1719 (3) & 1719 & 1727 & 1830 & 275 \\
 & 8 & 2470 (1) & 2470 & 2470 & 2470 & 275 \\
 & 16 & 3750 (1) & 3750 & 3750 & 3750 & 275 \\
\addlinespace[2pt]
Klosterman (lower row) & 1 & 1056 (3) & 1056 & 1065 & 1140 & 125 \\
 & 2 & 1257 (3) & 1257 & 1258 & 1300 & 125 \\
 & 4 & 1620 (1) & 1620 & 1620 & 1620 & 125 \\
 & 8 & 2260 (1) & 2260 & 2260 & 2260 & 125 \\
 & 16 & 3540 (1) & 3540 & 3540 & 3540 & 125 \\
\addlinespace[2pt]
DO-178B & 1 & 2301 (3) & 2301 & 2493 & 3660 & 1400 \\
 & 2 & 2502 (3) & 2502 & 2686 & 3820 & 1400 \\
 & 4 & 2904 (3) & 2904 & 3072 & 4140 & 1400 \\
 & 8 & 3708 (3) & 3708 & 3844 & 4780 & 1400 \\
 & 16 & 5316 (3) & 5316 & 5388 & 6060 & 1400 \\
\bottomrule
\end{tabular*}

\vspace{2pt}\parbox{\columnwidth}{\scriptsize RASTA without patterns on \code{mixed\_noiso.xml}, all hardware prices $\times$ the factor. Opt.: the unbounded optimum of $C_{\mathrm{tot}}$ and its cores; $k\le n$: the optimum with at most $n$ cores; Pro.: $C_{\mathrm{pro}}$ of the one-core optimum. All proven and verified.}

\end{table}

\emph{Fault model.} The fault model changes placement, not structure. With
random hardware faults, the optimum of the SIL-2 variant uses
\resPlacementRandomHwNprocs{} cores, since each checker must run on another
core; with systematic faults, a checker may share its doer's core, and the
optimum needs \resPlacementSystematicSwNprocs{}. The NVP case study is
optimal in \resNvpTime{}\,s on \resNvpNprocs{} cores.

\emph{Communication.} On the two instances with the bus, an ideal bus
changes neither cost nor $\mu^{*}$ (\cref{tab:comm}), so TDMA decides
nothing there. The communication assumptions matter more. Without the
black-channel exemption, the cost rises by \resCommCPatSilInheritCostPct{}
and \resCommPcSobelSilInheritCostPct{}\,\%, while subjecting transfers to
their core's rule changes nothing. With cross-checks on the bus, $\mu^{*}$
rises by \resCommCPatScopeAllMuPct{} and \resCommPcSobelScopeAllMuPct{}\,\%
and the superposed graph grows several times. So the exemption decides
cost, and the dedicated-link assumption decides throughput. The minimum
period of Sobel with a SIL-3 pattern and TDMA, \resCommSobelMu, stays
within the published bound of \resCommSobelBound.

\begin{table}[!t]
  \caption{Communication Variants}
  \label{tab:comm}
  \centering\footnotesize
  \setlength{\tabcolsep}{3pt}
\begin{tabular*}{\columnwidth}{@{\extracolsep{\fill}}l rrr rrr@{}}
\toprule
 & \multicolumn{3}{c}{\code{c\_pat}} & \multicolumn{3}{c}{\code{pc\_sobel}} \\
\cmidrule(lr){2-4}\cmidrule(l){5-7}
Variant & $|\Gplus|$ & $C_{\mathrm{tot}}$ & $\mu^{*}$ & $|\Gplus|$ & $C_{\mathrm{tot}}$ & $\mu^{*}$ \\
\midrule
TDMA bus (default) & \resCommCPatTdmaNodes{} & \resCommCPatTdmaCost{} & \resCommCPatTdmaMu{} & \resCommPcSobelTdmaNodes{} & \resCommPcSobelTdmaCost{} & \resCommPcSobelTdmaMu{} \\
ideal bus & \resCommCPatIdealNodes{} & \resCommCPatIdealCost{} & \resCommCPatIdealMu{} & \resCommPcSobelIdealNodes{} & \resCommPcSobelIdealCost{} & \resCommPcSobelIdealMu{} \\
\code{--comm-sil inherit} & \resCommCPatSilInheritNodes{} & \resCommCPatSilInheritCost{} & \resCommCPatSilInheritMu{} & \resCommPcSobelSilInheritNodes{} & \resCommPcSobelSilInheritCost{} & \resCommPcSobelSilInheritMu{} \\
\code{--comm-sil core} & \resCommCPatSilCoreNodes{} & \resCommCPatSilCoreCost{} & \resCommCPatSilCoreMu{} & \resCommPcSobelSilCoreNodes{} & \resCommPcSobelSilCoreCost{} & \resCommPcSobelSilCoreMu{} \\
\code{--comm-scope all} & \resCommCPatScopeAllNodes{} & \resCommCPatScopeAllCost{} & \resCommCPatScopeAllMu{} & \resCommPcSobelScopeAllNodes{} & \resCommPcSobelScopeAllCost{} & \resCommPcSobelScopeAllMu{} \\
\bottomrule
\end{tabular*}

\vspace{2pt}\parbox{\columnwidth}{\scriptsize $C_{\mathrm{tot}}$: proven minimum of the total cost; $\mu^{*}$: minimum period (separate THROUGHPUT solve). \code{inherit}: communication actors take the sender's SIL; \code{core}: they fall under the SIL-mixing rule of their core; \code{scope all}: pattern cross-checks are refined for the bus too.}

\end{table}

\emph{Generated arguments.} In the argument for random hardware faults,
\resGsnRandomHwSolutions{} of \resGsnRandomHwLeaves{} leaves
(\resGsnRandomHwSolutionShare{}\,\%) are Solutions that cite the check log;
under both fault classes, \resGsnBothSolutions{} of \resGsnBothLeaves{}
(\resGsnBothSolutionShare{}\,\%), and for NVP, \resGsnNvpSolutions{} of
\resGsnNvpLeaves{} (\resGsnNvpSolutionShare{}\,\%) (\cref{tab:gsn}). The
others are method, common-cause, fault-reaction, process, and partitioning
goals. The share measures what the deployment discharges, not adequacy. No
argument has a scenario out of scope: none of the three optima selects a
pattern with a scenario outside its fault model, which an earlier tie
between equal-cost patterns had hidden.

\begin{table}[!t]
  \caption{Statistics of the Generated Arguments}
  \label{tab:gsn}
  \centering\footnotesize
  \setlength{\tabcolsep}{3pt}
\begin{tabular*}{\columnwidth}{@{\extracolsep{\fill}}l ccc@{}}
\toprule
 & H & H+S & NVP (H+S) \\
\midrule
elements & \resGsnRandomHwElements{} & \resGsnBothElements{} & \resGsnNvpElements{} \\
goals & \resGsnRandomHwGoals{} & \resGsnBothGoals{} & \resGsnNvpGoals{} \\
undeveloped goals & \resGsnRandomHwUndeveloped{} & \resGsnBothUndeveloped{} & \resGsnNvpUndeveloped{} \\
Solutions & \resGsnRandomHwSolutions{} & \resGsnBothSolutions{} & \resGsnNvpSolutions{} \\
leaves & \resGsnRandomHwLeaves{} & \resGsnBothLeaves{} & \resGsnNvpLeaves{} \\
Solutions among leaves (\%) & \resGsnRandomHwSolutionShare{} & \resGsnBothSolutionShare{} & \resGsnNvpSolutionShare{} \\
records in the check log & \resGsnRandomHwChecks{} & \resGsnBothChecks{} & \resGsnNvpChecks{} \\
\code{DIVERSE} records & \resGsnRandomHwDiverse{} & \resGsnBothDiverse{} & \resGsnNvpDiverse{} \\
scenarios out of scope & \resGsnRandomHwOutOfScope{} & \resGsnBothOutOfScope{} & \resGsnNvpOutOfScope{} \\
\bottomrule
\end{tabular*}

\vspace{2pt}\parbox{\columnwidth}{\scriptsize Arguments generated for the verified TOTALCOST optima of \code{f\_random\_hw} (H), \code{f\_both3} (H+S), and \code{v\_nvp}.}

\end{table}

\emph{Other solvers.} We ran the Gecode, Chuffed, and HiGHS back ends
bundled with \mzn{}~\resSetupMiniZinc{}, with one thread and
\resSetupTimeLimit{}\,s (\cref{tab:xsolver}). Gecode and HiGHS proved only
\resXsolverGecodeProven{} and \resXsolverHighsProven{} of the
\resXsolverInstances{} instances, and HiGHS found verified incumbents on
\resXsolverHighsSolved{} in all. Chuffed stopped with an error on all of
them, as it cannot handle a product of variables whose sign is not fixed.
\cpsat{} with one worker proved \resXsolverCpSatOneProven{}, each equal to
the four-worker optimum. So three optima are confirmed by a second solver,
and the verifier is the main independent check.

\begin{table}[!t]
  \caption{Cross-Solver Runs on the Instances of \cref{tab:results}}
  \label{tab:xsolver}
  \centering\footnotesize
  \setlength{\tabcolsep}{3pt}
\begin{tabular*}{\columnwidth}{@{\extracolsep{\fill}}l r c c c @{}}
\toprule
Instance & \cpsat{} (\resSetupThreads{}) & Gecode & HiGHS & Chuffed \\
\midrule
\code{p\_rasta} & 3330 & -- & -- & err. \\
\code{f\_random\_hw} & 3349 & -- & -- & err. \\
\code{d\_random\_hw} & 1857 & -- & 2585 & err. \\
\code{d\_systematic\_sw} & 1846 & -- & \cmark & err. \\
\code{f\_sw3} & 3365 & -- & 5470 & err. \\
\code{f\_both3} & 3406 & -- & 4349 & err. \\
\code{c\_pat} & 3349 & -- & -- & err. \\
\code{pc\_sobel} & 1401 & -- & \cmark & err. \\
\code{v\_nvp} & 5680 & \cmark & 6330 & err. \\
\code{m\_3app} & 5665 & -- & -- & err. \\
\bottomrule
\end{tabular*}

\vspace{2pt}\parbox{\columnwidth}{\scriptsize Second column: the proven optimum of $C_{\mathrm{tot}}$ with \resSetupThreads{} workers. Other columns: one thread, \resSetupTimeLimit{}\,s. \cmark: proven optimal and equal; a number: best verified incumbent at the limit; --: no solution; err.: the solver stopped with an error.}

\end{table}

\emph{Mutation analysis.} We corrupted \resMutSolMutants{} verified optima
on \resMutSolInstances{} instances with \resMutSolOperators{} operators
(\cref{tab:mut-ops}). Each mutant breaks one requirement. The fields that
follow from the change (WCETs, core SILs, FCR use, costs) are repaired, and
every period is set to the ceiling of the mutated deployment's MCR, so the
period check cannot catch a mutant by accident. The exact variant~B of the
model serves as the oracle; it rejects all of them. We also solved the
same instances with \resMutModelMutants{} model mutants, each without one
constraint family (\resMutModelRuns{} runs, \cref{tab:mut-model}); \resMutModelManifested{}
solutions violate the original model, in
\resMutModelFamiliesManifested{} families; the other families never changed
a solution. The model alone is not a sound oracle: it rejects
\resMutModelModelOnly{} correct deployments through the gap of
\cref{sec:model-throughput}.

The current verifier rejected all \resMutSolRejectedCur{} corrupted
solutions and all \resMutModelRejectedCur{} violating mutant solutions,
without a false alarm, and its intended check fired for every corrupted
solution; it also accepts all \resMutRecheckSolutions{} verified solutions
that the experiments stored. The first verifier, which took WCETs, static
orders, activation, and costs from the solution, rejected only
\resMutSolRejectedArt{} and \resMutModelRejectedArt{}. It missed every
mutant that changes only a WCET, the order structure, activation, the
admissible pattern, partitioning capability, or a cost. The generator then
emitted an argument for each of these \resMutGsnEmittedFaultyArt{}
corrupted solutions. Where it rejected a mutant of these families, the
period or deadlock check fired, not a check of the broken constraint.

The generator refused every rejected solution, every catalog mismatch
(\resMutGenCatalogRefused{} of \resMutGenCatalogCases), and every
contradicting fault model (\resMutGenFaultModelRefused{} of
\resMutGenFaultModelCases). On the NVP instance, the inference is ambiguous,
so a given fault model has nothing to contradict
(\resMutGenFaultModelAmbiguous{} cases, not counted). It emitted an argument
for all \resMutGenControls{} uncorrupted solutions. Of \resMutArgMutants{}
mutations of generated arguments in \resMutArgOperators{} kinds, the first
audit caught \resMutArgCaughtAudit; it missed all \resMutArgBadIndex{}
Solutions that cite a record index past the end of the log. The current
audit, which requires every cited record to exist and to have passed,
catches \resMutArgCaughtCur.

\begin{table}[!t]
  \caption{Corrupted Solutions per Operator: Rejected by the Verifier}
  \label{tab:mut-ops}
  \centering\footnotesize
  \setlength{\tabcolsep}{3pt}
\begin{tabular*}{\columnwidth}{@{\extracolsep{\fill}}>{\raggedright\arraybackslash}p{1.75in} c c c@{}}
\toprule
Mutation (what it breaks) & Mutants & First & Current \\
\midrule
a placement relation of the selected pattern & \resMutOpPlaceManifestedArt{} & \resMutOpPlaceRejectedArt{} & \resMutOpPlaceRejectedCur{} \\
a node's SIL below its requirement & \resMutOpSilNodeManifestedArt{} & \resMutOpSilNodeRejectedArt{} & \resMutOpSilNodeRejectedCur{} \\
a core's provisioned SIL below a node's & \resMutOpSilCoreManifestedArt{} & \resMutOpSilCoreRejectedArt{} & \resMutOpSilCoreRejectedCur{} \\
a core SIL above its type's ceiling & \resMutOpCeilingManifestedArt{} & \resMutOpCeilingRejectedArt{} & \resMutOpCeilingRejectedCur{} \\
two SILs on an unpartitioned core & \resMutOpMixManifestedArt{} & \resMutOpMixRejectedArt{} & \resMutOpMixRejectedCur{} \\
a reported period below the MCR & \resMutOpPeriodManifestedArt{} & \resMutOpPeriodRejectedArt{} & \resMutOpPeriodRejectedCur{} \\
partitioning on a type without it & \resMutOpPartForgeManifestedArt{} & \resMutOpPartForgeRejectedArt{} & \resMutOpPartForgeRejectedCur{} \\
a WCET that differs from the table & \resMutOpWcetManifestedArt{} & \resMutOpWcetRejectedArt{} & \resMutOpWcetRejectedCur{} \\
a static order across cores & \resMutOpOrderCrossManifestedArt{} & \resMutOpOrderCrossRejectedArt{} & \resMutOpOrderCrossRejectedCur{} \\
two static orders on one core & \resMutOpOrderSplitManifestedArt{} & \resMutOpOrderSplitRejectedArt{} & \resMutOpOrderSplitRejectedCur{} \\
a cyclic static order & \resMutOpOrderCycleManifestedArt{} & \resMutOpOrderCycleRejectedArt{} & \resMutOpOrderCycleRejectedCur{} \\
a component of the selected pattern inactive & \resMutOpDeactivateManifestedArt{} & \resMutOpDeactivateRejectedArt{} & \resMutOpDeactivateRejectedCur{} \\
an inadmissible pattern & \resMutOpPatternNoneManifestedArt{} & \resMutOpPatternNoneRejectedArt{} & \resMutOpPatternNoneRejectedCur{} \\
a required field missing & \resMutOpDropFieldManifestedArt{} & \resMutOpDropFieldRejectedArt{} & \resMutOpDropFieldRejectedCur{} \\
a misreported cost & \resMutOpCostManifestedArt{} & \resMutOpCostRejectedArt{} & \resMutOpCostRejectedCur{} \\
\midrule
all & \resMutSolManifested{} & \resMutSolRejectedArt{} & \resMutSolRejectedCur{} \\
\bottomrule
\end{tabular*}

\vspace{2pt}\parbox{\columnwidth}{\scriptsize Corrupted copies of verified optima on \resMutSolInstances{} instances; every mutant violates the model (exact variant B as the oracle). First / Current: mutants rejected by the first verifier and by the current one.}

\end{table}

\begin{table}[!t]
  \caption{Model Mutants per Constraint Family}
  \label{tab:mut-model}
  \centering\footnotesize
  \setlength{\tabcolsep}{3pt}
\begin{tabular*}{\columnwidth}{@{\extracolsep{\fill}}l c c c c@{}}
\toprule
Constraint family removed & Runs & Viol. & First & Current \\
\midrule
placement relations, \cref{eq:placement} & 10 & 9 & 9 & 9 \\
node SIL $\ge$ requirement, \cref{eq:sil} & 10 & 10 & 10 & 10 \\
node SIL $\le$ core SIL, \cref{eq:sil} & 10 & 0 & 0 & 0 \\
core SIL $\le$ type ceiling, \cref{eq:sil} & 10 & 7 & 7 & 7 \\
SIL-mixing rule, \cref{eq:rule} & 10 & 10 & 10 & 10 \\
partitioning capability, \cref{eq:sil} & 10 & 9 & 0 & 9 \\
WCET follows binding, \cref{eq:wcet} & 10 & 10 & 8 & 10 \\
channel edges of \cref{eq:pot} & 10 & 7 & 7 & 7 \\
static-order edges of \cref{eq:pot} & 10 & 9 & 9 & 9 \\
wrap edges of \cref{eq:pot} & 10 & 6 & 6 & 6 \\
implied load bound, \cref{eq:load} & 10 & 0 & 0 & 0 \\
activation, \cref{eq:activation} & 10 & 10 & 7 & 10 \\
admissible pattern & 10 & 10 & 1 & 10 \\
one static order per core, \cref{eq:order} & 10 & 0 & 0 & 0 \\
order within one core, \cref{eq:order} & 10 & 9 & 3 & 9 \\
memory bound, \cref{eq:memory} & 10 & 0 & 0 & 0 \\
remote-channel path, \cref{eq:commpath} & 2 & 1 & 1 & 1 \\
period coupling, \cref{eq:coupling} & 1 & 1 & 1 & 1 \\
\midrule
all & \resMutModelRuns{} & \resMutModelManifested{} & \resMutModelRejectedArt{} & \resMutModelRejectedCur{} \\
\bottomrule
\end{tabular*}

\vspace{2pt}\parbox{\columnwidth}{\scriptsize One copy of the model per family, the family disabled; optima of the \resMutSolInstances{} instances (TOTALCOST, 120\,s). Viol.: solutions that violate the original model (exact variant B as the oracle). First / Current: violating solutions rejected by the first and the current verifier.}

\end{table}

\FloatBarrier
\section{Lessons and Negative Results}\label{sec:lessons}

\subsection{Failure Modes During Development}\label{sec:lessons-dev}
The recurring lesson of building \safedse{} is that constraint models and
their tooling fail silently. Almost every failure below produced a
plausible, well-formed output rather than an error, and was caught only by
an independent check or by reading the output instead of the exit code.
This is the reason for the verifier of \cref{sec:verification}.

\emph{Silent under-constraint.} In \mzn{}, \code{=} binds tighter than
\code{\textbackslash/}, so an assignment of a disjunction without
parentheses leaves the variable free whenever the second disjunct holds.
A disjunction placed outside a quantifier instead of inside it once
collapsed the safety design space of all actors to two points: either every
actor took one option or every actor took the other. The solutions looked
reasonable and came from a trivially small space; per-actor pattern
variables with per-actor guards, the root of guarded superposition, came
out of this failure. A static order encoded as an open chain, without the
wrap edge, silently dropped the bound ``period at least the core's load''
and reported a graph with many actors on two cores at full pipeline rate. A
direct edge of a remote channel kept next to its communication path was a
shortcut past the whole transfer delay, which made the bus free. The same
shape reappeared with explicit voters, where a surviving direct edge let
the output bypass the voter (\cref{sec:model-superposition,sec:model-comm}).

\emph{Wrong modeling intuitions.} Communication actors were first placed in
the static order of their core and counted in its load, which gave wrong
periods for the communication benchmarks. A transfer overlaps computation
and only lengthens the period when the bus is the bottleneck. Three
optimization intuitions did not survive measurement. Symmetry breaking was
not the bottleneck, a flattened encoding of the WCET lookup was slower than
the nested one, and tying the activity of communication actors to the
remote flag made solving slower. The standing rule became to measure first,
and to build the deliberately unsound variant first to bound what an
optimization can gain. Pre-unfolded HSDF benchmarks also lose the
provenance of each node's SDF actor, which the copy-order symmetry needs, so
multi-rate SDF input is preferred.

\emph{Wrong results from tooling.} A killed installation left a \mzn{}
binary that existed and crashed, and every solve returned no solution; it
looked like an unsatisfiable model. Instances renamed or dropped from the
build script survived in the output directory and were solved against a
newer model. A cross-solver sweep printed nothing for minutes because it
skipped timeouts silently, which hid that on the pattern instances no
second solver answered at all (\cref{sec:eval-rq4}).

\emph{Fluent but false output.} The worst failures were arguments that
passed every structural check and were false. One was an argument for an actor with
three-version programming that argued it ``by development process alone''
because the catalog lacked the pattern. The other was the possibility of a
wrong fault-model flag yielding claims of coverage the design never had.
Both are
now refusals (\cref{sec:ver-gsn}). Pre-checks were also weaker than they
looked. A diversity check that counted all core types passed on a platform
whose third type was capped below the required SIL. Its first fix fired
too eagerly and rejected an instance whose other patterns were placeable.
A satisfiability pre-check that does not fire is not evidence of
satisfiability, and one that fires too eagerly is worse than none.

\subsection{Defects Found by the Evaluation}\label{sec:lessons-eval}
Running the experiments of this report on the tool found further defects
(\cref{tab:defects}). Each was found by a disagreement between two
independent computations: model and verifier, two oracles, two runs, or a
result and its expected value. Each defect was fixed in the tool, and every
affected experiment was re-run; the results in this report are from the
fixed tool. Some of them changed results that a first
draft of the evaluation had reported. Before the slot-sharing fix, a
checker that shared a slot with a full replica was priced as a replica,
which raised the minimum period of the RASTA pattern instances. Before the
no-promotion fix, the NP flow activated every pattern slot and appeared to
lose heavily against joint exploration. Before the fix of an implied
constraint, tight period bounds made feasible instances infeasible. And an
argument statistic depended on a tie between equal-cost patterns. The last
case is the general lesson of the evaluation. A quantity of a non-unique
optimum (its period, its pattern, an argument derived from it) is not a
property of the instance, so the evaluation reports only proven objective
values, or fixes the tie explicitly.

\begin{table*}[!t]
  \caption{Defects Found While Running the Evaluation, and Their Fixes}
  \label{tab:defects}
  \centering\footnotesize
  \setlength{\tabcolsep}{3pt}
  \begin{tabular*}{\textwidth}{@{\extracolsep{\fill}}l
    >{\raggedright\arraybackslash}p{2.2in}
    >{\raggedright\arraybackslash}p{1.5in}
    >{\raggedright\arraybackslash}p{1.95in}@{}}
    \toprule
    Part & Defect & Found by & Resolution \\
    \midrule
    Model & Implied constraint on heavy node pairs ignored activity, so an
      owner could not share a core with its own inactive slot; the
      front-end's least core count also counted slots &
      tight bounds made instances with a verified design infeasible &
      guarded by the activity of both nodes; core count from always-active
      nodes \\
    Model & Recurring pattern cost counted once per slot of an actor &
      reading the model for the paper & summed over application actors only \\
    Model & No-promotion rule also bound inactive slots, which then had to
      be active & every NP solution activated all slots & rule guarded by
      activity; NP re-run, now equal to joint \\
    Model & A path through an inactive slot can raise a design's least
      period above its MCR & verifier rejected two time-limited incumbents
      by period & documented; exact variants kept as patches and used as a
      cross-check (\cref{sec:model-throughput}) \\
    Front-end & Slots shared by position between patterns, so a checker
      took a replica's WCET and the reverse & a slot whose role and WCET
      type differed & slots shared only by equal WCET type; all RQs re-run \\
    Front-end & Owner type of a replica not found for a multi-rate owner &
      building multi-rate synthetic graphs & lookup in the owner's
      application, ignoring the copy suffix; regression test \\
    Front-end & \code{sil\_offset} read as raising a component above its
      actor's SIL, against the sources & a review of the safety semantics &
      every component at the actor's SIL; the offset may only lower the doer
      (unused by the model) \\
    Verifier & A zero-weight token-free cycle through inactive slots
      reported as deadlock & solver and verifier disagreed on an optimum &
      checks the deployed MSAG, inactive nodes removed, orders spliced \\
    Verifier & Local refined channel lost its direct edge; wrong tokens on a
      remote path & reading the verifier against the model & edges as in
      the model \\
    Verifier & Required the integer period to equal a fractional MCR &
      a correct design failed verification & compares with
      $\lceil\mathrm{MCR}\rceil$ \\
    Verifier & Took WCETs, orders, activation, and costs from the solution &
      mutation analysis (\cref{sec:eval-rq4}) & re-derives every field; six
      new record kinds \\
    Generator & Named the first oracle Karp, though it is Lawler's search &
      reading the oracle & texts corrected \\
    Generator & Audit accepted an undeveloped goal without a reason, and a
      citation past the end of the log & reading the audit; mutation
      analysis & both rejected \\
    Generator & Claimed an attained SIL; cited \code{DIFFERENT} records for
      separate FCRs & a review of the safety semantics & one claim per
      relation, common-cause goals, catalog preconditions instead of SIL \\
    Catalog & Klosterman's SIL~3 took the lower of two table rows; the
      NVP multiplier key matches no role & reading the source; the
      extension study & higher row (lower one as a variant); key harmless,
      as its default equals the stated value \\
    Harness & SL depended on which step-1 optimum the solver returned &
      two runs of the same SL row differed & canonical step 1 (ties by
      period, then core index) \\
    Harness & More than one thread started eight workers; without
      intermediate solutions, a time-limited run reported no incumbent &
      process monitor; a run with no solution at the time limit & worker
      count passed explicitly; intermediate solutions always on \\
    \bottomrule
  \end{tabular*}

  \vspace{2pt}\parbox{\textwidth}{\scriptsize Each fix is a commit of the tool's
  history with a regression test where one applies; the affected
  experiments were re-run afterwards (\cref{app:artifact}).}
\end{table*}

\FloatBarrier
\section{Discussion and Limitations}\label{sec:discussion}
The results hold within the model and under its assumptions. Timing
assumes sound WCETs and self-timed execution in one static order per core,
without preemption or time-triggered tables, so applications that share a
core need a common period. Replicas and checkers always fire, so a detected
fault changes token values, not the schedule. The latency constraint is a
periodic-phase estimate, not a worst-case bound; the verifier simulates the
transient and flags a larger worst case. The fault model leaves out
transient, timing, omission, and interconnect faults, common causes across
FCRs, latent faults, and degraded modes. \safedse{} computes no failure
probabilities, and \code{DIVERSE} means distinct core types, not an
independent design.

The model develops the doer at the actor's SIL, although the sources allow
a lower one~\cite{koopman2016challenges}, and does not price diverse
development. In a model variant (\cref{sec:eval-ext}), the lower doer SIL
made joint optima up to \resModExtDoerMaxSavingPct{}\,\% cheaper, priced
diversity up to \resModExtDiverseMaxExtraPct{}\,\% dearer, and both
together took \resModExtBothTimeRatioMed{} times the median solve time. The
cost profiles are our reading of published estimates; that the profile can
flip the architecture rests on one row of Klosterman's table
(\cref{sec:eval-rq4}). Cross-checks between pattern components are assumed
to use a dedicated link, and transfers are treated as non-safety under the
black-channel principle; both assumptions are stated in every generated
argument, and RQ4 shows that each decides a result.

Because inactive slots keep their channels, the least period of a design
can exceed the MCR of its deployment (\cref{sec:model-throughput}), so the
model can reject a design that meets a bound. This occurred in time-limited
incumbents (synthetic graphs with 32 actors and \code{m\_3app} with at
most four cores) and in one proven safety-last design under the two-token
variant; no reported result is affected. The exact variant~B confirmed
every infeasibility and all but one conclusive RQ2 and RQ4 result, but
solves more slowly. With one thread, Gecode
and HiGHS proved only \resXsolverGecodeProven{} and \resXsolverHighsProven{}
of the \resXsolverInstances{} instances of \cref{tab:results}, and Chuffed
stopped with an error, so the verifier is the main independent check. The
pattern instances come from three small applications, and exact search
reaches its limit at a few tens of HSDF actors (\cref{sec:eval-rq3}).

On the tooling side, the fault-model inference is ambiguous on small
catalogs, such as the NVP case study, where all three fault models
reproduce the admissible patterns; the generator then names none and
cannot list out-of-scope scenarios. The generator emits one argument per
instance without GSN modules or away goals, which becomes large for several
applications. Platforms with multi-core cards solve more slowly than ones
with single-core cards, because value precedence is sound only within a
card. Finally, a per-actor ``not promotable'' attribute would let a designer
forbid promotion for selected actors only, between the global switch and
free promotion.

A faster exact formulation, priced pattern-dependent SILs, verdict tokens as
decisions under a process safety time, a declared interconnect with routed
cross-checks, GSN modules, and industrial case studies are future work.

\FloatBarrier
\section{Conclusion}\label{sec:conclusion}
Whether a safety architecture pattern protects a function and fits its
period depends on where its components run and how they are scheduled.
\safedse{} decides the pattern, its deployment, and the platform in one
constraint model over a fixed superposed graph and proves the optimum within
that model. On the benchmarks, the value of joint exploration lay in
feasibility under a throughput bound: which pattern-first rule works depends
on catalog timing and pricing, and joint exploration needs no such rule.
Verifier checks discharge the argument's claims on placement, SIL
provisioning, isolation, and timing; claims on detection, common causes,
fault reaction, and the development process remain with the engineer.

\FloatBarrier
\appendices
\section{Proof of the Potential Characterization}\label{app:proof}
Fix the binding, the static orders, the pattern selection, the TDMA
allocation, and the buffer sizes. The MSAG is then a fixed digraph
$M=(V,E)$ with integer node times $T_i\ge0$ and integer tokens
$\tok_{ij}\ge0$ on each edge. Self-edges and wrap edges hold one token,
order edges none, and buffer back-edges the buffer sizes. Within one connected
component, all periods are equal by \cref{eq:coupling}, so one $\mu$
applies to every edge of the component, and the argument is per component.

\begin{proposition}[\Cref{prop:potentials}, full form]
A vector $\pot$ with $\pot_j\ge\pot_i+T_i-\tok_{ij}\,\mu$ for all
$(i,j)\in E$ exists iff every directed cycle $C$ of $M$ satisfies
$\sum_{i\in C}T_i\le\mu\sum_{(i,j)\in C}\tok_{ij}$. If $T$, $\tok$, and
$\mu$ are integers, a feasible $\pot$ exists iff an integral one does.
\end{proposition}

\begin{IEEEproof}
(Only if.) Let $C=(i_0,i_1,\dots,i_{m-1},i_0)$. Summing the $m$ inequalities
of the edges of $C$ gives
$\sum_k\pot_{i_{k+1}}\ge\sum_k\pot_{i_k}+\sum_kT_{i_k}-\mu\sum_k
\tok_{i_ki_{k+1}}$, indices modulo $m$. Every node of $C$ appears once as
head and once as tail, so the two potential sums are equal, and
$0\ge\sum_{i\in C}T_i-\mu\sum_{(i,j)\in C}\tok_{ij}$.

(If.) Set $w_{ij}=T_i-\mu\,\tok_{ij}$. By assumption, every cycle has
weight $w(C)\le0$. Add a node $s$ with an edge $(s,j)$ of weight zero to
every $j\in V$; $s$ has no incoming edge, so no new cycle arises. For every
$j$, let $\ell_j$ be the largest weight of an $s$--$j$ walk. It is finite:
removing a cycle from a walk does not decrease its weight, as
$w(C)\le0$, so the maximum is attained by a simple path, and there are
finitely many. Further, $\ell_j\ge0$ because of the edge $(s,j)$. For each
edge $(i,j)\in E$, a best $s$--$i$ walk followed by $(i,j)$ is an $s$--$j$
walk of weight $\ell_i+w_{ij}$, hence $\ell_j\ge\ell_i+w_{ij}$, i.e.,
$\pot=\ell$ satisfies every inequality. If $T$, $\tok$, and $\mu$ are
integers, every $w_{ij}$ and every $\ell_j$ is an integer. Shifting all
potentials by a constant keeps them feasible, so the model's anchor of one
potential at zero loses nothing, and as $0\le\ell_j\le\sum_iT_i$, the
bounded domain of the potentials in the model holds $\ell$.
\end{IEEEproof}

\emph{Corollary 1 (period).} For fixed $T$ and $\tok$, the feasible $\mu$
are exactly those with
\[
  \mu\ge\max_{C:\,\tok(C)>0}T(C)/\tok(C)=\mathrm{MCR},
\]
provided every token-free cycle has $T(C)=0$. The least real $\mu$ is the
MCR, and the least integer $\mu$ is $\lceil\mathrm{MCR}\rceil$, which the
model reports.

\emph{Corollary 2 (deadlock).} A token-free cycle with $T(C)>0$ makes the
system infeasible for every $\mu$. A deadlocking design is thus excluded by
the same inequalities, without separate deadlock machinery.

\emph{Corollary 3 (zero-weight cycles).} Token-free cycles with $T(C)=0$
satisfy the condition with $0\le0$ and do not constrain $\mu$. In
\safedse{}, they arise only through neutralized nodes. A verifier that
checks the graph with all nodes would flag them as deadlock; the verifier
therefore checks the deployed graph (\cref{sec:ver-verifier}).

\emph{Corollary 4 (relation to the deployed graph).} Remove the inactive
nodes and splice each static order around them. Every cycle of this
deployed graph maps to an MSAG cycle with the same total time and tokens. A
spliced order edge $a\to b$ becomes $a\to s\to b$ with $T_s=0$ and no
tokens on either edge, and a spliced wrap edge keeps its token. Hence the
model's least $\mu$ for a design is never below the deployed graph's MCR,
and neutralization cannot make the period optimistic. The converse does not
hold in general. A channel of an unselected pattern into an inactive slot,
followed by the slot's order edge, is an MSAG path without a counterpart in
the deployed graph. A cycle through it can raise the least $\mu$ above
the deployed MCR. Removing the channel edges at inactive endpoints gives
exactly the deployed MCR on the designs where this was observed, which is
what exact variant~B achieves inside the model (\cref{sec:model-throughput}).

\emph{Corollary 5 (other objectives).} When $\mu$ is not the objective, the
inequalities give only $\mathrm{MCR}\le\mu\le$ the period bound. The
reported $\mu$ equals $\lceil\mathrm{MCR}\rceil$ only if the search assigns
it its least value; the evaluation enforces this by the period step
(\cref{sec:eval-setup}), and the verifier checks it.

\section{Instances and Parameters}\label{app:instances}
All pattern instances except Sobel and \code{m\_3app} use the RASTA speech
front end with the required SILs of \cref{fig:example-a}: \code{frontEnd}
and \code{compJah} at SIL~3, \code{backEnd} at SIL~2, and all other actors
at SIL~1. \Cref{tab:instances} lists their other inputs. Every time value is
an integer in the benchmark's unit. On the card platforms, the WCET of an
application actor is its nominal execution time in the benchmark file times
the cycle factor of the core mode (\cref{tab:platforms}), rounded. The WCET
of a checker or a voter is \resSetupCheckerScale\,\% of its owner's nominal
time times the cycle factor, and a full replica has its owner's WCET. These
WCET tables are placeholders generated by the tool, not measurements of a
real platform. Sobel on Rosvall's platform uses the published WCETs. The
synthetic graphs of RQ3 take the generator's execution times per processor
type. Every node has the same baseline development cost
(\cref{tab:profiles}); card and core prices are illustrative
(\cref{tab:platforms}); the recurring cost units of each pattern are in
\cref{tab:catalog-full}. All costs are in one unitless cost unit.

\begin{table*}[!t]
  \caption{Inputs of the Pattern Instances}
  \label{tab:instances}
  \centering\footnotesize
  \setlength{\tabcolsep}{3pt}
  \begin{tabular*}{\textwidth}{@{\extracolsep{\fill}}l l l l l
    >{\raggedright\arraybackslash}p{2.1in}@{}}
    \toprule
    Instance & Application(s) & Platform file & Fault model & Bus &
    Remarks \\
    \midrule
    \code{p\_rasta} & RASTA & \code{mixed.xml} & H & ideal & certified
      partitioning on the R-card \\
    \code{f\_random\_hw} & RASTA & \code{mixed\_noiso.xml} & H & ideal & -- \\
    \code{d\_random\_hw} & RASTA & \code{mixed\_noiso.xml} & H & ideal &
      SIL~3 actors lowered to SIL~2 \\
    \code{d\_systematic\_sw} & RASTA & \code{mixed\_noiso.xml} & S & ideal &
      SIL~3 actors lowered to SIL~2 \\
    \code{f\_sw3} & RASTA & \code{mixed\_3type.xml} & S & ideal & -- \\
    \code{f\_both3} & RASTA & \code{mixed\_3type.xml} & H+S & ideal & -- \\
    \code{c\_pat} & RASTA & \code{mixed\_noiso.xml} & H & TDMA & -- \\
    \code{pc\_sobel} & Sobel & Rosvall's & H & TDMA & \code{get\_pixel} at
      SIL~3, others SIL~1; period bound loosened for the cost runs \\
    \code{v\_nvp} & RASTA & \code{mixed\_4type.xml} & H+S & ideal &
      separate catalog with the NVP pattern \\
    \code{m\_3app} & Sobel, SUSAN, RASTA & \code{mixed\_noiso.xml} & H & ideal
      & Sobel at SIL~2, SUSAN at SIL~0; one period per application \\
    \bottomrule
  \end{tabular*}

  \vspace{2pt}\parbox{\textwidth}{\scriptsize All instances use the catalog
  \code{patterns.yaml} unless noted and the cost profile of Myklebust et al.
  RQ4 also uses the running example without patterns on
  \code{mixed\_noiso.xml} (cost profiles, price sweep) and on
  \code{mixed\_costly.xml} (\cref{fig:profcons}). The variants of RQ2 replace
  only the catalog (\code{tok2}) or the WCET tables (\code{chk10},
  \code{chk50}).}
\end{table*}

\section{Artifact Map}\label{app:artifact}
Every number of this report comes from CSV files that the experiment
scripts wrote on the machine of \cref{sec:eval-setup}. The script
\code{make\_macros.py} turns them into the macros used in the text and in
the data files of the plots; \code{make\_report.py} writes the
tables and plot data that only this report uses. A row enters a number only
if it is optimal (or proven infeasible) and passed the verifier.
\Cref{tab:artifact} lists, per experiment, the driver, its CSV files, and
where the results appear. The tool itself (model, front-end, verifier,
generator, regression tests, and data) is in the artifact repository
(\artifacturl). The experiment scripts, input variants, generated synthetic
graphs, CSV files, and the generated tables and plot data are in the
experiments repository (\experimentsurl).

\begin{table*}[!t]
  \caption{Experiments, Drivers, Data, and Where They Appear}
  \label{tab:artifact}
  \centering\footnotesize
  \setlength{\tabcolsep}{3pt}
  \begin{tabular*}{\textwidth}{@{\extracolsep{\fill}}l
    >{\raggedright\arraybackslash}p{1.5in}
    >{\raggedright\arraybackslash}p{1.75in}
    >{\raggedright\arraybackslash}p{1.6in}@{}}
    \toprule
    Experiment & Driver & CSV (\code{experiments/results/}) & Used in \\
    \midrule
    RQ1 benchmarks & \code{harness.py rq1} & \code{e1\_parity} &
      \cref{tab:rq1-hw} \\
    RQ1 \desyde{} & \code{desyde\_cmp.py}, \code{desyde/build.sh} &
      \code{e7\_desyde}, \code{e7\_safedse} & \cref{tab:desyde} \\
    RQ2 flows & \code{harness.py rq2}, \code{rq2tight}, \code{rq2thr},
      \code{rq2var}, \code{rq2reps}, \code{pftable} &
      \code{e2\_joint}, \code{e2\_tight}, \code{e2\_thr}, \code{e2\_var},
      \code{e2\_reps}, \code{e2\_pfchoice} &
      \cref{fig:joint,fig:worse,fig:thr},
      \cref{tab:rq2-nobound,tab:rq2-all,tab:rq2-thr,tab:pfchoice} \\
    RQ2 certificates & \code{certify\_pf.py} & \code{e2\_cert} &
      \cref{sec:eval-rq2} \\
    RQ2 exactness & \code{exact\_model.py}, \code{exact\_check.py} &
      \code{e2\_exact}, \code{e2\_exact\_sl} & \cref{sec:eval-rq2} \\
    Three applications, promotion & \code{breadth.py multi}, \code{patcons} &
      \code{e5\_multi}, \code{e5\_multi\_exact}, \code{e5\_multi\_exact\_long},
      \code{e4\_patcons} & \cref{tab:multi,tab:patcons} \\
    Model extensions & \code{modelext\_build.py}, \code{modelext.py} & \code{e8\_modelext} &
      \cref{tab:modelext,tab:modelext-flows} \\
    RQ3 & \code{gen\_synthetic.py}, \code{harness.py rq3},
      \code{breadth.py flat} & \code{e3\_scal}, \code{e3\_flat} &
      \cref{fig:scalability,fig:flat}, \cref{tab:scal} \\
    RQ4 sensitivity & \code{harness.py rq4}, \code{breadth.py price},
      \code{comm} & \code{e4\_sens}, \code{e4\_price}, \code{e4\_comm} &
      \cref{fig:sensitivity,fig:profcons}, \cref{tab:price,tab:comm} \\
    RQ4 arguments & \code{gsn\_stats.py} & \code{e4\_gsn},
      \code{gsn/*.json} & \cref{tab:gsn}, \cref{fig:gsn,fig:gsn-compjah} \\
    Pipeline, solvers & \code{breadth.py pipe}, \code{xsolver} &
      \code{e4\_pipe}, \code{e4\_xsolver} & \cref{tab:pipe,tab:xsolver} \\
    Mutation analysis & \code{mutation.py} & \code{e6\_mut},
      \code{e6\_oracle}, \code{e6\_recheck} &
      \cref{tab:mut-ops,tab:mut-model} \\
    \bottomrule
  \end{tabular*}

  \vspace{2pt}\parbox{\textwidth}{\scriptsize Drivers are in
  \code{experiments/}; each is resumable and writes one CSV row per run.
  The catalog, scenario, tactic, profile, and platform tables are generated
  from the tool's data files; the GSN excerpts by
  \code{make\_gsn\_figure.py} from the stored arguments.}
\end{table*}

\FloatBarrier
\section*{Acknowledgment}
The authors used the generative AI tools Claude (Anthropic), DeepSeek, and
Gemini (Google) to test and refactor the software and to improve the text of
the paper and of this report. The authors reviewed all results of
this use and take full responsibility for the content.

\bibliographystyle{IEEEtran}
\bibliography{IEEEabrv,refs,refs-report}

\end{document}